\documentclass{amsart}
\usepackage[T1]{fontenc}
\usepackage{array}
\usepackage{booktabs}
\usepackage{multirow}
\usepackage{amsmath}
\usepackage{amssymb}
\usepackage{graphicx}

\makeatletter

\providecommand{\tabularnewline}{\\}

\usepackage[foot]{amsaddr}

\newtheorem{theorem}{Theorem}[section]
\newtheorem{definition}[theorem]{Definition}
\newtheorem{remark}[theorem]{Remark}

\newtheorem{lemma}[theorem]{Lemma}
\newtheorem{proposition}[theorem]{Proposition}
\newtheorem{assumption}[theorem]{Assumption}
\newtheorem{example}[theorem]{Example}

\newtheorem{aplemma}{Lemma}[section]

\usepackage[all]{xy}

\usepackage{balance}

\usepackage{tikz}
\usetikzlibrary{positioning, arrows, matrix}

\usepackage{psfrag}

\usepackage{dsfont}

\usepackage{url}

\mathchardef\mhyphen="2D

\usepackage{verbatim}

\usepackage[absolute]{textpos}

\usepackage{color}

\makeatother

\begin{document}
\title[Characterization of Failures in Multi-Hop Networked Control]{A Probabilistic Characterization of Random and Malicious Communication Failures in Multi-Hop Networked Control}

\thanks{This work was supported  in part by the JST CREST Grant  No.\ JPMJCR15K3 and by JSPS under Grant-in-Aid for Scientific Research Grant No.~15H04020.}

\author{A. Cetinkaya}
\email{ahmet@sc.dis.titech.ac.jp}

\author{H. Ishii}
\email{ishii@c.titech.ac.jp}
\address[A. Cetinkaya, H. Ishii]{Department of Computer Science, Tokyo Institute of Technology, Yokohama, 226-8502, Japan.}

\author{T. Hayakawa}
\address[T. Hayakawa]{Department of Systems and Control Engineering, Tokyo Institute of Technology, Tokyo 152-8552, Japan.}
\email{hayakawa@sc.e.titech.ac.jp}

\begin{abstract} The control problem of a linear discrete-time dynamical
system over a multi-hop network is explored. The network is assumed
to be subject to packet drops by malicious and nonmalicious nodes
as well as random and malicious data corruption issues. We utilize
asymptotic tail-probability bounds of transmission failure ratios
to characterize the links and paths of a network as well as the network
itself. This probabilistic characterization allows us to take into
account multiple failures that depend on each other, and coordinated
malicious attacks on the network. We obtain a sufficient condition
for the stability of the networked control system by utilizing our
probabilistic approach. We then demonstrate the efficacy of our results
in different scenarios concerning transmission failures on a multi-hop
network. \end{abstract} 


\maketitle

\section{Introduction}

Networked control systems incorporate communication networks to facilitate
the exchange of measurement and control data between the plant and
the controller \cite{hespanha2007}. Recently, multi-hop networks
have been utilized in networked control operations \cite{gupta2009data,alur2011compositional,smarra2012optimal,d2013faultTAC,d2016resilient,smarra2017}.
A multi-hop network such as a wireless ad hoc network consists of
a number of nodes that are connected with a number of communication
links. Ensuring orderly operation of a multi-hop networked control
system can be challenging, as packets may sometimes fail to be transmitted
at different parts of the network due to various reasons. 

One of the reasons for transmission failures in a multi-hop network
is channel noise in the communication links, which may corrupt the
contents of state and control input data packets. The occurrence of
data corruption in a communication network may be modeled using random
processes \cite{khayam2003markov}. In addition to channel noise,
network congestion may also cause packet transmission failures. Routers
may drop some packets to mitigate congestion \cite{floyd1993random}. 

Furthermore, it has become apparent that malicious attacks may also
hamper transmissions in a networked control system. For instance,
jamming attacks by an adversary may interfere with the communication
on links and effectively prevent transmission of packets. This issue
was investigated in several works from the viewpoints of wireless
communications \cite{xu2005feasibility,pelechrinis2011}, as well
as control and game theory \cite{amin2009,de2015inputtran,liu2014stochastic,li2015jamming,IFACde2016networked}.
Transmissions of state and control input information between the plant
and the controller may also fail due to malicious activities of routers.
Malicious routers may intentionally drop some of the packets coming
from and/or headed to certain nodes of the network \cite{d2016resilient, awerbuch2008odsbr, mahmoud2014security}.
The detection of such routing attacks can be challenging especially
when the malicious nodes act normal for certain periods of time (see
grayhole attacks in \cite{jhaveri2012attacks}). Furthermore, networks
may face both malicious routing and random packet losses due to link
errors \cite{mizrak2009detecting,shu2015privacy}. Understanding the
effects of malicious attacks on networks is important from the viewpoint
of cyber security of networked control systems \cite{cardenas2008research, wholejournal2015}. 

Our goal in this paper is to explore the effects of random and malicious
transmission failures in a general multi-hop communication network
and develop a network characterization to be used in the analysis
of networked control systems. The key problem here is to characterize
the failures for the overall multi-hop network in a nonconservative
way while still taking into account mutually-dependent packet failures
and coordinated malicious attacks on the network. 

In the literature, researchers have proposed different characterizations
of packet failures in a networked control system. Specifically, \cite{gupta2009data}
explored control over a network with multiple links that introduce
random packet drops. The results obtained in \cite{gupta2009data}
utilize packet drop probabilities on the edges that constitute cut-sets
of the network graph. This approach is also utilized in network characterization
by \cite{dana2006capacity}. Furthermore, \cite{alur2011compositional}
discussed almost sure networked control system stability, and  \cite{quevedo2013state}
studied networked state estimation problem. Recently, \cite{smarra2017}
investigated mean square stability and robustness under delays and
packet losses through a Markov jump linear system framework. Prior
to these works, random packet losses have been studied for the simpler
single-hop case. There, Bernoulli processes as well as time-homogeneous
and time-inhomogeneous Markov chains are utilized for modeling purposes.
In particular, \cite{sinopoli2004kalman} investigated the single-hop
networked Kalman filtering problem under Bernoulli-type packet losses.
Furthermore, for single-hop networked stabilization problems, packet
losses have been modeled with Bernoulli processes in \cite{ishii2009,Lemmon:2011:ASS:1967701.1967744}
and finite-state time-homogeneous Markov chains in \cite{gupta2009,okano2014}.
In \cite{ahmetcdc2015,cetinkaya2015output,cetinkaya-tac}, where both
random packet losses and malicious attacks in a single-hop networked
control problem is considered, we modeled random packet losses through
time-inhomogeneous Markov chains. In addition, stability as well as
robust and optimal control problems concerned with general time-inhomogeneous
Markov jump linear systems with polytopic uncertainties in transition
probability matrices are investigated in \cite{lun2016cdc,lun2017ifac,lun2017cdc},
where the obtained results are applicable to multi-hop networked control
systems with random packet losses.

In this paper, we consider a network with not only random failures
but also malicious activities on nodes or links. As a result, failures
may not always be modeled by Markov processes and packet losses may
not always have well-defined probabilities. In \cite{d2016resilient},
researchers explored a multi-hop network model with malicious nodes.
There, the fault detection and isolation problem is explored for the
case where the nodes induce delay in transmissions. Here, we consider
a different setup and a different modeling approach. In particular,
we do not investigate a detection problem, and we do not model the
effect of delays. In this paper, we would like to characterize the
overall packet exchange failures on a network between the plant and
the controller by using the properties of the paths of the network
and the communication links on those paths. We also note that our
network characterization provides a high level model and it is tailored
to be utilized for stability analysis in networked control as in \cite{gupta2009data}.
Instead of specifying underlying physical channel models and routing
protocols, we take a probabilistic approach to characterize transmission
failure ratios on the links, paths of a network, as well as the network
itself.

Our approach for modeling the overall packet failures in a network
is built upon tail-probability bounds for the binary-valued processes
that describe the occurrences of failures on the network. Specifically,
each link on a multi-hop network is described through an asymptotic
tail-probability bound of the transmission failure ratio of that link.
This tail-probability-based approach is different from the typical
random packet loss modeling approach of assigning probabilities to
failures. Our approach can capture failures that occur due to both
malicious and nonmalicious reasons. In fact, we utilized the tail-probability-based
approach in \cite{cetinkaya-tac} to study the combined effects of
malicious-attacks following the discrete-time version of the attack
model in \cite{de2015inputtran,IFACde2016networked} and random packet
losses modeled as Markov chains. Through this modeling approach, \cite{cetinkaya-tac}
provides a Lyapunov-based stability analysis, which is further enhanced
in \cite{cetinkayatactoappear2018} by using linear programming techniques.
Differently from \cite{cetinkaya-tac}, we show in this paper that
when tail-probability bounds for the links are available, then we
can obtain tail-probability bounds describing the overall failures
on individual paths of a network. Then those bounds are used for deriving
tail-probability bounds describing the overall failures on the network
itself. Using our proposed characterizations, we obtain a probabilistic
upper-bound for the average number of packet exchange failures between
the plant and the controller, which we use in almost sure stability
analysis of a discrete-time linear networked control system. 

In the multi-hop setting, the location of failures and whether multiple
failures depend on each other or not critically affect the quality
of communication and hence the stabilization performance of the controller.
Especially, the situation becomes more serious when the network is
targeted by a number of adversaries that launch coordinated attacks
on different locations in the communication network. Our tail-probability
bound approach can handle such worst-case situations in a unified
manner. In addition to our investigation on those situations, we also
explore the case where one or more paths/links are known to be associated
with random failures and the corresponding indicator processes are
mutually independent. For such cases, we show that tighter results
can be obtained. This is done by exploiting certain properties of
the hidden Markov models that characterize the random failures. 

The organization of the rest of the paper is as follows. In Section~\ref{sec:Control-over-Multi-hop},
we explain the networked control problem over a multi-hop network.
We then present our multi-hop network characterization in Sections~\ref{sec:Random-and-Malicious}
and \ref{Sec:FailuresOnPaths}. Specifically, in Section~\ref{sec:Random-and-Malicious},
we give a characterization of the overall transmission failures in
a multi-hop network based on the failures on individual paths of that
network. Furthermore, in Section~\ref{Sec:FailuresOnPaths}, we investigate
failures on the paths that occur due to data-corruption and packet-dropping
issues at nodes and communication links. We demonstrate the utility
of our results in Section~\ref{sec:Illustrative}, and conclude the
paper in Section~\ref{sec:Conclusion}.

We note that the conference version of this paper appeared in \cite{ahmetnecsys2016}.
In this paper, we provide more detailed discussions throughout the
paper. Furthermore, we present new results in Sections~\ref{sec:Random-and-Malicious},
\ref{Sec:FailuresOnPaths} and new examples in Section~\ref{sec:Illustrative}. 

\begin{table}[t]
\caption{Notation concerning the links and the paths of a graph, binary-valued
failure indicator processes concerning those links and paths, and
the classes of failure indicator processes. }
\label{TableForNotation} 

\renewcommand{\arraystretch}{1.0} 

{\centering \fontsize{7}{11.52}\selectfont 

\begin{tabular}{cl}
\multirow{1}{*}{$\mathcal{P}_{i}$} & $i$th path on a graph representing a communication network\tabularnewline
$\mathcal{P}_{i,j}$ & $j$th communication link on path $\mathcal{P}_{i}$\tabularnewline
$l_{\mathcal{P}_{i}}$ & Binary-valued process indicating overall transmission failures on
path $\mathcal{P}_{i}$\tabularnewline
$l_{\mathcal{P}_{i}}^{\mathcal{P}_{i,j}}$ & Binary-valued process indicating transmission failures on $j$th link
of path $\mathcal{P}_{i}$\tabularnewline
$\Lambda_{\rho}$ & General class of failure indicator processes\tabularnewline
$\Gamma_{p_{0},p_{1}}$ & Class of indicator processes that characterize random failures\tabularnewline
$\Pi_{\kappa,w}$ & Class of failure indicator processes that characterize malicious attacks \tabularnewline
$\theta^{l}$ & Markov chain of a hidden Markov model with the output process $l\in\Gamma_{p_{0},p_{1}}$ \tabularnewline
$h^{l}$ & Output function of a hidden Markov model with the output process $l\in\Gamma_{p_{0},p_{1}}$ \tabularnewline
$f(\mathcal{P})$ & First link of path $\mathcal{P}$\tabularnewline
$\mathcal{R}(\mathcal{P})$  & Subpath obtained from path $\mathcal{P}$ by removing its first link\tabularnewline
\end{tabular}

} 
\end{table}

In this paper, we use $\mathbb{N}_{0}$ and $\mathbb{N}$ to denote
nonnegative and positive integers, respectively. The notation $\mathrm{\mathbb{P}}[\cdot]$
denotes the probability on a probability space $(\Omega,\mathcal{F},\mathbb{P})$
with filtration $\{\mathcal{F}_{t}\}_{t\in\mathbb{N}_{0}}$. For binary
numbers, the notation $\vee$ represents the or-operation; moreover,
$\wedge$ represents the and-operation. Furthermore, Table~\ref{TableForNotation}
provides notation concerning graphs, binary-valued processes, and
classes of such processes. 

\section{Control over Multi-hop Networks}

\label{sec:Control-over-Multi-hop}

In this section, we investigate the control of a linear plant over
multi-hop networks depicted in Figure~\ref{Flo:operation}. On these
networks, the plant and the controller exchange state measurement
and control input packets. The transmissions are not subject to delay;
however, there may be failures in packet exchange attempts between
the plant and the controller. 

We describe the linear dynamics of the plant by 
\begin{align}
x(t+1) & =Ax(t)+Bu(t),\quad x(0)=x_{0},\quad t\in\mathbb{N}_{0},\label{eq:system}
\end{align}
 where $A\in\mathbb{R}^{n\times n}$ and $B\in\mathbb{R}^{n\times m}$
respectively denote the state and input matrices; furthermore, $x(t)\in\mathbb{R}^{n}$
and $u(t)\in\mathbb{R}^{m}$ are the state and the control input vectors,
respectively. 

\begin{figure}
\centering  \includegraphics[width=0.7\columnwidth]{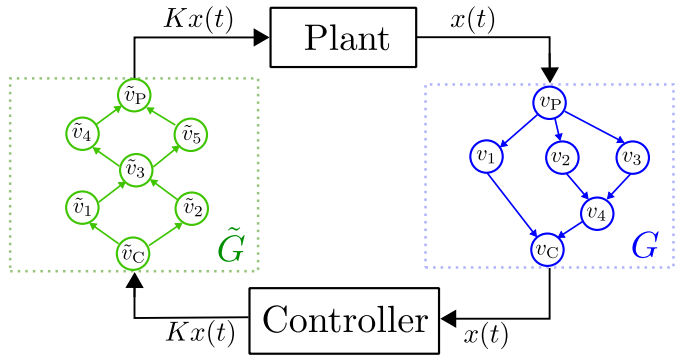}
\vskip -7pt

\caption{\fontsize{9}{9}\selectfont{Multi-hop networked control system}}
 \label{Flo:operation}
\end{figure}

The plant and the controller attempt to exchange state measurement
and control input packets at each time instant $t$. Packet exchanges
are attempted over multi-hop networks $G$ and $\tilde{G}$ as shown
in Figure~\ref{Flo:operation}. The detailed models of these networks
will be given in Section~\ref{subsec:MultihopModel}. We denote success
or failure states of packet exchange attempts by using the binary-valued
process $\{l(t)\in\{0,1\}\}_{t\in\mathbb{N}_{0}}$. In the case of
a successful packet exchange ($l(t)=0$), the plant transmits the
state measurement to the controller; the controller uses the received
state information together with a linear static feedback control law
to compute the control input, which is then sent back to the plant.
The transmitted control input is applied at the plant side. At certain
time instants, either the state packet or the control input packet
cannot be delivered due to network issues such as packet drops, jamming
attacks, and other communication errors. In such cases packet exchange
attempts fail ($l(t)=1$), and the control input at the plant side
is set to $0$, which is one of the common approaches in the literature
(see \cite{hespanha2007} and the references therein). 

Under this characterization, the control input $u(t)$ applied at
the plant side is given by 
\begin{align}
u(t) & \triangleq\left(1-l(t)\right)Kx(t),\quad t\in\mathbb{N}_{0},\label{eq:control-input}
\end{align}
where $K\in\mathbb{R}^{m\times n}$ represents the feedback gain. 

Although we consider a static state-feedback control setup here, the
techniques that we develop in this paper can also be used in conjunction
with other control approaches. In particular, the predictive control
approach of \cite{quevedo2011input} and the event-triggered output-feedback
control approach from our earlier work \cite{cetinkaya2015output}
can be studied within the context of multi-hop networked control by
using the network characterizations that we develop here. 

We assume that the information packets between the plant and the controller
propagate with no delay, although there may be transmission failures
due to:

1) packet drops by malicious nodes to prevent communication and/or
by nonmalicious nodes to avoid congestion; 

2) data corruption on communication links because of random channel
errors and/or malicious jamming attacks. 

We now introduce a class of processes that is useful in describing
packet failure indicators in the paper. 

\begin{definition} [$\Lambda_{\rho}$] \label{DefinitionLambda} Given
a scalar $\rho\in[0,1]$ we define the class of binary-valued processes
$\Lambda_{\rho}$ by 
\begin{align*}
\Lambda_{\rho}\triangleq\big\{ l\colon l(t)\in\{0,1\},t\in\mathbb{N}_{0};\sum_{k=1}^{\infty}\mathbb{P}[\sum_{t=0}^{k-1}l(t)>\rho k]<\infty\big\}.
\end{align*}

\end{definition}

By this definition, we have $\Lambda_{\rho_{1}}\subseteq\Lambda_{\rho_{2}}$
for $\rho_{1}\leq\rho_{2}$. Furthermore, any binary-valued process
$l$ satisfies $l\in\Lambda_{1}$. 

In \cite{cetinkaya-tac}, we explored a problem similar to the one
that we discuss in this paper. There, we considered a single direct
communication channel between the plant and the controller. To describe
the packet losses on this channel, we proposed a probabilistic characterization
that is based on the following assumption. 

\begin{assumption} \label{MainAssumption} For the packet exchange
failure indicator $l(\cdot)$, we have $l\in\Lambda_{\rho}$ with
$\rho\in[0,1]$. 

\end{assumption} 

This assumption allows us to characterize a range of scenarios in
a unified manner through the scalar $\rho\in[0,1]$. For instance,
the case where all packet exchange attempts fail ($l(t)=1$, $t\in\mathbb{N}_{0}$)
can be described by setting $\rho=1$. Moreover, in the case where
all packet exchange attempts are successful ($l(t)=0$, $t\in\mathbb{N}_{0}$),
then $l\in\Lambda_{0}$. In addition to these two extreme cases, as
we illustrate throughout the paper, Assumption~\ref{MainAssumption}
can also be used to describe random failures and malicious attacks. 

When the packet exchange failures in a networked control system satisfy
Assumption~\ref{MainAssumption}, then the scalar $\rho\in[0,1]$
also represents a bound on the asymptotic packet exchange failure
ratio. Specifically, it follows from Borel-Cantelli Lemma that $\limsup_{k\to\infty}\frac{1}{k}\sum_{t=0}^{k-1}l(t)\leq\rho,$
almost surely. When this inequality holds with a small $\rho$, it
means that the packet exchanges fail statistically rarely. We showed
in \cite{cetinkaya-tac} that the plant \eqref{eq:system} can be
stabilized over a network, if $\rho$ is sufficiently small. 

The stability analysis method developed in \cite{cetinkaya-tac} allows
us to obtain the following result, which presents sufficient stability
conditions for the closed-loop networked control system \eqref{eq:system},
\eqref{eq:control-input}.

\begin{theorem} \label{Stability-Theorem} Consider the dynamical
system \eqref{eq:system}, \eqref{eq:control-input}. Suppose Assumption~\ref{MainAssumption}
holds with scalar $\rho\in[0,1]$. If there exist a positive-definite
matrix $P\in\mathbb{R}^{n\times n}$ and scalars $\beta\in(0,1),$
$\varphi\in[1,\infty)$ such that 
\begin{align}
 & \left(A+BK\right)^{\mathrm{T}}P\left(A+BK\right)-\beta P\leq0,\label{eq:betacond}\\
 & A^{\mathrm{T}}PA-\varphi P\leq0,\label{eq:varphicond}\\
 & (1-\rho)\ln\beta+\rho\ln\varphi<0,\label{eq:betaandvarphicond}
\end{align}
then the zero solution $x(t)\equiv0$ of the closed-loop system \eqref{eq:system},
\eqref{eq:control-input} is asymptotically stable almost surely.
\end{theorem}

In Theorem~\ref{Stability-Theorem}, the conditions \eqref{eq:betacond}
and \eqref{eq:varphicond} characterize stability and instability
of the closed-loop and the open-loop dynamics through scalars $\beta$
and $\varphi$. These scalars also appear in condition \eqref{eq:betaandvarphicond}.
When $\rho$ (and hence the packet exchange failure ratio) is sufficiently
small so that \eqref{eq:betaandvarphicond} is satisfied, then we
have almost sure asymptotic stability, which implies $\mathbb{P}[\lim_{t\to\infty}\|x(t)\|=0]=1.$ 

In Sections~\ref{sec:Random-and-Malicious} and \ref{Sec:FailuresOnPaths}
below, we will present some key methods for obtaining $\rho$ of a
\emph{multi-hop networked control system} \eqref{eq:system}, \eqref{eq:control-input}
to facilitate its stability analysis with Theorem~\ref{Stability-Theorem}.
Specifically, we will consider the setting where the state measurement
and control input packets are attempted to be transmitted over \emph{multi-hop
networks }instead of a single direct channel, which was considered
in \cite{cetinkaya-tac}\emph{.} We will use assumptions similar to
Assumption~\ref{MainAssumption} to characterize packet transmission
failures on the paths between the plant and the controller as well
as the individual links on those paths. We will then show that the
packet exchange failures (represented by $l(\cdot)$) in the overall
networked control system satisfy Assumption~\ref{MainAssumption}
with a scalar $\rho\in[0,1]$ that depends on the network structure
as well as failure models of links. 

To facilitate the analysis in the following sections, we now define
two more classes of binary-valued processes that are distinct from
$\Lambda_{\rho}$. Our goal is to characterize more specific models
for random and malicious failures. 

In our earlier work \cite{cetinkaya-tac}, we utilized time-inhomogeneous
Markov chains for characterizing random failures in a single communication
channel. When we consider a multi-hop network composed of a number
of channels that face random transmission failures, we are required
to introduce a different characterization. This is because the overall
failures in the network depends on the failures on each individual
channel. Hence, the overall failures cannot always be described as
a Markov chain. In order to describe random failures in multi-hop
networks, we utilize time-inhomogeneous \emph{hidden Markov models}
(see \cite{anderson1999wiener,cappe2009inference,vidyasagar2014hidden}).
The binary-valued process $\{l(t)\in\{0,1\}\}_{t\in\mathbb{N}_{0}}$
has a time-inhomogeneous hidden Markov model if 
\begin{align}
l(t) & =h^{l}(\theta^{l}(t)),\quad t\in\mathbb{N}_{0},\label{eq:hiddenmarkov}
\end{align}
where $h^{l}\colon\Theta^{l}\to\{0,1\}$ is a binary-valued function
on a set $\Theta^{l}$ of finite number of elements, and moreover,
$\{\theta^{l}(t)\in\Theta^{l}\}_{t\in\mathbb{N}_{0}}$ is an $\mathcal{F}_{t}$-adapted,
finite-state, and time-inhomogeneous Markov chain with initial distributions
$\vartheta_{q}^{l}\in[0,1]$, $q\in\Theta^{l}$, and time-varying
transition probability functions $p_{q,r}^{l}\colon\mathbb{N}_{0}\to[0,1]$,
$q,r\in\Theta^{l}$, satisfying 
\begin{align*}
\begin{array}{c}
\mathbb{P}[\theta^{l}(0)=q]=\vartheta_{q}^{l},\\
\mathbb{P}[\theta^{l}(t+1)=r|\theta^{l}(t)=q]=p_{q,r}^{l}(t),\quad t\in\mathbb{N}_{0}.
\end{array}
\end{align*}
The process $\{l(t)\in\{0,1\}\}_{t\in\mathbb{N}_{0}}$ is also called
the output process of a hidden Markov model. Notice that $\{l(t)\in\{0,1\}\}_{t\in\mathbb{N}_{0}}$
depends on the Markov chain $\{\theta^{l}(t)\in\Theta^{l}\}_{t\in\mathbb{N}_{0}}$
through function $h^{l}$, but it is not necessarily a Markov chain
itself. Specifically, in certain cases, we may have $\mathbb{P}[l(t+1)=1|\mathcal{F}_{t}]\neq\mathbb{P}[l(t+1)=1|l(t)]$,
which shows that Markov property (see \cite{billingsley1986}) does
not hold. This is for example the case where $l(\cdot)$ represents
the failures on a Gilbert-Elliott channel (see \cite{sadeghi2008,ellis2014}).
We also note that hidden Markov models naturally arise in the description
of multi-hop networks. For instance, $l(\cdot)$ may be the failure
indicator of a path with multiple links. Even if the failures on individual
links may satisfy the Markov property, the overall failure indicator
$l(\cdot)$ does not satisfy it due to dependence on the failure/success
states of all individual links. In such cases, $l(\cdot)$ follows
a hidden Markov model, where the Markov chain $\{\theta^{l}(t)\in\Theta^{l}\}_{t\in\mathbb{N}_{0}}$
represents the combined states of all individual links. 

For a given binary-valued output process $\{l(t)\in\{0,1\}\}_{t\in\mathbb{N}_{0}}$
associated with a time-inhomogeneous hidden Markov model, let $\Theta_{0}^{l},\Theta_{1}^{l}\subset\Theta^{l}$
be given by 
\begin{align*}
\Theta_{0}^{l} & \triangleq\{r\in\Theta^{l}\colon h^{l}(r)=0\},\quad\Theta_{1}^{l}\triangleq\{r\in\Theta^{l}\colon h^{l}(r)=1\}.
\end{align*}
In the next definition we introduce a class of binary-valued output
processes associated with time-inhomogeneous hidden Markov models. 

\begin{definition} [$\Gamma_{p_0,p_1}$] \label{Definition-Hidden-Markov}
Given scalars $p_{0},p_{1}\in[0,1]$ we define the class $\Gamma_{p_{0},p_{1}}$
of binary-valued output processes of time-inhomogeneous hidden Markov
models by
\begin{align*}
\Gamma_{p_{0},p_{1}} & \triangleq\big\{ l\colon\sum_{r\in\Theta_{0}^{l}}p_{q,r}^{l}(t)\leq p_{0},\,q\in\Theta^{l},\,t\in\mathbb{N}_{0};\sum_{r\in\Theta_{1}^{l}}p_{q,r}^{l}(t)\leq p_{1},\,q\in\Theta^{l},\,t\in\mathbb{N}_{0}\big\}.
\end{align*}

\end{definition}

The advantage of the class $\Gamma_{p_{0},p_{1}}$is that by utilizing
the values $p_{0}$ and $p_{1}$ associated with different hidden
Markov models, we can characterize the combined effects of those models
even if detailed information on the Markov chain $\theta^{l}$ and
the output function $h^{l}$ of those models are not available. 

The following result establishes the relation between the classes
$\Gamma_{p_{0},p_{1}}$ and $\Lambda_{\rho}$. 

\begin{proposition} \label{PropositionGammaLambdaRelation} Consider
the binary-valued output process $\{l(t)\in\{0,1\}\}_{t\in\mathbb{N}_{0}}$
of a time-inhomogeneous hidden Markov model. If $l\in\Gamma_{p_{0},p_{1}}$
with $p_{1}\in(0,1)$, then we have $l\in\Lambda_{\rho}$ for any
$\rho\in(p_{1},1]$. \end{proposition}

Before, we give the proof of Proposition~\ref{PropositionGammaLambdaRelation},
we first present a technical result that provides upper bounds on
the tail probabilities of sums involving a binary-valued output process
associated with a time-inhomogeneous hidden Markov model.

\begin{lemma}\label{KeyMarkovLemma} Let $\{\xi(t)\in\Xi\}_{t\in\mathbb{N}_{0}}$
be a finite-state time-inhomogeneous Markov chain with transition
probabilities $p_{q,r}\colon\mathbb{N}_{0}\to[0,1]$, $q,r\in\Xi$,
and let $\Xi_{1}\subset\Xi$ be given by $\Xi_{1}\triangleq\{r\in\Xi\colon h(r)=1\}$,
where $h\colon\Xi\to\{0,1\}$ is a binary-valued function. Furthermore,
let $\{\chi(t)\in\{0,1\}\}_{t\in\mathbb{N}_{0}}$ be a binary-valued
process that is independent of $\{\xi(t)\in\Xi\}_{t\in\mathbb{N}_{0}}$.
Assume 
\begin{align}
 & \sum_{r\in\Xi_{1}}p_{q,r}(t)\leq\tilde{p},\quad q\in\Xi,\quad t\in\mathbb{N}_{0},\label{eq:xicond}\\
 & \sum_{k=1}^{\infty}\mathbb{P}[\sum_{t=0}^{k-1}\chi(t)>\tilde{w}k]<\infty,\label{eq:chicond}
\end{align}
where $\tilde{p}\in(0,1)$, $\tilde{w}\in(0,1]$. We then have for
$\rho\in(\tilde{p}\tilde{w},\tilde{w})$, 
\begin{align}
\mathbb{P}[\sum_{t=0}^{k-1}h(\xi(t))\chi(t)>\rho k] & \leq\psi_{k},\quad k\in\mathbb{N},\label{eq:keylemmaresult1}
\end{align}
 where $\psi_{k}\triangleq\tilde{\sigma}_{k}+\phi^{-\rho k+1}\frac{\left((\phi-1)\tilde{p}+1\right)^{\tilde{w}k}-1}{(\phi-1)\tilde{p}}$
, $\phi\triangleq\frac{\frac{\rho}{\tilde{w}}(1-\tilde{p})}{\tilde{p}(1-\frac{\rho}{\tilde{w}})}$,
$\tilde{\sigma}_{k}\triangleq\mathbb{P}[\sum_{t=0}^{k-1}\chi(t)>\tilde{w}k],k\in\mathbb{N}$.
Moreover, $\sum_{k=1}^{\infty}\psi_{k}<\infty.$ \end{lemma}

Lemma~\ref{KeyMarkovLemma} is an essential tool for dealing with
different failure scenarios specific to multi-hop networks, and it
generalizes a result for the fully observable Markov chains from our
previous work \cite{cetinkaya-tac}. In particular, in the case where
$\Xi=\{0,1\}$ and $h(r)=r$, $r\in\Xi$, Lemma~\ref{KeyMarkovLemma}
recovers Lemma~A.1 of \cite{cetinkaya-tac}. The proof Lemma~\ref{KeyMarkovLemma}
is given in the Appendix.

We are now ready to prove Proposition~\ref{PropositionGammaLambdaRelation}.

\begin{proof}[Proof of Proposition~\ref{PropositionGammaLambdaRelation}]
Notice that $l\in\Lambda_{\rho}$ for $\rho=1$, since $l(\cdot)$
is binary-valued. For the case $\rho\in(p_{1},1)$, we show $l\in\Lambda_{\rho}$
by employing Lemma~\ref{KeyMarkovLemma}. Specifically, let $\tilde{p}=p_{1}$,
$\tilde{w}=1$, and define the processes $\{\xi(t)\in\{0,1\}\}_{t\in\mathbb{N}_{0}}$
and $\{\chi(t)\in\{0,1\}\}_{t\in\mathbb{N}_{0}}$ with $\xi(t)=\theta_{1}^{l}(t)$
and $\chi(t)=1$, $t\in\mathbb{N}_{0}$. Since the conditions in \eqref{eq:xicond}
and \eqref{eq:chicond} are satisfied with $\tilde{p}=p_{1}$, $h=h^{l},$
$\Xi=\Theta^{l}$, and $\Xi_{1}=\Theta_{1}^{l}$, it follows from
Lemma~\ref{KeyMarkovLemma} that 
\begin{align*}
\sum_{k=1}^{\infty}\mathbb{P}[\sum_{t=0}^{k-1}l(t)>\rho k] & =\sum_{k=1}^{\infty}\mathbb{P}[\sum_{t=0}^{k-1}h^{l}(\theta_{1}^{l}(t))>\rho k]=\sum_{k=1}^{\infty}\mathbb{P}[\sum_{t=0}^{k-1}h(\xi(t))\chi(t)>\rho k]<\infty,
\end{align*}
 which completes the proof. \end{proof}

Next, we introduce a class for binary-valued processes that we employ
in characterizing the timing of malicious attacks. 

\begin{definition} [$\Pi_{\kappa,w}$] \label{DefinitionPi} Given
scalars $\kappa\geq0$, $w\in[0,1]$ we define the class of binary-valued
processes $\Pi_{\kappa,w}$ by 
\begin{align*}
\Pi_{\kappa,w} & \triangleq\big\{ l\colon l(t)\in\{0,1\},t\in\mathbb{N}_{0};\,\,\mathbb{P}\big[\sum_{t=0}^{k-1}l(t)\leq\kappa+wk\big]=1,\,\,k\in\mathbb{N}\big\}.
\end{align*}

\end{definition}

The characterization for the class $\Pi_{\kappa,w}$ is based on a
discrete-time version of the malicious attack model used by \cite{de2015inputtran,IFACde2016networked}.
In this model, occurrences of malicious attacks are described by a
process $l(\cdot)$ such that $\mathbb{P}\big[\sum_{t=0}^{k-1}l(t)\leq\kappa+wk\big]=1$,
$k\in\mathbb{N}$, where $\kappa\ge0$ represents an upper-bound for
a number of initial attacks, and $w\in(0,1)$ represents a bound on
the average attack rate. Lemma~2.3 in \cite{cetinkaya-tac} shows
that if $l\in\Pi_{\kappa,w}$ with $w\in(0,1)$, then $l\in\Lambda_{\rho}$
for any $\rho\in(w,1]$.

Note that the malicious attack characterization through the class
$\Pi_{\kappa,w}$ does not require the process $l\in\Pi_{\kappa,w}$
to follow a particular distribution at each time. This is the key
difference of the class $\Pi_{\kappa,w}$ from the class $\Gamma_{p_{0},p_{1}}$
that represents random failures. 

There are several ways an attacker can strategize when to cause transmission
failures. For instance, game-theoretic \cite{liu2014stochastic,li2015jamming}
and optimization-based methods can be used by the attacker to decide
the timing of attacks. An important property of the class $\Pi_{\kappa,w}$
is that it characterizes attacks by their maximum average attack rate
but not by the specific timing strategy they follow. Thus by using
$\Pi_{\kappa,w}$, we can capture the uncertainty in the generation
of attacks, which may follow a deterministic strategy, or may involve
randomness. Interestingly, an attacker can also make use of system
dynamics as well as past/present state information to decide the timing
of attacks to cause more damage to the system. In the following example,
we discuss such an attack scenario. 

\begin{example} \label{StateDependentOptimizationStrategy} Consider
the scenario where the plant's communication node $v_{\mathrm{P}}$
in Figure~\ref{Flo:operation} is compromised by an attacker. The
attacker is assumed to utilize the knowledge of system dynamics and
the state information for deciding whether to transmit packets or
not. We represent the attacker's actions with a binary-valued process
$\{l_{\mathrm{A}}(t)\in\{0,1\}\}_{t\in\mathbb{N}_{0}}$, where $l_{\mathrm{A}}(t)=0$
indicates that the attacker transmits the state packet at time $t$
to nodes $v_{1},v_{2},v_{3}$ and $l_{\mathrm{A}}(t)=1$ indicates
no transmission. The attacker decides the values of the binary-valued
process $l_{\mathrm{A}}(t)$ according to 
\begin{align}
l_{\mathrm{A}}(t) & =a_{t}^{*},\label{eq:attackstr1}
\end{align}
 where $(a_{t}^{*},a_{t+1}^{*},\cdots,a_{t+N-1}^{*})\in\{0,1\}^{N}$
is a solution to the optimization problem
\begin{align}
 & \begin{array}{c}
\underset{(a_{t},\cdots,a_{t+N-1})}{\mathrm{maximize}}\,\,\,\,\sum_{i=1}^{N}\|\hat{x}(t+i)\|^{2}\\
\quad\mathrm{subject\,\,to}\,\,\sum_{i=0}^{t-1}l_{\mathrm{A}}(i)+\sum_{i=t}^{t+j-1}a_{i}\leq\kappa_{\mathrm{A}}+w_{\mathrm{A}}(t+j),\quad j\in\{1,\ldots,N\},
\end{array}\label{eq:attackstr2}
\end{align}
with $\hat{x}(t+i)=\big(A+(1-a_{t+i-1})BK\big)\big(A+(1-a_{t+i-2})BK\big)\cdots\big(A+(1-a_{t})BK\big)x(t)$,
$\kappa_{\mathrm{A}}\geq0$, $w_{\mathrm{A}}\in(0,1)$, and $N\in\mathbb{N}$.
Here, $\hat{x}(t+i)$ denotes the attacker's prediction of the future
state at time $t+i$. In this strategy, the attacker decides to attack
($l_{\mathrm{A}}(t)=1$) or not ($l_{\mathrm{A}}(t)=0$) at time $t$,
based on the solution of the optimization problem where the goal is
to maximize the sum of squared predicted future state norms over the
interval $[t+1,t+N]$, while keeping the attack rate below a certain
value $w_{\mathrm{A}}\in(0,1)$. The positive integer $N\in\mathbb{N}$
is the horizon in the optimization problem, and it can be large if
the attacker has sufficient computational resources. The optimization
problem in \eqref{eq:attackstr2} is solved at each time step and
the updated state information is used by the attacker for decision
making. 

The process $\{l_{\mathrm{A}}(t)\in\{0,1\}\}_{t\in\mathbb{N}_{0}}$
under the attack strategy \eqref{eq:attackstr1}, \eqref{eq:attackstr2}
is not a Markov process, since the value of $l_{\mathrm{A}}(t)$ depends
on not just the action $l_{\mathrm{A}}(t-1)$ but all previous actions
$l_{\mathrm{A}}(0),l_{\mathrm{A}}(1),\ldots,l_{\mathrm{A}}(t-1)$.
Moreover, $l_{\mathrm{A}}(t)$ depends also on the state value $x(t)$. 

Notice that under this attack strategy, we have $l_{\mathrm{A}}\in\Pi_{\kappa_{\mathrm{A}},w_{\mathrm{A}}}$,
and hence, $l_{\mathrm{A}}\in\Lambda_{\rho}$ for any $\rho\in(w_{\mathrm{A}},1]$.
This illustrates the generality of the class $\Lambda_{\rho}$ characterized
in Definition~\ref{DefinitionLambda}. Not only the output processes
of hidden Markov models from the class $\Gamma$, but also the non-Markovian,
state-dependent, and optimization-based attacks from the class $\Pi$
belong to the class $\Lambda_{\rho}$ for certain values of $\rho$.
In Section~\ref{subsec:State-dependent-Attacks-by}, we further illustrate
the effects of such attacks. \end{example} 

As we discussed above, a process $l$ that belongs to either of the
classes $\Pi_{\kappa,w}$ and $\Gamma_{p_{0},p_{1}}$ also belongs
to the class $\Lambda_{\rho}$ for a suitable value of $\rho$. This
observation suggests us that both random failures and malicious attacks
can be characterized by utilizing the class $\Lambda_{\rho}$. We
note also that there are cases where a process may belong to $\Lambda_{\rho}$,
even though it does not belong to $\Pi_{\kappa,w}$ or $\Gamma_{p_{0},p_{1}}$.
The following example discusses such a case. 

\begin{example}In our recent work \cite{ahmetsiam2018}, we investigated
the effects of channel noise and jamming attacks on a wireless communication
channel. There, we considered a physical channel model to determine
the wireless transmission failure probabilities. In particular, we
used the binary-valued process $\{l_{\mathrm{J}}(t)\in\{0,1\}\}_{t\in\mathbb{N}_{0}}$
to indicate the transmission failures and the process $\{v_{\mathrm{J}}(t)\in[0,\infty)\}_{t\in\mathbb{N}_{0}}$
to denote the jamming interference power. These two processes are
related to each other through the equality 
\begin{align}
\mathbb{P}[l_{\mathrm{J}}(t)=1|v_{\mathrm{J}}(t)=v^{*}] & =p(v^{*}),\quad v^{*}\geq0,\label{eq:lcharacterize1}
\end{align}
 where $p\colon[0,\infty)\to[0,1]$ is a function determined by the
properties (such as the constant channel noise power) associated with
the underlying channel. We showed that $l_{\mathrm{J}}\in\Lambda_{\rho}$
holds for certain $\rho$ values. Specifically, if there exist $\kappa_{\mathrm{J}}\geq0,\overline{v}_{\mathrm{J}}\geq0$
such that $\sum_{i=0}^{t-1}v(i)\leq\kappa_{\mathrm{J}}+\overline{v}_{\mathrm{J}}t$
for $t\in\mathbb{N}$, then we have $l_{\mathrm{J}}\in\Lambda_{\rho}$
for $\rho\in(\hat{p}(\overline{v}),1]$, where $\hat{p}\colon[0,\infty)\to[0,1]$
is a concave function that upper-bounds $p$. Notice that in this
case, $l_{\mathrm{J}}$ depends on the uncertain power process $v_{\mathrm{J}}$.
Moreover, $l_{\mathrm{J}}$ does not necessarily belong to the class
$\Pi$ or $\Gamma$, even though it belongs to $\Lambda_{\rho}$ for
$\rho\in(\hat{p}(\overline{v}),1]$. 

\end{example}

As we make it clear in the following sections, the class $\Lambda_{\rho}$
has useful properties. For instance, if two processes $l_{1}$ and
$l_{2}$ belong to classes $\Lambda_{\rho_{1}}$ and $\Lambda_{\rho_{2}}$,
then the processes $l_{\wedge}$ and $l_{\vee}$ defined by setting
$l_{\wedge}(t)=l_{1}(t)\wedge l_{2}(t)$ and $l_{\vee}=l_{1}(t)\vee l_{2}(t)$
belong to classes $\Lambda_{\rho_{\wedge}}$ and $\Lambda_{\rho_{\vee}}$,
respectively, where $\rho_{\wedge}$ and $\rho_{\vee}$ depend on
$\rho_{1}$ and $\rho_{2}$. Such properties enable us to model the
failures on both the links and the paths of a network by using processes
that belong to the class $\Lambda_{\rho}$ for certain values of $\rho$. 

\section{Random and Malicious Packet Failures in Multi-Hop Networks}

\label{sec:Random-and-Malicious}

In this section, we present a framework for modeling random and malicious
packet transmission failures in the multi-hop networks that are used
for exchanging state and control input packets between the plant and
the controller. 

\subsection{Multi-Hop Network Model}

\label{subsec:MultihopModel}

We follow the approach of \cite{dana2006capacity} and represent the
networks between the plant and the controller by using directed acyclic
graphs. To model the network, over which the \emph{state packets}
are transmitted from the plant to the controller, we consider the
directed acyclic graph $G\triangleq(V,E)$, where $V$ denotes the
set of nodes, and $E\subset V\times V$ denotes the set of edges.
Here the nodes and the edges correspond respectively to communication
devices and links. We represent the nodes at the plant and the controller
with $v_{\mathrm{P}}\in V$ and $v_{\mathrm{C}}\in V$, respectively.
A \emph{path} $\mathcal{P}$ from a node $v_{1}\in V$ to another
node $v_{h}\in V$ is identified as a sequence of nonrepeating edges
$\mathcal{P}=\big((v_{1},v_{2}),(v_{2},v_{3}),\ldots,(v_{h-1},v_{h})\big)$.
We write $|\mathcal{P}|$ to denote the number of edges on the path
$\mathcal{P}$. 

Similarly, the network used for transmission of the \emph{control
input packets} from the controller to the plant is represented by
graph $\tilde{G}\triangleq(\tilde{V},\tilde{E})$ with plant and controller
nodes $\tilde{v}_{\mathrm{P}}\in\tilde{V}$ and $\tilde{v}_{\mathrm{C}}\in\tilde{V}$.
In practice, the same physical network may be used for transmission
of both the state and the control input packets. For those cases,
the nodes in $V$ and $\tilde{V}$ would correspond to the same physical
devices. 

We assume that there exists at least one directed path from node $v_{\mathrm{P}}$
to node $v_{\mathrm{C}}$ in $G$, and at least one directed path
from node $\tilde{v}_{\mathrm{C}}$ to node $\tilde{v}_{\mathrm{P}}$
in $\tilde{G}$. This ensures that the underlying communication network
topology allows packet transmissions between the plant and the controller.
Note that each path represents a possible transmission route. When
multiple routes are utilized, the same packet is attempted to be transmitted
on all those routes. In the ideal case where packet drops and data
corruption do not occur, packets can be delivered in either one of
these routes at all times. In this paper, we are interested in the
nonideal case, where at certain times, transmissions on these routes
may fail. 

\begin{example} We show an example of $G$ and $\tilde{G}$ in Fig.~\ref{Flo:operation},
where the nodes corresponding to the plant and the controller are
not directly connected, but state and control input packets can still
be transmitted with the help of the intermediate nodes $v_{1},\ldots,v_{4}$
in $G$, and $\tilde{v}_{1},\ldots,\tilde{v}_{5}$ in $\tilde{G}$.
\end{example}

Intermediate nodes in networks forward data packets that they receive
from their incoming edges to their outgoing edges. Depending on the
communication protocol, the forwarding method may differ. For instance,
in the \emph{broadcast} method, intermediate nodes forward all data
packets that they receive from their incoming edges to all the nodes
that they are connected with their outgoing edges. On the other hand,
it may also be the case that intermediate nodes follow a specific
routing scheme, where a packet coming from a certain incoming edge
is forwarded through a certain outgoing edge \cite{medhi2010network}. 

A packet exchange between the plant and the controller may fail if
the state or the control input packets are dropped or get corrupted.
Here note that corrupted data packets are allowed to be transmitted
over intermediate nodes, but they are detected and discarded at the
plant/controller nodes. Error-detecting codes can be used for this
purpose. Note also that if the controller only receives corrupted
versions of a state packet, the control input is not computed. \vskip 2pt

In the following sections, we present some key results for the analysis
of packet failures on network $G$, which are directly applicable
for analyzing $\tilde{G}$. In particular, we characterize failures
on $G$ in terms of the failures on different paths between the plant
and the controller. Then we present a set of results that relate data
corruption and packet dropping issues of nodes and links to the failures
on each individual path of $G$. These results enable us to obtain
$\rho\in[0,1]$ in Assumption~\ref{MainAssumption}, which is essential
for analyzing the system in Fig.~\ref{Flo:operation} with Theorem~\ref{Stability-Theorem}.
We emphasize that the central problem here is to find $\rho$ for
the overall multi-hop network in Fig.~\ref{Flo:operation} in a nonconservative
way while still taking into account mutually-dependent packet failures
and coordinated attacks on the network. 

\subsection{Packet Transmission Failures on Networks}

\label{subsec:ong}

We use the binary-valued process $\{l_{G}(t)\in\{0,1\}\}_{t\in\mathbb{N}_{0}}$
to indicate transmission failures on $G$. Specifically, $l_{G}(t)=0$
means that the state packet $x(t)$ sent from the plant node $v_{\mathrm{P}}$
is successfully received at the controller node $v_{\mathrm{C}}$.
On the other hand, $l_{G}(t)=1$ indicates a failure, that is, the
controller does not receive the state $x(t)$. 

Let $c\in\mathbb{N}$ denote the number of paths on graph $G$ from
the node $v_{\mathrm{P}}$ to the node $v_{\mathrm{C}}$, and let
$\mathcal{P}_{i},i\in\{1,\ldots,c\}$, denote these paths. In the
example network $G$ in Fig.~\ref{Flo:operation}, there are $c=3$
paths 
\begin{align}
 & \mathcal{P}_{1}=\big((v_{\mathrm{P}},v_{1}),(v_{1},v_{\mathrm{C}})\big),\,\,\mathcal{P}_{2}=\big((v_{\mathrm{P}},v_{2}),(v_{2},v_{4}),(v_{4},v_{\mathrm{C}})\big),\nonumber \\
 & \mathcal{P}_{3}=\big((v_{\mathrm{P}},v_{3}),(v_{3},v_{4}),(v_{4},v_{\mathrm{C}})\big).\label{eq:example-paths}
\end{align}
 Note that different paths may include the same link. Hence, when
packet transmission is attempted on multiple paths, a link that is
shared on those paths may be used multiple times. For instance, $(v_{4},v_{\mathrm{C}})$
is on both $\mathcal{P}_{2}$ and $\mathcal{P}_{3}$. Hence, $(v_{4},v_{\mathrm{C}})$
may be utilized twice to forward the packets coming from $v_{2}$
and $v_{3}$. On the other hand, the framework that we describe below
also allows modeling the case where one of the packets is dropped
at node $v_{4}$ and not transmitted further. Furthermore, the packet
drop can be random or malicious. 

We use $\{l_{\mathcal{P}_{i}}(t)\in\{0,1\}\}_{t\in\mathbb{N}_{0}}$
to indicate whether the state packet $x(t)$ is successfully transmitted
to the controller through path $\mathcal{P}_{i}$ or not. Specifically,
$l_{\mathcal{P}_{i}}(t)=0$ represents a successful transmission.
On the other hand $l_{\mathcal{P}_{i}}(t)=1$ may indicate that the
path is not utilized for transmission due to the particular routing
scheme, or it may indicate a failure. Failures occur if packets get
dropped on the path or if they get corrupted. 

Thus, in network $G$, the transmission of the state packet $x(t)$
from node $v_{\mathrm{P}}$ to node $v_{\mathrm{C}}$ results in failure
if $l_{\mathcal{P}_{i}}(t)=1$ for all paths $\mathcal{P}_{i},i\in\{1,\ldots,c\}$.
Therefore, $l_{G}(\cdot)$ is given by 
\begin{align}
l_{G}(t) & =l_{\mathcal{P}_{1}}(t)\wedge l_{\mathcal{P}_{2}}(t)\wedge\cdots\wedge l_{\mathcal{P}_{c}}(t).\label{eq:lg-multiple-paths}
\end{align}
The following result presents a probabilistic and asymptotic bound
for the packet transmission failure ratio of $G$ as a function of
the bounds for the individual paths $\mathcal{P}_{i}$.

\begin{proposition} \label{Proposition-and-operation} Assume for
each path $\mathcal{P}_{i}$ that we have $l_{\mathcal{P}_{i}}\in\Lambda_{\rho_{\mathcal{P}_{i}}}$,
where $\rho_{\mathcal{P}_{i}}\in[0,1]$. Then $l_{G}\in\Lambda_{\rho_{G}}$
with $\rho_{G}\triangleq\min_{i\in\{1,\ldots,c\}}\rho_{\mathcal{P}_{i}}$. 

\end{proposition} 

\begin{proof}First let $i^{*}\in\arg\min_{i\in\{1,\ldots,c\}}\rho_{\mathcal{P}_{i}}$.
It follows that $\rho_{G}=\min_{i\in\{1,\ldots,c\}}\rho_{\mathcal{P}_{i}}=\rho_{\mathcal{P}_{i^{*}}}$.
Now, by \eqref{eq:lg-multiple-paths}, 
\begin{align*}
\sum_{t=0}^{k-1}l_{G}(t) & \leq\sum_{t=0}^{k-1}l_{\mathcal{P}_{i}}(t),\quad i\in\{1,\ldots,c\},
\end{align*}
 and hence $\sum_{t=0}^{k-1}l_{G}(t)\leq\sum_{t=0}^{k-1}l_{\mathcal{P}_{i^{*}}}(t)$.
Therefore, 
\begin{align*}
\mathbb{P}[\sum_{t=0}^{k-1}l_{G}(t)>\rho_{G}k] & \leq\mathbb{P}[\sum_{t=0}^{k-1}l_{\mathcal{P}_{i^{*}}}(t)>\rho_{\mathcal{P}_{i^{*}}}k].
\end{align*}
 The result then follows, since $l_{\mathcal{P}_{i^{*}}}\in\Lambda_{\rho_{\mathcal{P}_{i^{*}}}}$.
\end{proof}

The scalars $\rho_{\mathcal{P}_{i}}$, $i\in\{1,\ldots,c\}$, in Proposition~\ref{Proposition-and-operation}
represent bounds for asymptotic packet failure ratios on different
paths of network $G$. Proposition~\ref{Proposition-and-operation}
indicates that the minimum of these scalars is also a bound for the
packet failure ratio of the whole network. Observe that if $\rho_{\mathcal{P}_{i}}=0$
for some path $\mathcal{P}_{i}$, then we have $\rho_{G}=0$, which
means that the state can be securely and reliably transmitted to the
controller at all time instants. This is because the transmission
on path $\mathcal{P}_{i}$ never fails. If, on the other hand, $\rho_{\mathcal{P}_{i}}=1$
for all paths $\mathcal{P}_{i}$, then $\rho_{G}=1$, indicating all
packet transmission attempts fail, almost surely. 

Note that in Proposition~\ref{Proposition-and-operation}, we do
not assume that $\{l_{\mathcal{P}_{i}}(t)\in\{0,1\}\}_{t\in\mathbb{N}_{0}}$
are mutually-independent processes. This allows us to deal with the
scenarios where transmission failures on different paths may depend
on each other. In particular, we can consider coordinated packet dropout
attacks of several malicious routers on different paths. For instance,
two malicious routers $v_{2}$ and $v_{3}$ in Fig.~\ref{Flo:operation}
may skip forwarding packets at the same time. Then transmissions on
paths $\mathcal{P}_{2}$ and $\mathcal{P}_{3}$ given in \eqref{eq:example-paths}
would both fail. Similarly, Proposition~\ref{Proposition-and-operation}
is also useful when links on different paths are attacked at the same
time by coordinated jamming attackers. 

Another scenario that can be explored through Proposition~\ref{Proposition-and-operation}
is related to packet drops by nonmalicious routers to prevent congestion
\cite{floyd1993random}. For example, a nonmalicious router $v_{4}$
in Fig.~\ref{Flo:operation} may choose to forward only one of the
packets coming from $v_{2}$ and $v_{3}$. Then, $l_{\mathcal{P}_{2}}(\cdot)$
and $l_{\mathcal{P}_{3}}(\cdot)$ would be dependent processes. In
particular, if there are no other failures in the network, then we
have $l_{\mathcal{P}_{2}}(t)=1-l_{\mathcal{P}_{3}}(t)$. 

We remark that by utilizing additional properties of the indicator
processes $l_{\mathcal{P}_{i}}(\cdot)$ for paths, we can obtain a
better asymptotic failure bound $\rho_{G}$ than the one provided
in Proposition~\ref{Proposition-and-operation}. In particular, if
one or more paths are known to be associated with random failures
and the corresponding indicator processes are mutually independent,
we can obtain tighter results than Proposition~\ref{Proposition-and-operation}.
To this end, we first present the following result on the properties
of a process that is obtained by using $\wedge$ operation on the
output processes of two mutually-independent hidden Markov models. 

\begin{theorem} \label{Theorem-HiddenMarkov} Consider the binary-valued
output processes $\{l^{(1)}\in\{0,1\}\}_{t\in\mathbb{N}_{0}}$ and
$\{l^{(2)}\in\{0,1\}\}_{t\in\mathbb{N}_{0}}$ of hidden Markov models
such that $l^{(1)}\in\Gamma_{p_{0}^{(1)},p_{1}^{(1)}}$, $l^{(2)}\in\Gamma_{p_{0}^{(2)},p_{1}^{(2)}}$
with $p_{0}^{(1)},p_{1}^{(1)},p_{0}^{(2)},p_{1}^{(2)}\in[0,1]$. Suppose
that the Markov chains $\{\theta^{l^{(1)}}(t)\in\Theta^{l^{(1)}}\}_{t\in\mathbb{N}_{0}}$
and $\{\theta^{l^{(2)}}(t)\in\Theta^{l^{(2)}}\}_{t\in\mathbb{N}_{0}}$
associated with the processes $l^{(1)}$ and $l^{(2)}$ are mutually
independent. Let $\tilde{p}_{0}\triangleq\min\{p_{0}^{(1)}+p_{1}^{(1)}p_{0}^{(2)},p_{0}^{(2)}+p_{1}^{(2)}p_{0}^{(1)},1\}$
and $\tilde{p}_{1}\triangleq p_{1}^{(1)}p_{1}^{(2)}$. Then the process
$\{\tilde{l}(t)\in\{0,1\}\}_{t\in\mathbb{N}_{0}}$ defined by 
\begin{align}
\tilde{l}(t) & =l^{(1)}(t)\wedge l^{(2)}(t),\quad t\in\mathbb{N}_{0},\label{eq:ltildedef}
\end{align}
 is the output process of a time-inhomogeneous hidden Markov model,
and moreover, $\tilde{l}\in\Gamma_{\tilde{p}_{0},\tilde{p}_{1}}$. 

\end{theorem}

\begin{proof}Let $\Theta^{\tilde{l}}\triangleq\{(q^{(1)},q^{(2)})\colon q^{(1)}\in\Theta^{l^{(1)}},q^{(2)}\in\Theta^{l^{(2)}}\}$
and define the bivariate process $\{\theta^{\tilde{l}}(t)\in\Theta^{\tilde{l}}\}_{t\in\mathbb{N}_{0}}$
by 
\begin{align*}
\theta^{\tilde{l}}(t) & =(\theta^{l^{(1)}}(t),\theta^{l^{(2)}}(t)),\quad t\in\mathbb{N}_{0}.
\end{align*}
 It follows that $\{\theta^{\tilde{l}}(t)\in\Theta^{\tilde{l}}\}_{t\in\mathbb{N}_{0}}$
is a time-inhomogeneous Markov chain with initial distribution $\vartheta_{(q^{(1)},q^{(2)})}^{\tilde{l}}=\vartheta_{q^{(1)}}^{l^{(1)}}\vartheta_{q^{(2)}}^{l^{(2)}}$,
$q^{(1)}\in\Theta^{l^{(1)}},q^{(2)}\in\Theta^{l^{(2)}}$, and time-varying
transition probabilities $p_{(q^{(1)},q^{(2)}),(r^{(1)},r^{(2)})}^{\tilde{l}}(t)=p_{q^{(1)},r^{(1)}}^{l^{(1)}}(t)p_{q^{(2)},r^{(2)}}^{l^{(2)}}(t)$,
$t\in\mathbb{N}_{0}$. Here, for $j\in\{1,2\}$, $\vartheta^{l^{(j)}}$
and $p^{l^{(j)}}(\cdot)$ respectively denote the initial distribution
and the transition probability function for the Markov chain $\{\theta^{l^{(j)}}(t)\in\Theta^{l^{(j)}}\}_{t\in\mathbb{N}_{0}}$
associated with the output process $\{l^{(j)}\in\{0,1\}\}_{t\in\mathbb{N}_{0}}$.
Furthermore, it follows from \eqref{eq:ltildedef} that 
\begin{align*}
\tilde{l}(t) & =l^{(1)}(t)\wedge l^{(2)}(t)=l^{(1)}(t)l^{(2)}(t)=h^{l^{(1)}}(\theta^{l^{(1)}}(t))h^{l^{(2)}}(\theta^{l^{(2)}}(t)),\quad t\in\mathbb{N}_{0}.
\end{align*}
 Now let $h^{\tilde{l}}\colon\Theta^{\tilde{l}}\to\{0,1\}$ be given
by 
\begin{align*}
h^{\tilde{l}}((q,r)) & =h^{l^{(1)}}(q)h^{l^{(2)}}(r),\quad(q,r)\in\Theta^{\tilde{l}}.
\end{align*}
 It follows that \eqref{eq:hiddenmarkov} holds with $l$ replaced
with $\tilde{l}$. Thus, $\{\tilde{l}(t)\in\{0,1\}\}_{t\in\mathbb{N}_{0}}$
is the output process of a time-inhomogeneous hidden Markov model. 

Our next goal is to prove $\tilde{l}\in\Gamma_{\tilde{p}_{0},\tilde{p}_{1}}$
by showing 
\begin{align}
\sum_{(r^{(1)},r^{(2)})\in\Theta_{0}^{\tilde{l}}}p_{(q^{(1)},q^{(2)}),(r^{(1)},r^{(2)})}^{\tilde{l}}(t) & \leq\tilde{p}_{0},\label{eq:tildep0}\\
\sum_{(r^{(1)},r^{(2)})\in\Theta_{1}^{\tilde{l}}}p_{(q^{(1)},q^{(2)}),(r^{(1)},r^{(2)})}^{\tilde{l}}(t) & \leq\tilde{p}_{1},\quad t\in\mathbb{N}_{0}.\label{eq:tildep1}
\end{align}
 First, we show \eqref{eq:tildep0}. Observe that 
\begin{align}
\Theta_{0}^{\tilde{l}} & =\{(r^{(1)},r^{(2)})\in\Theta^{\tilde{l}}\colon h^{\tilde{l}}((r^{(1)},r^{(2)}))=0\}\nonumber \\
 & =\{(r^{(1)},r^{(2)})\in\Theta^{\tilde{l}}\colon h^{l^{(1)}}(r^{(1)})h^{l^{(2)}}(r^{(2)})=0\}.\label{eq:thetazerofirst}
\end{align}
 Now let $j_{1},j_{2}\in\{1,2\}$ be such that $j_{1}\neq j_{2}$.
It follows from \eqref{eq:thetazerofirst} that 
\begin{align}
\Theta_{0}^{\tilde{l}} & =\{(r^{(1)},r^{(2)})\in\Theta^{\tilde{l}}\colon h^{l^{(j_{1})}}(r^{(j_{1})})=1,h^{l^{(j_{2})}}(r^{(j_{2})})=0\}\nonumber \\
 & \quad\cup\{(r^{(1)},r^{(2)})\in\Theta^{\tilde{l}}\colon h^{l^{(j_{1})}}(r^{(j_{1})})=0\}.\label{eq:thetazerosecond}
\end{align}
Hence, by \eqref{eq:thetazerosecond}, we obtain 
\begin{align}
 & \sum_{(r^{(1)},r^{(2)})\in\Theta_{0}^{\tilde{l}}}p_{(q^{(1)},q^{(2)}),(r^{(1)},r^{(2)})}^{\tilde{l}}(t)\nonumber \\
 & \,\,=\sum_{r^{(j_{1})}\in\Theta_{0}^{l^{(j_{1})}},r^{(j_{2})}\in\Theta^{l^{(j_{2})}}}p_{(q^{(1)},q^{(2)}),(r^{(1)},r^{(2)})}^{\tilde{l}}(t)\nonumber \\
 & \,\,\quad+\sum_{r^{(j_{1})}\in\Theta_{1}^{l^{(j_{1})}},r^{(j_{2})}\in\Theta_{0}^{l^{(j_{2})}}}p_{(q^{(1)},q^{(2)}),(r^{(1)},r^{(2)})}^{\tilde{l}}(t)\nonumber \\
 & \,\,=\sum_{r^{(j_{1})}\in\Theta_{0}^{l^{(j_{1})}}}\sum_{r^{(j_{2})}\in\Theta^{l^{(j_{2})}}}p_{q^{(j_{1})},r^{(j_{1})}}^{l^{(j_{1})}}(t)p_{q^{(j_{2})},r^{(j_{2})}}^{l^{(j_{2})}}(t)\nonumber \\
 & \,\,\quad+\sum_{r^{(j_{1})}\in\Theta_{1}^{l^{(j_{1})}}}\sum_{r^{(j_{2})}\in\Theta_{0}^{l^{(j_{2})}}}p_{q^{(j_{1})},r^{(j_{1})}}^{l^{(j_{1})}}(t)p_{q^{(j_{2})},r^{(j_{2})}}^{l^{(j_{2})}}(t)\nonumber \\
 & \,\,=\sum_{r^{(j_{2})}\in\Theta^{l^{(j_{2})}}}p_{q^{(j_{2})},r^{(j_{2})}}^{l^{(j_{2})}}(t)\sum_{r^{(j_{1})}\in\Theta_{0}^{l^{(j_{1})}}}p_{q^{(j_{1})},r^{(j_{1})}}^{l^{(j_{1})}}(t)\nonumber \\
 & \,\,\quad+\sum_{r^{(j_{1})}\in\Theta_{1}^{l^{(j_{1})}}}p_{q^{(j_{1})},r^{(j_{1})}}^{l^{(j_{1})}}(t)\sum_{r^{(j_{2})}\in\Theta_{0}^{l^{(j_{2})}}}p_{q^{(j_{2})},r^{(j_{2})}}^{l^{(j_{2})}}(t).\label{eq:tildep0derivation1}
\end{align}
 Now, since $l^{(1)}\in\Gamma_{p_{0}^{(1)},p_{1}^{(1)}}$, $l^{(2)}\in\Gamma_{p_{0}^{(2)},p_{1}^{(2)}}$,
we have $\sum_{r^{(j_{1})}\in\Theta_{0}^{l^{(j_{1})}}}p_{q^{(j_{1})},r^{(j_{1})}}^{l^{(j_{1})}}(t)\leq p_{0}^{(j_{1})}$,
$\sum_{r^{(j_{1})}\in\Theta_{1}^{l^{(j_{1})}}}p_{q^{(j_{1})},r^{(j_{1})}}^{l^{(j_{1})}}(t)\leq p_{1}^{(j_{1})}$,
and $\sum_{r^{(j_{2})}\in\Theta_{0}^{l^{(j_{2})}}}p_{q^{(j_{2})},r^{(j_{2})}}^{l^{(j_{2})}}(t)\leq p_{0}^{(j_{2})}$,
$t\in\mathbb{N}_{0}$. Furthermore, we have $\sum_{r^{(j_{2})}\in\Theta^{l^{(j_{2})}}}p_{q^{(j_{2})},r^{(j_{2})}}^{l^{(j_{2})}}(t)=1$,
since the summation is over all possible states $r^{(j_{2})}\in\Theta^{l^{(j_{2})}}$.
Using these inequalities in \eqref{eq:tildep0derivation1}, we obtain
\begin{align}
\sum_{(r^{(1)},r^{(2)})\in\Theta_{0}^{\tilde{l}}}p_{(q^{(1)},q^{(2)}),(r^{(1)},r^{(2)})}^{\tilde{l}}(t) & \leq p_{0}^{(j_{1})}+p_{1}^{(j_{1})}p_{0}^{(j_{2})},\label{eq:tildep0derivation2}
\end{align}
 for $t\in\mathbb{N}_{0}$. Since \eqref{eq:tildep0derivation2} holds
for all $j_{1},j_{2}\in\{1,2\}$ such that $j_{1}\neq j_{2}$, we
have 
\begin{align}
\sum_{(r^{(1)},r^{(2)})\in\Theta_{0}^{\tilde{l}}}p_{(q^{(1)},q^{(2)}),(r^{(1)},r^{(2)})}^{\tilde{l}}(t) & \leq\min\{p_{0}^{(1)}+p_{1}^{(1)}p_{0}^{(2)},p_{0}^{(2)}+p_{1}^{(2)}p_{0}^{(1)}\},\quad t\in\mathbb{N}_{0}.\label{eq:tildep0derivation3}
\end{align}
Furthermore, noting that $\Theta_{0}^{\tilde{l}}\subset\Theta^{\tilde{l}}$,
we obtain $\sum_{(r^{(1)},r^{(2)})\in\Theta_{0}^{\tilde{l}}}p_{(q^{(1)},q^{(2)}),(r^{(1)},r^{(2)})}^{\tilde{l}}(t)\leq\sum_{(r^{(1)},r^{(2)})\in\Theta^{\tilde{l}}}p_{(q^{(1)},q^{(2)}),(r^{(1)},r^{(2)})}^{\tilde{l}}(t)=1$,
$t\in\mathbb{N}_{0}$. By using this inequality, it follows from \eqref{eq:tildep0derivation3}
that \eqref{eq:tildep0} holds. 

Next, we show \eqref{eq:tildep1}. Notice that 
\begin{align}
\Theta_{1}^{\tilde{l}} & =\{(r^{(1)},r^{(2)})\in\Theta^{\tilde{l}}\colon h^{\tilde{l}}((r^{(1)},r^{(2)}))=1\}\nonumber \\
 & =\{(r^{(1)},r^{(2)})\in\Theta^{\tilde{l}}\colon h^{l^{(1)}}(r^{(1)})h^{l^{(2)}}(r^{(2)})=1\}\nonumber \\
 & =\{(r^{(1)},r^{(2)})\in\Theta^{\tilde{l}}\colon h^{l^{(1)}}(r^{(1)})=1,h^{l^{(2)}}(r^{(2)})=1\}.\label{eq:thetaonefirst}
\end{align}
 Noting that $l^{(1)}\in\Gamma_{p_{0}^{(1)},p_{1}^{(1)}}$, $l^{(2)}\in\Gamma_{p_{0}^{(2)},p_{1}^{(2)}}$,
we use \eqref{eq:thetaonefirst} to obtain for $t\in\mathbb{N}_{0}$,
\begin{align*}
 & \sum_{(r^{(1)},r^{(2)})\in\Theta_{1}^{\tilde{l}}}p_{(q^{(1)},q^{(2)}),(r^{(1)},r^{(2)})}^{\tilde{l}}(t)=\sum_{r^{(1)}\in\Theta_{1}^{l^{(1)}},r^{(2)}\in\Theta_{1}^{l^{(2)}}}p_{(q^{(1)},q^{(2)}),(r^{(1)},r^{(2)})}^{\tilde{l}}(t)\\
 & \quad=\sum_{r^{(1)}\in\Theta_{1}^{l^{(1)}}}\sum_{r^{(2)}\in\Theta_{1}^{l^{(2)}}}p_{q^{(1)},r^{(1)}}^{l^{(1)}}(t)p_{q^{(2)},r^{(2)}}^{l^{(2)}}(t)\\
 & \quad=\sum_{r^{(1)}\in\Theta_{1}^{l^{(1)}}}p_{q^{(1)},r^{(1)}}^{l^{(1)}}(t)\sum_{r^{(2)}\in\Theta_{1}^{l^{(2)}}}p_{q^{(2)},r^{(2)}}^{l^{(2)}}(t)\leq p_{1}^{(1)}p_{1}^{(2)},
\end{align*}
which implies \eqref{eq:tildep1}. Now since \eqref{eq:tildep0} and
\eqref{eq:tildep1} hold, we have $\tilde{l}\in\Gamma_{\tilde{p}_{0},\tilde{p}_{1}}$.
\end{proof}

Theorem~\ref{Theorem-HiddenMarkov} shows that when two hidden Markov
output processes $l^{(1)}$ and $l^{(2)}$ are combined with $\wedge$
operation, the resulting process $\tilde{l}$ is also a hidden Markov
output process. Furthermore, Theorem~\ref{Theorem-HiddenMarkov}
provides the values of $\tilde{p}_{0},\tilde{p}_{1}$ for which $\tilde{l}\in\Gamma_{\tilde{p}_{0},\tilde{p}_{1}}$. 

This result can be applied to obtain $\rho_{G}$. For instance, consider
the case $c=2$, where $l_{\mathcal{P}_{1}}(\cdot)$ and $l_{\mathcal{P}_{2}}(\cdot)$
are the output processes of hidden Markov models such that $l_{\mathcal{P}_{1}}\in\Gamma_{p_{0}^{(1)},p_{1}^{(1)}}$,
$l_{\mathcal{P}_{2}}\in\Gamma_{p_{0}^{(2)},p_{1}^{(2)}}$. It follows
from Theorem~\ref{Theorem-HiddenMarkov} with $l^{(1)}=l_{\mathcal{P}_{1}}$,
$l^{(2)}=l_{\mathcal{P}_{2}}$, and $\tilde{l}=l_{G}$ that $l_{G}\in\Gamma_{\tilde{p}_{0},\tilde{p}_{1}}$
with $\tilde{p}_{1}\triangleq p_{1}^{(1)}p_{1}^{(2)}$. Now, suppose
that $p_{1}^{(1)}p_{1}^{(2)}<1$. Notice that since $l_{\mathcal{P}_{1}}\in\Gamma_{p_{0}^{(1)},p_{1}^{(1)}}$
and $l_{\mathcal{P}_{2}}\in\Gamma_{p_{0}^{(2)},p_{1}^{(2)}}$, we
have $l_{\mathcal{P}_{1}}\in\Lambda_{\rho_{\mathcal{P}_{1}}}$ and
$l_{\mathcal{P}_{2}}\in\Lambda_{\rho_{\mathcal{P}_{2}}}$ with $\rho_{\mathcal{P}_{1}}\in(p_{1}^{(1)},1]$
and $\rho_{\mathcal{P}_{2}}\in(p_{1}^{(2)},1]$. The direct application
of Proposition~\ref{Proposition-and-operation} gives $l_{G}\in\Lambda_{\rho_{G}}$
with $\rho_{G}=\min\{\rho_{\mathcal{P}_{1}},\rho_{\mathcal{P}_{2}}\}$.
However, by applying Proposition~\ref{PropositionGammaLambdaRelation},
we can obtain a smaller value for $\rho_{G}$. In fact by Proposition~\ref{PropositionGammaLambdaRelation},
we obtain $l_{G}\in\Lambda_{\rho_{G}}$ for any $\rho_{G}\in(p_{1}^{(1)}p_{1}^{(2)},1]$.
Notice that in the case where $c>2$, Theorem~\ref{Theorem-HiddenMarkov}
can be applied repeatedly. For instance, when $c=3$, we can use Theorem~\ref{Theorem-HiddenMarkov}
first for $l_{\mathcal{P}_{2}}(t)\wedge l_{\mathcal{P}_{3}}(t)$ and
then for $l_{G}(t)=l_{\mathcal{P}_{1}}(t)\wedge(l_{\mathcal{P}_{2}}(t)\wedge l_{\mathcal{P}_{3}}(t))$. 

Now consider the case where the graph $G$ possesses some paths with
indicator processes that are mutually independent but not all of them
are associated with random failures. Even for this case, we can obtain
results that are tighter than Proposition~\ref{Proposition-and-operation}.
To this end, we first provide the following result where we derive
properties of a process that is obtained by using $\wedge$ operation
on a hidden Markov output process $l^{(1)}\in\Gamma_{p_{0}^{(1)}p_{1}^{(1)}}$
and a binary-valued process $l^{(2)}\in\Lambda_{\rho^{(2)}}$. 

\begin{theorem} \label{Theorem-Gamma-Lambda} ~~Consider the binary-valued
processes $\{l^{(1)}\in\{0,1\}\}_{t\in\mathbb{N}_{0}}$ and $\{l^{(2)}\in\{0,1\}\}_{t\in\mathbb{N}_{0}}$
that satisfy $l^{(1)}\in\Gamma_{p_{0}^{(1)},p_{1}^{(1)}}$ and $l^{(2)}\in\Lambda_{\rho^{(2)}}$
with $p_{1}^{(1)}\rho^{(2)}<1$. Then the process $\{\tilde{l}(t)\in\{0,1\}\}_{t\in\mathbb{N}_{0}}$
defined by 
\begin{align}
\tilde{l}(t) & =l^{(1)}(t)\wedge l^{(2)}(t),\quad t\in\mathbb{N}_{0},\label{eq:ltildedef-1}
\end{align}
satisfies $\tilde{l}\in\Lambda_{\tilde{\rho}}$ for all $\tilde{\rho}\in(p_{1}^{(1)}\rho^{(2)},1]$.
\end{theorem}

\begin{proof} For the case where $p_{1}^{(1)}\rho^{(2)}=0$ and the
case where $p_{1}^{(1)}=1$, the result follows from Proposition~\ref{Proposition-and-operation}.
Now consider the case where $p_{1}^{(1)}\in(0,1),\rho^{(2)}\in(0,1]$.
We will first show that $\tilde{l}\in\Lambda_{\tilde{\rho}}$ for
$\tilde{\rho}\in(p_{1}^{(1)}\rho^{(2)},\rho^{(2)})$. The process
$\{\chi(t)\in\{0,1\}\}_{t\in\mathbb{N}_{0}}$ defined by $\chi(t)=l^{(2)}(t),t\in\mathbb{N}_{0}$,
satisfies \eqref{eq:chicond} with $\tilde{w}=\rho^{(2)}$. Furthermore,
$\{\xi(t)\in\{0,1\}\}_{t\in\mathbb{N}_{0}}$ defined by $\xi(t)=\theta^{l^{(1)}}(t),t\in\mathbb{N}_{0}$,
satisfies \eqref{eq:xicond} with $\tilde{p}=p_{1}^{(1)}$, $h=h^{l^{(1)}},$
$\Xi=\Theta^{l^{(1)}}$, and $\Xi_{1}=\Theta_{1}^{l^{(1)}}$. It then
follows from Lemma~\ref{KeyMarkovLemma} that for all $\rho\in(\tilde{p}\tilde{w},\tilde{w})$,
we have 
\begin{align*}
 & \sum_{k\in\mathbb{N}}\mathbb{P}[\sum_{t=0}^{k-1}h(\xi(t))\chi(t)>\rho k]=\sum_{k\in\mathbb{N}}\mathbb{P}[\sum_{t=0}^{k-1}l^{(1)}(t)\wedge l^{(2)}(t)>\rho k]<\infty.
\end{align*}
 Now since $\tilde{p}\tilde{w}=p_{1}^{(1)}\rho^{(2)}$ and $\tilde{w}=\rho^{(2)}$,
it follows that $\tilde{l}\in\Lambda_{\tilde{\rho}}$ holds for $\tilde{\rho}\in(p_{1}^{(1)}\rho^{(2)},\rho^{(2)})$.
This implies $\tilde{l}\in\Lambda_{\tilde{\rho}}$, $\tilde{\rho}\in[\rho^{(2)},1]$,
since $\Lambda_{\rho_{1}}\subseteq\Lambda_{\rho_{2}}$ for any $\rho_{1},\rho_{2}\in[0,1]$
such that $\rho_{1}\leq\rho_{2}$. Consequently, we have $\tilde{l}\in\Lambda_{\tilde{\rho}}$
for $\tilde{\rho}\in(p_{1}^{(1)}\rho^{(2)},1]=(p_{1}^{(1)}\rho^{(2)},\rho^{(2)})\cup[\rho^{(2)},1]$.
\end{proof}

Theorem~\ref{Theorem-Gamma-Lambda} is concerned with $\wedge$ operation
applied to a process $l^{(1)}(\cdot)$ from the hidden Markov model
class $\Gamma_{p_{0}^{(1)},p_{1}^{(1)}}$ and another process $l^{(2)}(\cdot)$
from the class $\Lambda_{\rho^{(2)}}$. It is shown that if $p_{1}^{(1)}\rho^{(2)}<1$,
then this operation results in a process $\tilde{l}$ that satisfies
$\tilde{l}\in\Lambda_{\tilde{\rho}}$ for all $\tilde{\rho}\in(p_{1}^{(1)}\rho^{(2)},1]$.
We note that the application of Proposition~\ref{Proposition-and-operation}
to this situation would allow us to show $\tilde{l}\in\Lambda_{\tilde{\rho}}$
for all $\tilde{\rho}\in(\min\{p_{1}^{(1)},\rho^{(2)}\},1]$. Notice
that Proposition~\ref{Proposition-and-operation} in this case is
conservative since $\min\{p_{1}^{(1)},\rho^{(2)}\}>p_{1}^{(1)}\rho^{(2)}$.
On the other hand, we note that Proposition~\ref{Proposition-and-operation}
allows us to deal with processes that are not mutually independent. 

\begin{remark} Theorem~\ref{Theorem-HiddenMarkov} explains the
joint effects of random transmission failures happening on two different
paths. On the other hand, Theorem~\ref{Theorem-Gamma-Lambda} is
concerned with the case where one path is associated with random packet
losses and the other path may be subject to more general types of
failures (random, malicious, or a combination of both). In this sense,
Theorem~\ref{Theorem-Gamma-Lambda} considers a more general situation
than that considered in Theorem~\ref{Theorem-HiddenMarkov}. In fact
Theorem~\ref{Theorem-Gamma-Lambda} can also be utilized for two
processes that are both associated with random failures. The difference
between Theorem~\ref{Theorem-HiddenMarkov} and Theorem~\ref{Theorem-Gamma-Lambda}
is that they describe the joint effects of the processes through different
classes ($\Gamma$ and $\Lambda$, respectively). 

We also note that when two processes associated with random failures
are involved, Theorem~\ref{Theorem-HiddenMarkov} is more advantageous
than Theorem~\ref{Theorem-Gamma-Lambda}, since Theorem~\ref{Theorem-Gamma-Lambda}
may result in conservatism in certain situations. Consider for example
$l^{(1)}\in\Gamma_{p_{0}^{(1)},p_{1}^{(1)}}$ and $l^{(2)}\in\Gamma_{p_{0}^{(2)},p_{1}^{(2)}}$
associated both with random failures and $l^{(3)}\in\Lambda_{\rho^{(3)}}$
associated with malicious packet losses. Since $l^{(2)}\in\Lambda_{\rho^{(2)}}$
with $\rho^{(2)}\in(p_{1}^{(2)},1]$, we can apply Theorem~\ref{Theorem-Gamma-Lambda}
to obtain $l^{(1)}\wedge l^{(2)}\in\Lambda_{\tilde{\rho}}$ for all
$\tilde{\rho}\in(p_{1}^{(1)}\rho^{(2)},1]$. Here, by Theorem~\ref{Theorem-Gamma-Lambda},
the process associated with joint failures ($l^{(1)}\wedge l^{(2)}$)
belongs to the class $\Lambda_{\tilde{\rho}}$. Now, in order to explain
the joint effects of $l^{(1)}\wedge l^{(2)}$ together with yet another
process $l^{(3)}\in\Lambda_{\rho^{(3)}}$, we would be required to
use Proposition~\ref{Proposition-and-operation}, which would result
in conservatism as explained after the proof of Theorem~\ref{Theorem-Gamma-Lambda}.
In order not to introduce unnecessary conservatism, we can first apply
Theorem~\ref{Theorem-HiddenMarkov} to $l^{(1)}\in\Gamma_{p_{0}^{(1)},p_{1}^{(1)}}$
and $l^{(2)}\in\Gamma_{p_{0}^{(2)},p_{1}^{(2)}}$ to obtain $l^{(1)}\wedge l^{(2)}\in\Gamma_{\tilde{p}_{0},\tilde{p}_{1}}$
with $\tilde{p}_{0}\triangleq\min\{p_{0}^{(1)}+p_{1}^{(1)}p_{0}^{(2)},p_{0}^{(2)}+p_{1}^{(2)}p_{0}^{(1)},1\}$
and $\tilde{p}_{1}\triangleq p_{1}^{(1)}p_{1}^{(2)}$. Notice now
that the process associated with joint failures ($l^{(1)}\wedge l^{(2)}$)
belongs to the class $\Gamma_{\tilde{p}_{0},\tilde{p}_{1}}$. Then,
we can apply Theorem~\ref{Theorem-Gamma-Lambda} for processes $l^{(1)}\wedge l^{(2)}\in\Gamma_{\tilde{p}_{0},\tilde{p}_{1}}$
and $l^{(3)}\in\Lambda_{\rho^{(3)}}$, obtaining a less conservative
result. \end{remark}

We summarize the results presented in this section in Table~\ref{TableForAnd},
where we indicate the classes obtained through the $\wedge$ operation. 

\begin{table}[t]
\caption{Comparison of the classes of processes obtained by combining processes
$l^{(1)}$ and $l^{(2)}$ of different classes through $\wedge$ operation.
For the case where $l^{(1)}\in\Lambda_{\rho^{(1)}}$ and $l^{(2)}\in\Lambda_{\rho^{(2)}}$,
the processes $l^{(1)}$ and $l^{(2)}$ can be dependent; for other
cases, $l^{(1)}$ and $l^{(2)}$ are assumed to be mutually independent. }
\label{TableForAnd} 

\renewcommand{\arraystretch}{1.2} 

{\centering \fontsize{7}{11.52}\selectfont 

\begin{tabular}{|c|c||c|}
\hline 
$l^{(1)}$ & $l^{(2)}$ & $l^{(1)}\wedge l^{(2)}$\tabularnewline
\hline 
\hline 
\multirow{3}{*}{$\Gamma_{p_{0}^{(1)},p_{1}^{(1)}}$} & \multirow{2}{*}{$\Gamma_{p_{0}^{(2)},p_{1}^{(2)}}$} & $\Gamma_{\tilde{p}_{0},\tilde{p}_{1}}$ with $\tilde{p}_{0}\triangleq\min\{p_{0}^{(1)}+p_{1}^{(1)}p_{0}^{(2)},p_{0}^{(2)}+p_{1}^{(2)}p_{0}^{(1)},1\}$
and $\tilde{p}_{1}\triangleq p_{1}^{(1)}p_{1}^{(2)}$\tabularnewline
 &  & (Theorem~\ref{Theorem-HiddenMarkov})\tabularnewline
\cline{2-3} 
 & $\Lambda_{\rho^{(2)}}$ & $\Lambda_{\tilde{\rho}}$ for $\tilde{\rho}\in(p_{1}^{(1)}\rho^{(2)},1]$
(Theorem~\ref{Theorem-Gamma-Lambda}) \tabularnewline
\hline 
\multirow{2}{*}{$\Lambda_{\rho^{(1)}}$} & $\Gamma_{p_{0}^{(2)},p_{1}^{(2)}}$ & $\Lambda_{\tilde{\rho}}$ for $\tilde{\rho}\in(p_{1}^{(2)}\rho^{(1)},1]$
(Theorem~\ref{Theorem-Gamma-Lambda}) \tabularnewline
\cline{2-3} 
 & $\Lambda_{\rho^{(2)}}$ & $\Lambda_{\tilde{\rho}}$ for $\tilde{\rho}\in[\min\{\rho^{(1)}\rho^{(2)}\},1]$
(Proposition~\ref{Proposition-and-operation})\tabularnewline
\hline 
\end{tabular}

} 
\end{table}

\begin{remark}

Although in this paper we assume delay-free packet transmissions,
in practice delays are inevitable. Moreover, different routes represented
by different paths may induce different amounts of delay for the transmission
of the same packet. This is because the numbers of links on different
paths $\mathcal{P}_{i}$ may be different, and furthermore, links
on those paths may have different transmission properties. As a result,
the receiver node (e.g., the controller node $v_{\mathrm{C}}$ in
$G$) may obtain the same packet from different paths at different
times. It is clear that the performance can be improved if the controller
acts immediately upon receiving the first uncorrupted state packet
from the path with the shortest delay. Furthermore, we can introduce
setups that utilize a delay threshold: a path $\mathcal{P}_{i}$ that
faces delay beyond the threshold can be considered to have faced a
failure ($l_{\mathcal{P}_{i}}(t)=0$) in transmission of the state
packet $x(t)$. In such setups, the delay properties of the links
on different paths have to be taken into account to fully model the
processes $\{l_{\mathcal{P}_{i}}(t)\in\{0,1\}\}_{t\in\mathbb{N}_{0}}$.
This requires further analysis that is different from what we provide
in this paper. In particular, the effects of different delay models
for links as well as the combined effects of delays and failures (due
to data corruption and packet dropping) need to be investigated. \end{remark} 

\section{Packet Transmission Failures on Paths of a Network}

\label{Sec:FailuresOnPaths}

So far, in the previous section, we have looked at how packet failures
on the paths of a network affect the overall packet transmission rate.
In this section, our goal is to explore the effect of the failures
at individual nodes and links of a path. To this end, we first consider
the scenario where packet transmission failures occur due to only
data corruption. We then explore the case where data corruption and
packet drops may occur on the same path. 

\subsection{Characterization for Data Corrupting Paths}

Let $\mathcal{P}_{i,j}$ denote the $j$th edge on path $\mathcal{P}_{i}$.
We use the binary-valued process $\{l_{\mathcal{P}_{i}}^{\mathcal{P}_{i,j}}(t)\in\{0,1\}\}_{t\in\mathbb{N}_{0}}$
to denote the data corruption indicator for this link. For example,
in Fig.~\ref{Flo:operation}, consider the second edge $\mathcal{P}_{1,2}=(v_{1},v_{\mathrm{C}})$
of path $\mathcal{P}_{1}$ in \eqref{eq:example-paths}. The state
$l_{\mathcal{P}_{1}}^{(v_{1},v_{\mathrm{C}})}(t)=1$ indicates that
at time $t$, the packet flowing on path $\mathcal{P}_{1}$ faces
data corruption on the link $(v_{1},v_{\mathrm{C}})$. This may be
due to a jamming attack on this link, or due to channel noise, and
moreover, it may also be the case that the node $v_{1}$ maliciously
corrupts the packet. 

The notation for the data corruption indicator allows us to distinguish
data corruption issues when we consider the same communication link
on different paths. For instance, communication link $(v_{4},v_{\mathrm{C}})$
is on both $\mathcal{P}_{2}$ and $\mathcal{P}_{3}$. It may be the
case that node $v_{4}$ corrupts all packets transmitted along $\mathcal{P}_{2}$,
but none of the packets transmitted along $\mathcal{P}_{3}$. This
situation can be described by setting $l_{\mathcal{P}_{2}}^{(v_{4},v_{\mathrm{C}})}(t)=1$,
$t\in\mathbb{N}_{0}$, and $l_{\mathcal{P}_{3}}^{(v_{4},v_{\mathrm{C}})}(t)=0$,
$t\in\mathbb{N}_{0}$.

State packet transmitted through path $\mathcal{P}_{i}$ is subject
to data corruption if there is data corruption on one (or more) of
the edges in this path. Hence, for each $i\in\{1,\ldots,c\}$, 
\begin{align}
l_{\mathcal{P}_{i}}(t) & =l_{\mathcal{P}_{i}}^{\mathcal{P}_{i,1}}(t)\vee l_{\mathcal{P}_{i}}^{\mathcal{P}_{i,2}}(t)\vee\cdots\vee l_{\mathcal{P}_{i}}^{\mathcal{P}_{i,|\mathcal{P}_{i}|}}(t).\label{eq:lp-or}
\end{align}

The next result shows that an upper-bound for the asymptotic transmission
failure rate of a path can be given as the sum of the failure rate
bounds of the links on the path. 

\begin{proposition}\label{Proposition-lp-or} Consider $\{l_{\mathcal{P}_{i}}(t)\in\{0,1\}\}_{t\in\mathbb{N}_{0}}$
given by \eqref{eq:lp-or}. Assume $l_{\mathcal{P}_{i}}^{\mathcal{P}_{i,j}}\in\Lambda_{\rho_{\mathcal{P}_{i}}^{\mathcal{P}_{i,j}}}$,
$j\in\{1,\ldots,|\mathcal{P}_{i}|\}$, where $\rho_{\mathcal{P}_{i}}^{\mathcal{P}_{i,j}}\in[0,1]$,
$j\in\{1,\ldots,|\mathcal{P}_{i}|\}$, satisfy $\sum_{j=1}^{|\mathcal{P}_{i}|}\rho_{\mathcal{P}_{i}}^{\mathcal{P}_{i,j}}\leq1$.
Then $l_{\mathcal{P}_{i}}\in\Lambda_{\rho_{\mathcal{P}_{i}}}$ with
$\rho_{\mathcal{P}_{i}}\triangleq\sum_{j=1}^{|\mathcal{P}_{i}|}\rho_{\mathcal{P}_{i}}^{\mathcal{P}_{i,j}}$. 

\end{proposition}

\begin{proof}By \eqref{eq:lp-or}, $\sum_{t=0}^{k-1}l_{\mathcal{P}_{i}}(t)\leq\sum_{j=1}^{|\mathcal{P}_{i}|}\sum_{t=0}^{k-1}l_{\mathcal{P}_{i}}^{\mathcal{P}_{i,j}}(t)$.
Hence, 
\begin{align}
 & \mathbb{P}[\sum_{t=0}^{k-1}l_{\mathcal{P}_{i}}(t)>\rho_{\mathcal{P}_{i}}k]\leq\mathbb{P}[\sum_{j=1}^{|\mathcal{P}_{i}|}\sum_{t=0}^{k-1}l_{\mathcal{P}_{i}}^{\mathcal{P}_{i,j}}(t)>\rho_{\mathcal{P}_{i}}k]\leq\sum_{j=1}^{|\mathcal{P}_{i}|}\mathbb{P}[\sum_{t=0}^{k-1}l_{\mathcal{P}_{i}}^{\mathcal{P}_{i,j}}(t)>\rho_{\mathcal{P}_{i}}^{\mathcal{P}_{i,j}}k],\label{eq:separate-sum}
\end{align}
for $k\in\mathbb{N}$, where to obtain the last inequality, we used
\begin{align*}
\mathbb{P}[\sum_{j=1}^{|\mathcal{P}_{i}|}\gamma_{j} & >0]\leq\mathbb{P}[\cup_{j=1}^{|\mathcal{P}_{i}|}\{\gamma_{j}>0\}]\leq\sum_{j=1}^{|\mathcal{P}_{i}|}\mathbb{P}[\gamma_{j}>0]
\end{align*}
 with $\gamma_{j}\triangleq\sum_{t=0}^{k-1}l_{\mathcal{P}_{i}}^{\mathcal{P}_{i,j}}(t)-\rho_{\mathcal{P}_{i}}^{\mathcal{P}_{i,j}}k$.
The result then follows from \eqref{eq:separate-sum} since $l_{\mathcal{P}_{i}}^{\mathcal{P}_{i,j}}\in\Lambda_{\rho_{\mathcal{P}_{i}}^{\mathcal{P}_{i,j}}}$,
$j\in\{1,\ldots,|\mathcal{P}_{i}|\}$. \end{proof} 

Proposition~\ref{Proposition-lp-or} can be used to characterize
overall failures on a path. Note that in Proposition~\ref{Proposition-lp-or},
the indicators for the links are not necessarily mutually-independent
processes. This allows us to model failures on different links that
depend on each other. In particular, we can explore the effect of
interference between links as well as coordinated jamming attacks
targeting several links at the same time. 

Note that in certain cases the result provided in Proposition~\ref{Proposition-lp-or}
can be improved in terms of conservativeness. In particular, if one
or more links in a path are known to be associated with random failures
and the corresponding indicator processes are mutually independent,
we can obtain less conservative results in comparison to Proposition~\ref{Proposition-lp-or}.
The following result is the counterpart of Theorem~\ref{Theorem-HiddenMarkov}
for $\vee$ operation and it is concerned with output processes of
two mutually-independent hidden Markov models. 

\begin{theorem} \label{Theorem-HiddenMarkov-1} Consider the binary-valued
output processes $\{l^{(1)}\in\{0,1\}\}_{t\in\mathbb{N}_{0}}$ and
$\{l^{(2)}\in\{0,1\}\}_{t\in\mathbb{N}_{0}}$ of hidden Markov models
such that $l^{(1)}\in\Gamma_{p_{0}^{(1)},p_{1}^{(1)}}$, $l^{(2)}\in\Gamma_{p_{0}^{(2)},p_{1}^{(2)}}$
with $p_{0}^{(1)},p_{1}^{(1)},p_{0}^{(2)},p_{1}^{(2)}\in[0,1]$. Suppose
that the Markov chains $\{\theta^{l^{(1)}}(t)\in\Theta^{l^{(1)}}\}_{t\in\mathbb{N}_{0}}$
and $\{\theta^{l^{(2)}}(t)\in\Theta^{l^{(2)}}\}_{t\in\mathbb{N}_{0}}$
associated with the processes $l^{(1)}$ and $l^{(2)}$ are mutually
independent. Let $\tilde{p}_{0}\triangleq p_{0}^{(1)}p_{0}^{(2)}$
and $\tilde{p}_{1}\triangleq\min\{p_{1}^{(1)}+p_{0}^{(1)}p_{1}^{(2)},p_{1}^{(2)}+p_{0}^{(2)}p_{1}^{(1)},1\}$.
Then the process $\{\tilde{l}(t)\in\{0,1\}\}_{t\in\mathbb{N}_{0}}$
defined by 
\begin{align}
\tilde{l}(t) & =l^{(1)}(t)\vee l^{(2)}(t),\quad t\in\mathbb{N}_{0},\label{eq:ltildedef-2}
\end{align}
 is the output process of a time-inhomogeneous hidden Markov model,
and moreover, $\tilde{l}\in\Gamma_{\tilde{p}_{0},\tilde{p}_{1}}$. 

\end{theorem}

\begin{proof}Let $\Theta^{\tilde{l}}\triangleq\{(q^{(1)},q^{(2)})\colon q^{(1)}\in\Theta^{l^{(1)}},q^{(2)}\in\Theta^{l^{(2)}}\}$
and define the bivariate process $\{\theta^{\tilde{l}}(t)\in\Theta^{\tilde{l}}\}_{t\in\mathbb{N}_{0}}$
by 
\begin{align*}
\theta^{\tilde{l}}(t) & =(\theta^{l^{(1)}}(t),\theta^{l^{(2)}}(t)),\quad t\in\mathbb{N}_{0}.
\end{align*}
 It follows that $\{\theta^{\tilde{l}}(t)\in\Theta^{\tilde{l}}\}_{t\in\mathbb{N}_{0}}$
is a time-inhomogeneous Markov chain with initial distribution $\vartheta_{(q^{(1)},q^{(2)})}^{\tilde{l}}=\vartheta_{q^{(1)}}^{l^{(1)}}\vartheta_{q^{(2)}}^{l^{(2)}}$,
$q^{(1)}\in\Theta^{l^{(1)}},q^{(2)}\in\Theta^{l^{(2)}}$, and time-varying
transition probabilities $p_{(q^{(1)},q^{(2)}),(r^{(1)},r^{(2)})}^{\tilde{l}}(t)=p_{q^{(1)},r^{(1)}}^{l^{(1)}}(t)p_{q^{(2)},r^{(2)}}^{l^{(2)}}(t)$,
$t\in\mathbb{N}_{0}$. Here, for $j\in\{1,2\}$, $\vartheta^{l^{(j)}}$
and $p^{l^{(j)}}(\cdot)$ respectively denote the initial distribution
and the transition probability function for the Markov chain $\{\theta^{l^{(j)}}(t)\in\Theta^{l^{(j)}}\}_{t\in\mathbb{N}_{0}}$
associated with the output process $\{l^{(j)}\in\{0,1\}\}_{t\in\mathbb{N}_{0}}$.
Furthermore, it follows from \eqref{eq:ltildedef} that 
\begin{align*}
\tilde{l}(t) & =l^{(1)}(t)\vee l^{(2)}(t)=l^{(1)}(t)+(1-l^{(1)}(t))l^{(2)}(t)\\
 & =h^{l^{(1)}}(\theta^{l^{(1)}}(t))+(1-h^{l^{(1)}}(\theta^{l^{(1)}}(t)))h^{l^{(2)}}(\theta^{l^{(2)}}(t)),
\end{align*}
 for $t\in\mathbb{N}_{0}$. Now let $h^{\tilde{l}}\colon\Theta^{\tilde{l}}\to\{0,1\}$
be given by 
\begin{align*}
h^{\tilde{l}}((q,r)) & =h^{l^{(1)}}(q)+(1-h^{l^{(1)}}(q))h^{l^{(2)}}(r),\quad(q,r)\in\Theta^{\tilde{l}}.
\end{align*}
 It follows that \eqref{eq:hiddenmarkov} holds with $l$ replaced
with $\tilde{l}$. Thus, $\{\tilde{l}(t)\in\{0,1\}\}_{t\in\mathbb{N}_{0}}$
is the output process of a time-inhomogeneous hidden Markov model. 

To show $\tilde{l}\in\Gamma_{\tilde{p}_{0},\tilde{p}_{1}}$, we can
apply Theorem~\ref{Theorem-HiddenMarkov}. To this end, we first
note that for an output process $l\in\Gamma_{p_{0},p_{1}}$, the complementing
process $\{l_{\mathrm{c}}(t)\in\{0,1\}\}_{t\in\mathbb{N}_{0}}$ given
by $l_{\mathrm{c}}(t)=1-l(t)$, $t\in\mathbb{N}_{0}$, is the output
process of a time-inhomogeneous hidden Markov model, for which $\Theta^{l_{\mathrm{c}}}=\Theta^{l}$,
$\theta^{l_{\mathrm{c}}}(t)=\theta^{l}(t)$, $t\in\mathbb{N}_{0}$.
Furthermore, we have $\Theta_{0}^{l_{\mathrm{c}}}=\Theta_{1}^{l}$
and $\Theta_{1}^{l_{\mathrm{c}}}=\Theta_{0}^{l}$, since $h^{l_{\mathrm{c}}}((q,r))=1-h^{l}((q,r))$,
$(q,r)\in\Theta^{l_{\mathrm{c}}}=\Theta^{l}.$ As a consequence, we
have $l\in\Gamma_{p_{0},p_{1}}$ if and only if $l_{\mathrm{c}}\in\Gamma_{p_{1},p_{0}}$.
In the following, we show $\tilde{l}\in\Gamma_{\tilde{p}_{0},\tilde{p}_{1}}$
by proving that $\{\tilde{l}_{\mathrm{c}}(t)\in\{0,1\}\}_{t\in\mathbb{N}_{0}}$
given by $\tilde{l}_{\mathrm{c}}(t)=1-\tilde{l}(t)$, $t\in\mathbb{N}_{0}$,
satisfies $\tilde{l}_{\mathrm{c}}\in\Gamma_{\tilde{p}_{1},\tilde{p}_{0}}$. 

First, observe that $\tilde{l}_{\mathrm{c}}(t)=1-\tilde{l}(t)=1-l^{(1)}(t)\vee l^{(2)}(t)=(1-l^{(1)}(t))\wedge(1-l^{(2)}(t))=l_{\mathrm{c}}^{(1)}(t)\wedge l_{\mathrm{c}}^{(2)}(t)$,
where $l_{\mathrm{c}}^{(i)}(t)=1-l^{(i)}(t),$ $i\in\{1,2\}$. Since
$l^{(1)}\in\Gamma_{p_{0}^{(1)},p_{1}^{(1)}},l^{(2)}\in\Gamma_{p_{0}^{(2)},p_{1}^{(2)}}$,
we have $l_{\mathrm{c}}^{(1)}\in\Gamma_{p_{1}^{(1)},p_{0}^{(1)}}$,
$l_{\mathrm{c}}^{(2)}\in\Gamma_{p_{1}^{(2)},p_{0}^{(2)}}$. Finally,
noting that $\tilde{l}_{\mathrm{c}}(\cdot)$ is obtained by using
$\wedge$ operation on processes $l_{\mathrm{c}}^{(1)}(\cdot)$, $l_{\mathrm{c}}^{(2)}(\cdot)$,
we can use Theorem~\ref{Theorem-HiddenMarkov}, to obtain $\tilde{l}_{\mathrm{c}}\in\Gamma_{\tilde{p}_{1},\tilde{p}_{0}}$,
which implies $\tilde{l}\in\Gamma_{\tilde{p}_{0},\tilde{p}_{1}}$.
\end{proof}

Theorem~\ref{Theorem-HiddenMarkov-1} shows that when two hidden
Markov output processes $l^{(1)}$ and $l^{(2)}$ are combined with
$\vee$ operation, the resulting process $\tilde{l}$ is also a hidden
Markov output process. Furthermore, it provides the values of $\tilde{p}_{0},\tilde{p}_{1}$
for which $\tilde{l}\in\Gamma_{\tilde{p}_{0},\tilde{p}_{1}}$. 

Theorem~\ref{Theorem-HiddenMarkov-1} can be applied in obtaining
$\rho_{\mathcal{P}_{i}}$ for a given path $\mathcal{P}_{i}$. Consider
for example a path $\mathcal{P}_{i}$ with $|\mathcal{P}_{i}|=2$
links. Assume that the failure indicator processes $l_{\mathcal{P}_{i}}^{\mathcal{P}_{i,1}}(\cdot)$,
$l_{\mathcal{P}_{i}}^{\mathcal{P}_{i,2}}(\cdot)$ associated with
the links are mutually independent and $l_{\mathcal{P}_{i}}^{\mathcal{P}_{i,1}}\in\Gamma_{p_{0}^{(1)},p_{1}^{(1)}}$,
$l_{\mathcal{P}_{i}}^{\mathcal{P}_{i,2}}\in\Gamma_{p_{0}^{(2)},p_{1}^{(2)}}$.
It follows from Theorem~\ref{Theorem-HiddenMarkov-1} with $l^{(1)}(t)=l_{\mathcal{P}_{i}}^{\mathcal{P}_{i,1}}(t)$
and $l^{(2)}(t)=l_{\mathcal{P}_{i}}^{\mathcal{P}_{i,2}}(t)$ that
$l_{\mathcal{P}_{i}}\in\Gamma_{\tilde{p}_{0},\tilde{p}_{1}}$ with
$\tilde{p}_{0}=p_{0}^{(1)}p_{0}^{(2)}$ and $\tilde{p}_{1}=\min\{p_{1}^{(1)}+p_{0}^{(1)}p_{1}^{(2)},p_{1}^{(2)}+p_{0}^{(2)}p_{1}^{(1)},1\}$.
Furthermore, if $\min\{p_{1}^{(1)}+p_{0}^{(1)}p_{1}^{(2)},p_{1}^{(2)}+p_{0}^{(2)}p_{1}^{(1)}\}<1$,
then by Proposition~\ref{PropositionGammaLambdaRelation}, we have
$l_{\mathcal{P}_{i}}\in\Lambda_{\rho_{\mathcal{P}_{i}}}$ with $\rho_{\mathcal{P}_{i}}\in(\min\{p_{1}^{(1)}+p_{0}^{(1)}p_{1}^{(2)},p_{1}^{(2)}+p_{0}^{(2)}p_{1}^{(1)}\},1]$. 

On the other hand, the direct application of Proposition~\ref{Proposition-lp-or}
provides a conservative result. To apply Proposition~\ref{Proposition-lp-or},
notice first that $l_{\mathcal{P}_{i}}^{\mathcal{P}_{i,1}}\in\Lambda_{\rho_{\mathcal{P}_{i,1}}}$,
$l_{\mathcal{P}_{i}}^{\mathcal{P}_{i,2}}\in\Lambda{}_{\rho_{\mathcal{P}_{i,2}}}$
with $\rho_{\mathcal{P}_{i,1}}\in(p_{1}^{(1)},1]$ and $\rho_{\mathcal{P}_{i,2}}\in(p_{1}^{(2)},1]$.
Hence, by Proposition~\ref{Proposition-lp-or}, we obtain the value
$\rho_{\mathcal{P}_{i}}=\rho_{\mathcal{P}_{i,1}}+\rho_{\mathcal{P}_{i,2}}$,
which implies that $l_{\mathcal{P}_{i}}\in\Lambda_{\rho_{\mathcal{P}_{i}}}$
with $\rho_{\mathcal{P}_{i}}\in(p_{1}^{(1)}+p_{1}^{(2)},1]$. The
inequality $\min\{p_{1}^{(1)}+p_{0}^{(1)}p_{1}^{(2)},p_{1}^{(2)}+p_{0}^{(2)}p_{1}^{(1)}\}\leq p_{1}^{(1)}+p_{1}^{(2)}$
allows us to conclude that Theorem~\ref{Theorem-HiddenMarkov-1}
provides a less conservative range for $\rho_{\mathcal{P}_{i}}$ compared
to Proposition~\ref{Proposition-lp-or}. 

In certain scenarios a path $\mathcal{P}_{i}$ may be composed of
communication links with mutually-independent indicator processes
some of which are not associated with random failures. In such scenarios,
it is again possible to obtain results that are less conservative
than those in Proposition~\ref{Proposition-lp-or}. Specifically,
in the following result we derive properties of a process that is
obtained by using $\vee$ operation on a hidden Markov output process
$l^{(1)}\in\Gamma_{p_{0}^{(1)}p_{1}^{(1)}}$ and a binary-valued process
$l^{(2)}\in\Lambda_{\rho^{(2)}}$.

\begin{theorem} \label{Theorem-Gamma-Lambda-1} ~~ Consider the
binary-valued processes $\{l^{(1)}\in\{0,1\}\}_{t\in\mathbb{N}_{0}}$
and $\{l^{(2)}\in\{0,1\}\}_{t\in\mathbb{N}_{0}}$ that satisfy $l^{(1)}\in\Gamma_{p_{0}^{(1)},p_{1}^{(1)}}$
and $l^{(2)}\in\Lambda_{\rho^{(2)}}$ with $p_{1}^{(1)}+p_{0}^{(1)}\rho^{(2)}<1$.
Then the process $\{\tilde{l}(t)\in\{0,1\}\}_{t\in\mathbb{N}_{0}}$
defined by 
\begin{align}
\tilde{l}(t) & =l^{(1)}(t)\vee l^{(2)}(t),\quad t\in\mathbb{N}_{0},\label{eq:ltildedef-1-1}
\end{align}
satisfies $\tilde{l}\in\Lambda_{\tilde{\rho}}$ for all $\tilde{\rho}\in(p_{1}^{(1)}+p_{0}^{(1)}\rho^{(2)},1]$.
\end{theorem}

\begin{proof} For the case where $p_{0}^{(1)}=1$ and the case where
$\rho^{(2)}=0$, the result follows from  Proposition~\ref{Proposition-lp-or}.
Here, we consider the case where $p_{0}^{(1)}\in(0,1),\rho^{(2)}\in(0,1]$.
Notice that since $\tilde{l}(\cdot)$ is a binary-valued process,
we have $\Lambda_{\tilde{\rho}}$ for $\tilde{\rho}=1$. Next, we
show $\tilde{l}\in\Lambda_{\tilde{\rho}}$ for $\tilde{\rho}\in(p_{1}^{(1)}+p_{0}^{(1)}\rho^{(2)},1)$.
First, \eqref{eq:ltildedef-1-1} implies 
\begin{align}
\tilde{l}(t) & =l^{(1)}(t)+(1-l^{(1)}(t))l^{(2)}(t),\quad t\in\mathbb{N}_{0}.\label{eq:or-alternative-representation}
\end{align}
Let $\tilde{L}(k)\triangleq\sum_{t=0}^{k-1}\tilde{l}(t)$. It follows
from \eqref{eq:or-alternative-representation} that 
\begin{align}
\tilde{L}(k) & =\sum_{t=0}^{k-1}l^{(1)}(t)+\sum_{t=0}^{k-1}(1-l^{(1)}(t))l^{(2)}(t),\quad k\in\mathbb{N}.\label{eq:Linusefulform}
\end{align}
Now, let $\epsilon\triangleq\tilde{\rho}-p_{1}^{(1)}-p_{0}^{(1)}\rho^{(2)}$,
$\epsilon_{2}\triangleq\min\{\frac{\epsilon}{2},\frac{\rho^{(2)}-p_{0}^{(1)}\rho^{(2)}}{2}\}$,
$\epsilon_{1}\triangleq\epsilon-\epsilon_{2}$, and define $\tilde{\rho}_{1}\triangleq p_{1}^{(1)}+\epsilon_{1}$,
$\tilde{\rho}_{2}\triangleq p_{0}^{(1)}\rho^{(2)}+\epsilon_{2}$.
Furthermore, let $\tilde{L}_{1}(k)\triangleq\sum_{t=0}^{k-1}l^{(1)}(t)$
and $\tilde{L}_{2}(k)\triangleq\sum_{t=0}^{k-1}(1-l^{(1)}(t))l^{(2)}(t)$.
We then have 
\begin{align}
\mathbb{P}[\tilde{L}(k)>\tilde{\rho}k] & =\mathbb{P}[\tilde{L}_{1}(k)+\tilde{L}_{2}(k)>\tilde{\rho}_{1}k+\tilde{\rho}_{2}k]\nonumber \\
 & \leq\mathbb{P}\left[\left\{ \tilde{L}_{1}(k)>\tilde{\rho}_{1}k\right\} \cup\left\{ \tilde{L}_{2}(k)>\tilde{\rho}_{2}k\right\} \right]\nonumber \\
 & \leq\mathbb{P}[\tilde{L}_{1}(k)>\tilde{\rho}_{1}k]+\mathbb{P}[\tilde{L}_{2}(k)>\tilde{\rho}_{2}k].\label{eq:keyprobabilityinequality}
\end{align}
In the following we show $\sum_{k=1}^{\infty}\mathbb{P}[\tilde{L}_{1}(k)>\tilde{\rho}_{1}k]<\infty$
and $\sum_{k=1}^{\infty}\mathbb{P}[\tilde{L}_{2}(k)>\tilde{\rho}_{2}k]<\infty$. 

First, note that 
\begin{align}
\tilde{\rho}_{1} & =p_{1}^{(1)}+\epsilon-\epsilon_{2}=\max\{p_{1}^{(1)}+\frac{\epsilon}{2},p_{1}^{(1)}+\epsilon-\frac{\rho^{(2)}-p_{0}^{(1)}\rho^{(2)}}{2}\}\nonumber \\
 & =\max\{\frac{p_{1}^{(1)}+\tilde{\rho}-p_{0}^{(1)}\rho^{(2)}}{2},\tilde{\rho}-p_{0}^{(1)}\rho^{(2)}-\frac{\rho^{(2)}-p_{0}^{(1)}\rho^{(2)}}{2}\}\nonumber \\
 & =\max\{\frac{p_{1}^{(1)}+\tilde{\rho}-p_{0}^{(1)}\rho^{(2)}}{2},\frac{2\tilde{\rho}-\rho^{(2)}(1+p_{0}^{(1)})}{2}\}.\label{eq:rho1ineq}
\end{align}
As $\frac{p_{1}^{(1)}+\tilde{\rho}-p_{0}^{(1)}\rho^{(2)}}{2}<1$ and
$\frac{2\tilde{\rho}-\rho^{(2)}(1+p_{0}^{(1)})}{2}<1$, it holds from
\eqref{eq:rho1ineq} that $\tilde{\rho}_{1}\in(p_{1}^{(1)},1)$. Since
$l^{(1)}\in\Gamma_{p_{0}^{(1)},p_{1}^{(1)}}$, we can use Proposition~\ref{PropositionGammaLambdaRelation}
with $\rho$ replaced with $\tilde{\rho}$ and $l$ replaced with
$l^{(1)}$ to obtain $\sum_{k=1}^{\infty}\mathbb{P}[\tilde{L}_{1}(k)>\tilde{\rho}_{1}k]<\infty$. 

Next, we use Lemma~\ref{KeyMarkovLemma} to show that $\sum_{k=1}^{\infty}\mathbb{P}[\tilde{L}_{2}(k)>\tilde{\rho}_{2}k]<\infty$.
To obtain this result, we first observe that $\tilde{\rho}_{2}>p_{0}^{(1)}\rho^{(2)}$,
since $\epsilon_{2}>0$. Moreover, 
\begin{align*}
\tilde{\rho}_{2} & =p_{0}^{(1)}\rho^{(2)}+\min\{\frac{\epsilon}{2},\frac{\rho^{(2)}-p_{0}^{(1)}\rho^{(2)}}{2}\}\leq p_{0}^{(1)}\rho^{(2)}+\frac{\rho^{(2)}-p_{0}^{(1)}\rho^{(2)}}{2}\\
 & <p_{0}^{(1)}\rho^{(2)}+\rho^{(2)}-p_{0}^{(1)}\rho^{(2)}=\rho^{(2)},
\end{align*}
and hence, we have $\tilde{\rho}_{2}\in(p_{0}^{(1)}\rho^{(2)},\rho^{(2)})$.
Let $\{l_{\mathrm{c}}^{(1)}(t)\in\{0,1\}\}_{t\in\mathbb{N}_{0}}$
be defined by $l_{\mathrm{c}}^{(1)}(t)=1-l^{(1)}(t),$ $t\in\mathbb{N}_{0}$.
Since $l^{(1)}\in\Gamma_{p_{0}^{(1)},p_{1}^{(1)}}$, we have $l_{\mathrm{c}}^{(1)}\in\Gamma_{p_{1}^{(1)},p_{0}^{(1)}}$.
Furthermore, since $l^{(2)}\in\Lambda_{\rho^{(2)}}$, conditions \eqref{eq:xicond},
\eqref{eq:chicond} in the Lemma~\ref{KeyMarkovLemma} hold with
$\tilde{p}=p_{0}^{(1)}$ and $\tilde{w}=\rho^{(2)}$, together with
processes $\{\xi(t)\in\{0,1\}\}_{t\in\mathbb{N}_{0}}$ and $\{\chi(t)\in\{0,1\}\}_{t\in\mathbb{N}_{0}}$
defined by setting $\xi(t)=\theta^{l_{\mathrm{c}}}(t)$, $\chi(t)=l_{\mathrm{M}}(t)$,
$t\in\mathbb{N}_{0}$. Now, we have $\tilde{L}_{2}(k)=\sum_{t=0}^{k-1}\xi(t)\chi(t)$
and hence, Lemma~\ref{KeyMarkovLemma} implies $\sum_{k=1}^{\infty}\mathbb{P}[\tilde{L}_{2}(k)>\tilde{\rho}_{2}k]<\infty$. 

Finally, by \eqref{eq:keyprobabilityinequality}, we arrive at 
\begin{align*}
\sum_{k=1}^{\infty}\mathbb{P}[\tilde{L}(k)>\tilde{\rho}k] & \leq\sum_{k=1}^{\infty}\mathbb{P}[\tilde{L}_{1}(k)>\tilde{\rho}_{1}k]+\sum_{k=1}^{\infty}\mathbb{P}[\tilde{L}_{2}(k)>\tilde{\rho}_{2}k]<\infty,
\end{align*}
which shows that $\tilde{l}\in\Lambda_{\tilde{\rho}}$ for all $\tilde{\rho}\in(p_{1}^{(1)}+p_{0}^{(1)}\rho^{(2)},1]$.
\end{proof}

Theorem~\ref{Theorem-Gamma-Lambda-1} is concerned with $\vee$ operation
applied to a process $l^{(1)}(\cdot)$ from the hidden Markov model
class $\Gamma_{p_{0}^{(1)},p_{1}^{(1)}}$ and another process $l^{(2)}(\cdot)$
from the class $\Lambda_{\rho^{(2)}}$. It is shown that the $\vee$
operation results in a process $\tilde{l}$ that satisfies $\tilde{l}\in\Lambda_{\tilde{\rho}}$
for all $\tilde{\rho}\in(p_{1}^{(1)}+p_{0}^{(1)}\rho^{(2)},1]$. Notice
that the application of Proposition~\ref{Proposition-and-operation}
to this situation would allow us to show $\tilde{l}\in\Lambda_{\tilde{\rho}}$
for all $\tilde{\rho}\in(p_{1}^{(1)}+\rho^{(2)},1]$. Proposition~\ref{Proposition-lp-or}
in this case is conservative since $p_{1}^{(1)}+\rho^{(2)}\geq p_{1}^{(1)}+p_{0}^{(1)}\rho^{(2)}$
(and $p_{1}^{(1)}+\rho^{(2)}>p_{1}^{(1)}+p_{0}^{(1)}\rho^{(2)}$ if
$p_{0}^{(1)}<1$). We remark that the advantage of Proposition~\ref{Proposition-lp-or}
may be that it allows us to deal with processes that are not mutually
independent. 

\begin{table}[t]
\caption{Comparison of the classes of processes obtained by combining processes
$l^{(1)}$ and $l^{(2)}$ of different classes through $\vee$ operation.
For the case where $l^{(1)}\in\Lambda_{\rho^{(1)}}$ and $l^{(2)}\in\Lambda_{\rho^{(2)}}$,
the processes $l^{(1)}$ and $l^{(2)}$ can be dependent; for other
cases, $l^{(1)}$ and $l^{(2)}$ are assumed to be mutually independent. }
\label{TableForOr} 

\renewcommand{\arraystretch}{1.2} 

{\centering \fontsize{7}{11.52}\selectfont 

\begin{tabular}{|c|c||c|}
\hline 
$l^{(1)}$ & $l^{(2)}$ & $l^{(1)}\vee l^{(2)}$\tabularnewline
\hline 
\hline 
\multirow{3}{*}{$\Gamma_{p_{0}^{(1)},p_{1}^{(1)}}$} & \multirow{2}{*}{$\Gamma_{p_{0}^{(2)},p_{1}^{(2)}}$} & $\Gamma_{\tilde{p}_{0},\tilde{p}_{1}}$ with $\tilde{p}_{0}\triangleq p_{0}^{(1)}p_{0}^{(2)}$
and $\tilde{p}_{1}\triangleq\min\{p_{1}^{(1)}+p_{0}^{(1)}p_{1}^{(2)},p_{1}^{(2)}+p_{0}^{(2)}p_{1}^{(1)},1\}$\tabularnewline
 &  & (Theorem~\ref{Theorem-HiddenMarkov-1})\tabularnewline
\cline{2-3} 
 & $\Lambda_{\rho^{(2)}}$ & $\Lambda_{\tilde{\rho}}$ for $\tilde{\rho}\in(p_{1}^{(1)}+p_{0}^{(1)}\rho^{(2)},1]$
(Theorem~\ref{Theorem-Gamma-Lambda-1}) \tabularnewline
\hline 
\multirow{2}{*}{$\Lambda_{\rho^{(1)}}$} & $\Gamma_{p_{0}^{(2)},p_{1}^{(2)}}$ & $\Lambda_{\tilde{\rho}}$ for $\tilde{\rho}\in(p_{1}^{(2)}+p_{0}^{(2)}\rho^{(1)},1]$
(Theorem~\ref{Theorem-Gamma-Lambda-1}) \tabularnewline
\cline{2-3} 
 & $\Lambda_{\rho^{(2)}}$ & $\Lambda_{\tilde{\rho}}$ for $\tilde{\rho}\in[\rho^{(1)}+\rho^{(2)},1]$
(Proposition~\ref{Proposition-lp-or})\tabularnewline
\hline 
\end{tabular}

} 
\end{table}

The results presented so far in this section are summarized in Table~\ref{TableForOr}.
There we show the classes of processes obtained by using the $\vee$
operation. 
\begin{table}[t]
\caption{Comparison of the conservativeness of our results with respect to
the range of $\tilde{\rho}$ of the process $\tilde{l}\in\Lambda_{\tilde{\rho}}$
obtained by combining mutually independent processes $l^{(1)}\in\Gamma_{p_{0}^{(1)},p_{1}^{(1)}}$
and $l^{(2)}\in\Gamma_{p_{0}^{(2)},p_{1}^{(2)}}$ through $\wedge$
or $\vee$ operations (i.e., $\tilde{l}(t)=l^{(1)}(t)\wedge l^{(2)}(t)$
or $\tilde{l}(t)=l^{(1)}(t)\vee l^{(2)}(t)$). For each operation,
less conservative results are marked with $\star$ and they provide
larger ranges of $\tilde{\rho}$ (indicating that $\tilde{l}\in\Lambda_{\tilde{\rho}}$
holds with smaller $\tilde{\rho}$ values close to the lower bound
of the range). However, these less conservative results are applicable
only when $l^{(1)}$ and $l^{(2)}$ are mutually independent. The
sequences of results that are not marked with $\star$ are more conservative,
but they are also applicable to scenarios where $l^{(1)}$ and $l^{(2)}$
are not necessarily mutually independent. }
\label{TableConservatism}

\renewcommand{\arraystretch}{1.2} \vskip 5pt

{\centering \fontsize{7}{11.52}\selectfont  

\begin{tabular}{|c|c||c|}
\hline 
Operation & Applied Results & $\tilde{\rho}$ range\tabularnewline
\hline 
\hline 
\multirow{2}{*}{$\wedge$} & $\star$ First Theorem~\ref{Theorem-HiddenMarkov}, then Proposition~\ref{PropositionGammaLambdaRelation}  & $(p_{1}^{(1)}p_{1}^{(2)},1]$ \tabularnewline
\cline{2-3} 
 & First Proposition~\ref{PropositionGammaLambdaRelation}, then Proposition~\ref{Proposition-and-operation} & $(\min\{p_{1}^{(1)},p_{1}^{(2)}\},1]$\tabularnewline
\hline 
\multirow{2}{*}{$\vee$} & $\star$ First Theorem~\ref{Theorem-HiddenMarkov-1}, then Proposition~\ref{PropositionGammaLambdaRelation}  & $(\min\{p_{1}^{(1)}+p_{0}^{(1)}p_{1}^{(2)},p_{1}^{(2)}+p_{0}^{(2)}p_{1}^{(1)}\},1]$ \tabularnewline
\cline{2-3} 
 & First Proposition~\ref{PropositionGammaLambdaRelation}, then Proposition~\ref{Proposition-lp-or} & $(p_{1}^{(1)}+p_{1}^{(2)},1]$\tabularnewline
\hline 
\end{tabular}

} 
\end{table}

\begin{remark} By utilizing \eqref{eq:lg-multiple-paths} and \eqref{eq:lp-or}
together with the results so far presented, we can obtain $\rho$
values for which the overall packet exchange failure indicator $l$
satisfies $l\in\Lambda_{\rho}$. As discussed in Section~\ref{sec:Control-over-Multi-hop},
$l\in\Lambda_{\rho}$ implies that the average number of packet exchange
failures is upper-bounded by $\rho$ (i.e., $\limsup_{k\to\infty}\frac{1}{k}\sum_{t=0}^{k-1}l(t)\leq\rho,$
almost surely).\emph{ }We would like $\rho$ to be a tight upper bound
so that we can avoid additional conservatism in checking the closed-loop
stability with Theorem~\ref{Stability-Theorem}. As discussed in
Section~4 of \cite{cetinkaya-tac}, conditions \eqref{eq:betacond}--\eqref{eq:betaandvarphicond}
of Theorem~\ref{Stability-Theorem} are tight for scalar systems
when $\lim_{k\to\infty}\frac{1}{k}\sum_{t=0}^{k-1}l(t)=\rho$. However,
notice that in certain cases, $l$ may be a nonergodic process for
which $\lim_{k\to\infty}\frac{1}{k}\sum_{t=0}^{k-1}l(t)$ takes different
values for different outcomes $\omega\in\Omega$ (see Remark 3.4 in
\cite{cetinkaya-tac}); moreover, in certain cases $\lim_{k\to\infty}\frac{1}{k}\sum_{t=0}^{k-1}l(t)$
may not exist. In such cases $\rho$ (as an upper bound of $\limsup_{k\to\infty}\frac{1}{k}\sum_{t=0}^{k-1}l(t)$)
is still useful in stability analysis through Theorem~\ref{Stability-Theorem}. 

In some cases, $\limsup_{k\to\infty}\frac{1}{k}\sum_{t=0}^{k-1}l(t)$
may be strictly smaller than $\rho$ that we obtain by using our results
in this paper. The gap between $\limsup_{k\to\infty}\frac{1}{k}\sum_{t=0}^{k-1}l(t)$
and $\rho$ can be identified if all link failure model parameters
and independence/dependence relations between communication links
are known. Note that it may be difficult to obtain an analytical expression
for $\limsup_{k\to\infty}\frac{1}{k}\sum_{t=0}^{k-1}l(t)$ even if
the properties of the failure processes are known exactly. This is
because such properties may be time-dependent and may have complicated
inter-dependence relations. We compare $\limsup_{k\to\infty}\frac{1}{k}\sum_{t=0}^{k-1}l(t)$
and $\rho$ numerically through repeated simulations for an example
case in Section~\ref{subsec:State-dependent-Attacks-by}. To guarantee
a small gap between $\limsup_{k\to\infty}\frac{1}{k}\sum_{t=0}^{k-1}l(t)$
and $\rho$, failure processes need to be placed in adequate classes
$\Gamma_{p_{0},p_{1}}$ and $\Pi_{\kappa,w}$ so that the bounds utilized
in Definitions~\ref{Definition-Hidden-Markov} and \ref{DefinitionPi}
are sufficiently tight. Furthermore, to keep the gap small, it is
also essential to utilize Theorems~\ref{Theorem-HiddenMarkov}, \ref{Theorem-Gamma-Lambda},
\ref{Theorem-HiddenMarkov-1}, and \ref{Theorem-Gamma-Lambda-1} when
the processes involved in $\wedge$ and $\vee$ operations are known
to be mutually independent and at least one of them is from the class
$\Gamma$. If, instead, Propositions~\ref{Proposition-and-operation}
and \ref{Proposition-lp-or} are used, this may introduce additional
gap between $\limsup_{k\to\infty}\frac{1}{k}\sum_{t=0}^{k-1}l(t)$
and $\rho$. This is because, in scenarios where random failures on
links/paths occur independently, Propositions~\ref{Proposition-and-operation}
and \ref{Proposition-lp-or} are more conservative. In Table~\ref{TableConservatism},
we compare these results with Theorems~\ref{Theorem-HiddenMarkov}
and \ref{Theorem-HiddenMarkov-1} in terms of the conservatism that
they may introduce. \end{remark} 

\subsection{Characterization of Paths with Packet Dropping Links}

In the previous section, we provided a characterization for paths
with data-corrupting links. In this section we look at packet drops. 

A packet dropout on a communication link $(v,w)$ may occur if $v$
is a malicious router that intentionally skips forwarding packets
(see blackhole and grayhole attacks in \cite{jhaveri2012attacks}).
A nonmalicious router may also drop packets to avoid congestion. In
addition to these two issues, a packet may also be dropped if the
header part of the packet, which includes information on the destination
of the packet, is corrupted. Furthermore, in scenarios where error-detection
is implemented at intermediate nodes, corruption on the data part
of a packet can be detected. As a result, a corrupted packet needs
not be further transmitted. Such a scenario can also be studied within
the packet drop framework. 

Consider a link $\mathcal{P}_{i,j}=(v,w)$ on path $\mathcal{P}_{i}$.
Let $t=0,1,\ldots$, denote the indices of packets that node $v$
possesses (or receives from previous nodes on path $\mathcal{P}_{i}$).
Observe that if there are links before $\mathcal{P}_{i,j}$ that are
packet-dropping, then the first packet (with index $t=0$) that $v$
receives may be different from $x(0)$, and it may be the state at
a later time. This is because $x(0)$ may have been dropped before
reaching node $v$. We use $l_{\mathcal{P}_{i,j}}(t)\in\{0,1\}$ to
indicate the status of transmission of the $(t+1)$th packet (with
index $t$) that node $v$ possesses to node $w$. 

The following result is concerned with the characterization of the
failures on a path with only packet dropping links. 

\begin{proposition}\label{Theorem-lp-all-packet-dropping} Let $\{l_{\mathcal{P}_{i}}(t)\in\{0,1\}\}_{t\in\mathbb{N}_{0}}$
denote the failure indicator of a path $\mathcal{P}_{i}$ composed
only of packet dropping links with failure indicators denoted by $\{l_{\mathcal{P}_{i}}^{\mathcal{P}_{i,j}}(t)\in\{0,1\}\}_{t\in\mathbb{N}_{0}}$.
Assume $l_{\mathcal{P}_{i}}^{\mathcal{P}_{i,j}}\in\Lambda_{\rho_{\mathcal{P}_{i}}^{\mathcal{P}_{i,j}}}$
with $\rho_{\mathcal{P}_{i}}^{\mathcal{P}_{i,j}}\in[0,1]$, $j\in\{1,\ldots,|\mathcal{P}_{i}|\}$,
that satisfy $\sum_{j=1}^{|\mathcal{P}_{i}|}\rho_{\mathcal{P}_{i}}^{\mathcal{P}_{i,j}}\leq1$.
Then $l_{\mathcal{P}_{i}}\in\Lambda_{\rho_{\mathcal{P}_{i}}}$ with
$\rho_{\mathcal{P}_{i}}\triangleq\sum_{j=1}^{|\mathcal{P}_{i}|}\rho_{\mathcal{P}_{i}}^{\mathcal{P}_{i,j}}$. 

\end{proposition}

The proof of this result is skipped since it is similar to that of
Theorem~\ref{Theorem-lp-general} provided in the next section, where
we consider paths that include data corrupting and packet dropping
communication links. 

\subsection{A General Characterization of Paths with Data Corruption and Packet
Drops}

In this section we investigate the effects of both data corruption
and packet dropouts. Without loss of generality, we assume that links
on a path are either \emph{data-corrupting} or \emph{packet-dropping},
but not both. Note that if in the original network, a link $(v,w)$
is subject to both of the issues, we can artificially add a node $v'$
and edges $(v,v')$, $(v',w)$ to the graph, and consider $(v,v')$
as a packet-dropping link and $(v',w)$ as a data-corrupting link. 

For a data-corrupting link $(v,w)$, packets available at $v$ are
always transmitted to $w$, but their content may be externally manipulated
or damaged during the transmission over this link. On the other hand,
if $(v,w)$ is a packet-dropping link, packets available at $v$ may
or may not be received by $w$, but never get corrupted on the communication
link $(v,w)$. 

Our goal here is to obtain a relation between the asymptotic packet
failure ratio of path $\mathcal{P}_{i}$ and the failure ratios of
the links on that path. To this end, we will use a recursive characterization
for describing packet failures on paths. Specifically, consider a
path 
\begin{align}
\mathcal{P} & \triangleq\big((v_{1},v_{2}),(v_{2},v_{3}),\ldots,(v_{h},v_{h+1})\big)\label{eq:path-p-def}
\end{align}
 of $h\geq1$ links, and consider the associated process $\{l_{\mathcal{P}}(t)\in\{0,1\}\}_{t\in\mathbb{N}_{0}}$.
The state $l_{\mathcal{P}}(t)=0$ indicates that the $(t+1)$th packet
that the first node $v_{1}$ possesses can be successfully transmitted
to the last node $v_{h+1}$, whereas $l_{\mathcal{P}}(t)=1$ indicates
a failure. 

\vskip 2pt

If $h=1$ in \eqref{eq:path-p-def}, then we have $l_{\mathcal{P}}(t)=l_{\mathcal{P}}^{(v_{1},v_{2})}(t)$,
$t\in\mathbb{N}_{0}$. Now consider the case $h\geq2$, and let $f(\mathcal{P})$
and $\mathcal{R}(\mathcal{P})$ respectively denote the first link
on $\mathcal{P}$, and the \emph{subpath} composed of the rest of
the links, that is, 
\begin{align}
f(\mathcal{P}) & =(v_{1},v_{2}),\label{eq:path-r-def}\\
\mathcal{R}(\mathcal{P}) & =\big((v_{2},v_{3}),\ldots,(v_{h},v_{h+1})\big).
\end{align}
We illustrate $f(\mathcal{P})$ and $\mathcal{R}(\mathcal{P})$ on
the left side of Fig.~\ref{Flo:lr-lp-relation}.

\vskip 2pt

Next, we show that transmission failures on a path $\mathcal{P}$
can be characterized through transmission failures on the link $f(\mathcal{P})$
and the subpath $\mathcal{R}(\mathcal{P})$. Let $\{l_{\mathcal{P}}^{f(\mathcal{P})}(t)\in\{0,1\}\}_{t\in\mathbb{N}_{0}}$
and $\{l_{\mathcal{P}}^{\mathcal{R}(\mathcal{P})}(t)\in\{0,1\}\}_{t\in\mathbb{N}_{0}}$
denote indicators of transmission failures on the link $f(\mathcal{P})$
and the subpath $\mathcal{R}(\mathcal{P})$. If the link $f(\mathcal{P})=(v_{1},v_{2})$
is a data-corrupting link, then we would have $l_{\mathcal{P}}(t)=l_{\mathcal{P}}^{f(\mathcal{P})}(t)\vee l_{\mathcal{P}}^{\mathcal{R}(\mathcal{P})}(t)$.
But this relation does not hold if $f(\mathcal{P})$ is a packet-dropping
link, because the index $t$ for $l_{\mathcal{P}}^{f(\mathcal{P})}(t)$
and $l_{\mathcal{P}}^{\mathcal{R}(\mathcal{P})}(t)$ represent different
packets. 

Now, we will introduce a new process $\{\widehat{l}_{\mathcal{P}}^{\mathcal{R}(\mathcal{P})}(t)\in\{0,1\}\}_{t\in\mathbb{N}_{0}}$
for which 
\begin{align}
l_{\mathcal{P}}(t) & =l_{\mathcal{P}}^{f(\mathcal{P})}(t)\vee\widehat{l}_{\mathcal{P}}^{\mathcal{R}(\mathcal{P})}(t),\quad t\in\mathbb{N}_{0}.\label{eq:p-fp-rp}
\end{align}
 If $f(\mathcal{P})$ is a data-corrupting link, then we define $\widehat{l}_{\mathcal{P}}^{\mathcal{R}(\mathcal{P})}(\cdot)$
by setting $\widehat{l}_{\mathcal{P}}^{\mathcal{R}(\mathcal{P})}(t)=l_{\mathcal{P}}^{\mathcal{R}(\mathcal{P})}(t)$,
$t\in\mathbb{N}_{0}$. On the other hand, if $f(\mathcal{P})$ is
a packet-dropping link, then for $t\in\mathbb{N}_{0}$, 
\begin{align}
\widehat{l}_{\mathcal{P}}^{\mathcal{R}(\mathcal{P})}(t) & \triangleq\begin{cases}
0,\,\, & \mathrm{if}\,\,l_{\mathcal{P}}^{f(\mathcal{P})}(t)=1,\\
l_{\mathcal{P}}^{\mathcal{R}(\mathcal{P})}\big(\hat{k}(t+1)-1\big),\,\, & \mathrm{if}\,\,l_{\mathcal{P}}^{f(\mathcal{P})}(t)=0,
\end{cases}\label{eq:overline-lr-def}
\end{align}
where $\hat{k}(t)\triangleq\sum_{i=0}^{t-1}(1-l_{\mathcal{P}}^{f(\mathcal{P})}(i))$,
$t\in\mathbb{N}$. In this definition, $\hat{k}(t+1)$ denotes the
number of packets that are successfully transmitted from node $v_{1}$
to node $v_{2}$ among the first $t+1$ packets that node $v_{1}$
possesses. Hence, the scalar $\hat{k}(t+1)-1$ represents the index
of the $\hat{k}(t+1)$th packet received by $v_{2}$. Moreover, $l_{\mathcal{P}}^{\mathcal{R}(\mathcal{P})}\big(\hat{k}(t+1)-1\big)$
indicates whether this packet is successfully transmitted from $v_{2}$
over $\mathcal{R}(\mathcal{P})$ to $v_{h+1}$. 

\vskip 2pt

Observe that by \eqref{eq:overline-lr-def}, $\widehat{l}_{\mathcal{P}}^{\mathcal{R}(\mathcal{P})}(t)$
is set to $0$, if $v_{1}$ drops the $(t+1)$th packet that it possesses.
On the other hand, if $v_{1}$ transmits this packet to $v_{2}$,
the state $\widehat{l}_{\mathcal{P}}^{\mathcal{R}(\mathcal{P})}(t)=1$
indicates that further transmission of this packet on subpath $\mathcal{R}(\mathcal{P})$
has failed, whereas $\widehat{l}_{\mathcal{P}}^{\mathcal{R}(\mathcal{P})}(t)=0$
indicates success. As a result, we have $l_{\mathcal{P}}(t)=l_{\mathcal{P}}^{f(\mathcal{P})}(t)+(1-l_{\mathcal{P}}^{f(\mathcal{P})}(t))\widehat{l}_{\mathcal{P}}^{\mathcal{R}(\mathcal{P})}(t)$,
and hence, \eqref{eq:p-fp-rp} holds. 

\vskip 1pt

The characterization above is recursive in the sense that it is recursively
applied to describe failures on $\mathcal{R}(\mathcal{P})$ by means
of failures on the first link $f(\mathcal{R}(\mathcal{P}))$ and the
subpath $\mathcal{R}(\mathcal{R}(\mathcal{P}))$. 

\vskip 3pt

\begin{figure}
\centering  \includegraphics[width=1\columnwidth]{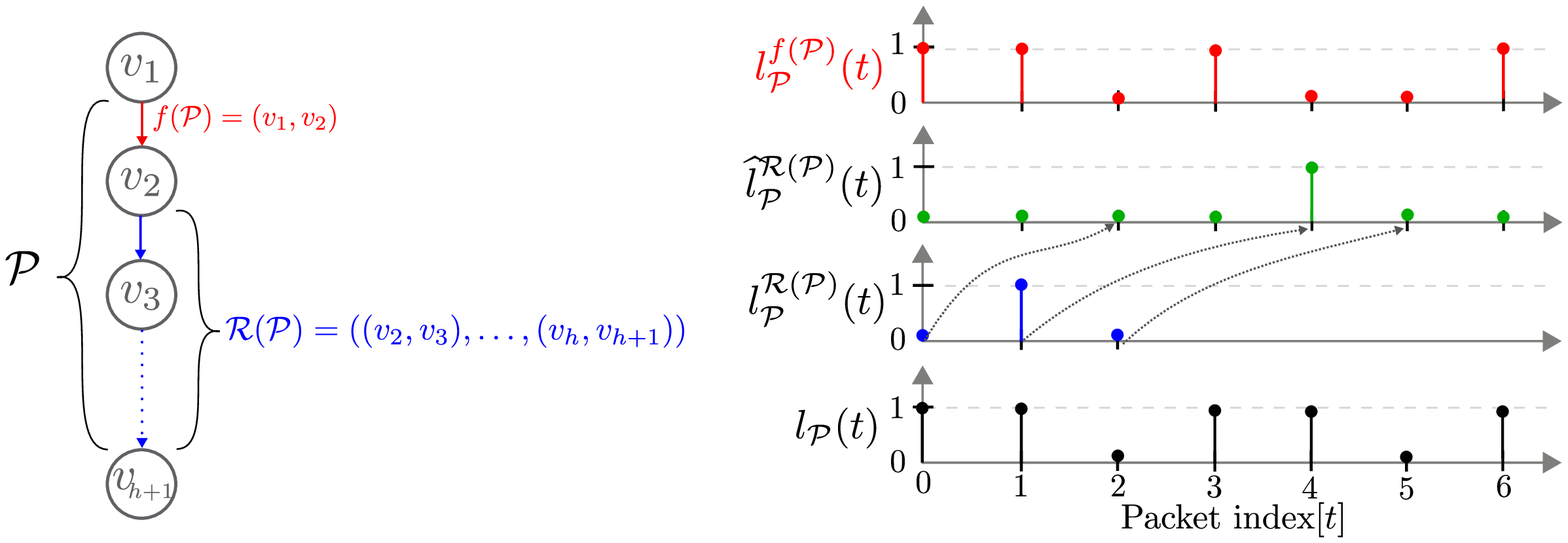}
\vskip -3pt

\caption{\fontsize{9}{9}\selectfont{{[}Left{]} Paths $\mathcal{P}$, $\mathcal{R}(\mathcal{P})$,
and packet-dropping link $f(\mathcal{P})$. {[}Right{]} Trajectories
of $l_{\mathcal{P}}^{f(\mathcal{P})}(\cdot)$, $\widehat{l}_{\mathcal{P}}^{\mathcal{R}(\mathcal{P})}(\cdot)$,
$l_{\mathcal{P}}^{\mathcal{R}(\mathcal{P})}(\cdot)$ and $l_{\mathcal{P}}(\cdot)$.} }
 \label{Flo:lr-lp-relation}
\end{figure}

\begin{example} In Fig.~\ref{Flo:lr-lp-relation}, we see a path
$\mathcal{P}$, its packet-dropping first link $f(\mathcal{P})=(v_{1},v_{2})$,
and its subpath $\mathcal{R}(\mathcal{P})$, together with sample
trajectories of $l_{\mathcal{P}}^{f(\mathcal{P})}(\cdot)$, $\widehat{l}_{\mathcal{P}}^{\mathcal{R}(\mathcal{P})}(\cdot)$,
$l_{\mathcal{P}}^{\mathcal{R}(\mathcal{P})}(\cdot)$ and $l_{\mathcal{P}}(\cdot)$.
Note that the link $f(\mathcal{P})=(v_{1},v_{2})$ drops the first
two packets, that is, $l_{\mathcal{P}}^{f(\mathcal{P})}(t)=1$ for
$t\in\{0,1\}$. By \eqref{eq:overline-lr-def}, we have $\widehat{l}_{\mathcal{P}}^{\mathcal{R}(\mathcal{P})}(0)=\widehat{l}_{\mathcal{P}}^{\mathcal{R}(\mathcal{P})}(1)=0$,
and as a consequence of \eqref{eq:p-fp-rp}, $l_{\mathcal{P}}(0)=l_{\mathcal{P}}(1)=1$.
Furthermore, the link $f(\mathcal{P})$ successfully transmits the
$3$rd packet, i.e, $l_{\mathcal{P}}^{f(\mathcal{P})}(2)=0$. The
value of $\widehat{l}_{\mathcal{P}}^{\mathcal{R}(\mathcal{P})}(2)$
represents whether this packet is successfully transmitted over the
subpath $\mathcal{R}(\mathcal{P})$ or not. Since this packet is the
first packet that $v_{2}$ receives, its transmission state is represented
by $l_{\mathcal{R}(\mathcal{P})}(0)$, and as a result, $\widehat{l}_{\mathcal{P}}^{\mathcal{R}(\mathcal{P})}(2)=l_{\mathcal{P}}^{\mathcal{R}(\mathcal{P})}(0)$.
We have $\widehat{l}_{\mathcal{P}}^{\mathcal{R}(\mathcal{P})}(2)=l_{\mathcal{P}}^{\mathcal{R}(\mathcal{P})}(0)=0$,
indicating successful transmission over $\mathcal{R}(\mathcal{P})$.
Now, by \eqref{eq:p-fp-rp}, $l_{\mathcal{P}}(2)=0$, which indicates
that the $3$rd packet possessed by the first node $v_{1}$ is successfully
transmitted to the last node $v_{h+1}$. \end{example}

The following result is essential for characterizing the failures
on path $\mathcal{P}$ through those on the first link $f(\mathcal{P})$
and the subpath $\mathcal{R}(\mathcal{P})$. We obtain the result
by first establishing the key inequality $\sum_{t=0}^{k-1}\widehat{l}_{\mathcal{P}}^{\mathcal{R}(\mathcal{P})}(t)\leq\sum_{t=0}^{k-1}l_{\mathcal{P}}^{\mathcal{R}(\mathcal{P})}(t)$,
and then applying some of the arguments that we employed for proving
Proposition~\ref{Proposition-lp-or}. 

\begin{lemma} \label{Lemma-recursive} Consider path $\mathcal{P}$
given in \eqref{eq:path-p-def} for $h\geq1$. Assume $l_{\mathcal{P}}^{f(\mathcal{P})}\in\Lambda_{\rho_{\mathcal{P}}^{f(\mathcal{P})}}$
and $l_{\mathcal{P}}^{\mathcal{R}(\mathcal{P})}\in\Lambda_{\rho_{\mathcal{P}}^{\mathcal{R}(\mathcal{P})}}$
with scalars $\rho_{\mathcal{P}}^{f(\mathcal{P})},\rho_{\mathcal{P}}^{\mathcal{R}(\mathcal{P})}\in[0,1]$,
such that $\rho_{\mathcal{P}}^{f(\mathcal{P})}+\rho_{\mathcal{P}}^{\mathcal{R}(\mathcal{P})}\leq1$.
It follows that $l_{\mathcal{P}}\in\Lambda_{\rho_{\mathcal{P}}}$
with $\rho_{\mathcal{P}}=\rho_{\mathcal{P}}^{f(\mathcal{P})}+\rho_{\mathcal{P}}^{\mathcal{R}(\mathcal{P})}$. 

\end{lemma}

\begin{proof} We will first show that $\widehat{l}_{\mathcal{P}}^{\mathcal{R}(\mathcal{P})}\in\Lambda_{\rho_{\mathcal{P}}^{\mathcal{R}(\mathcal{P})}}$.
If $f(\mathcal{P})$ is a data-corrupting link, we have $\widehat{l}_{\mathcal{P}}^{\mathcal{R}(\mathcal{P})}(t)=l_{\mathcal{P}}^{\mathcal{R}(\mathcal{P})}(t)$,
$t\in\mathbb{N}_{0}$, and therefore $\widehat{l}_{\mathcal{P}}^{\mathcal{R}(\mathcal{P})}\in\Lambda_{\rho_{\mathcal{P}}^{\mathcal{R}(\mathcal{P})}}$.
Now consider the situation where $f(\mathcal{P})$ is a packet-dropping
link. For this case, 
\begin{align*}
\sum_{t=0}^{k-1}\widehat{l}_{\mathcal{P}}^{\mathcal{R}(\mathcal{P})}(t) & =\begin{cases}
0,\quad & \hat{k}(k)=0,\\
\sum_{t=0}^{\hat{k}(k)-1}l_{\mathcal{P}}^{\mathcal{R}(\mathcal{P})}(t),\quad & \hat{k}(k)\neq0,
\end{cases}\quad k\in\mathbb{N}.
\end{align*}
Since $\hat{k}(k)\leq k$ and $l_{\mathcal{R}(\mathcal{P})}(t)\geq0$,
we have $\sum_{t=0}^{\hat{k}(k)-1}l_{\mathcal{P}}^{\mathcal{R}(\mathcal{P})}(t)\leq\sum_{t=0}^{k-1}l_{\mathcal{P}}^{\mathcal{R}(\mathcal{P})}(t)$,
and hence, $\sum_{t=0}^{k-1}\widehat{l}_{\mathcal{P}}^{\mathcal{R}(\mathcal{P})}(t)\leq\sum_{t=0}^{k-1}l_{\mathcal{P}}^{\mathcal{R}(\mathcal{P})}(t)$.
It follows that 
\begin{align*}
\mathbb{P}[\sum_{t=0}^{k-1}\widehat{l}_{\mathcal{P}}^{\mathcal{R}(\mathcal{P})}(t)>\rho_{\mathcal{P}}^{\mathcal{R}(\mathcal{P})}k] & \leq\mathbb{P}[\sum_{t=0}^{k-1}l_{\mathcal{P}}^{\mathcal{R}(\mathcal{P})}(t)>\rho_{\mathcal{P}}^{\mathcal{R}(\mathcal{P})}k].
\end{align*}
Now, as $l_{\mathcal{P}}^{\mathcal{R}(\mathcal{P})}\in\Lambda_{\rho_{\mathcal{P}}^{\mathcal{R}(\mathcal{P})}}$,
we also have $\widehat{l}_{\mathcal{P}}^{\mathcal{R}(\mathcal{P})}\in\Lambda_{\rho_{\mathcal{P}}^{\mathcal{R}(\mathcal{P})}}$. 

Finally, since \eqref{eq:p-fp-rp} holds, we have $l_{\mathcal{P}}(t)\leq l_{\mathcal{P}}^{f(\mathcal{P})}(t)+\widehat{l}_{\mathcal{P}}^{\mathcal{R}(\mathcal{P})}(t),$
$t\in\mathbb{N}_{0}$. Hence, 
\begin{align}
 & \mathbb{P}[\sum_{t=0}^{k-1}l_{\mathcal{P}}(t)>\rho_{\mathcal{P}}k]\leq\mathbb{P}[\sum_{t=0}^{k-1}l_{\mathcal{P}}^{f(\mathcal{P})}(t)+\sum_{t=0}^{k-1}\widehat{l}_{\mathcal{P}}^{\mathcal{R}(\mathcal{P})}(t)>\rho_{\mathcal{P}}k]\nonumber \\
 & \,=\mathbb{P}[\sum_{t=0}^{k-1}l_{\mathcal{P}}^{f(\mathcal{P})}(t)+\sum_{t=0}^{k-1}\widehat{l}_{\mathcal{P}}^{\mathcal{R}(\mathcal{P})}(t)>(\rho_{\mathcal{P}}^{f(\mathcal{P})}+\rho_{\mathcal{P}}^{\mathcal{R}(\mathcal{P})})k]\nonumber \\
 & \,\leq\mathbb{P}[\sum_{t=0}^{k-1}l_{\mathcal{P}}^{f(\mathcal{P})}(t)>\rho_{\mathcal{P}}^{f(\mathcal{P})}k]+\mathbb{P}[\sum_{t=0}^{k-1}\widehat{l}_{\mathcal{P}}^{\mathcal{R}(\mathcal{P})}(t)>\rho_{\mathcal{P}}^{\mathcal{R}(\mathcal{P})}k].\label{eq:probbound}
\end{align}
It then follows from \eqref{eq:probbound} together with $\sum_{k=1}^{\infty}\mathbb{P}[\sum_{t=0}^{k-1}l_{\mathcal{P}}^{f(\mathcal{P})}(t)>\rho_{\mathcal{P}}^{f(\mathcal{P})}k]<\infty$
and $\sum_{k=1}^{\infty}\mathbb{P}[\sum_{t=0}^{k-1}\widehat{l}_{\mathcal{P}}^{\mathcal{R}(\mathcal{P})}(t)>\rho_{\mathcal{P}}^{\mathcal{R}(\mathcal{P})}k]<\infty$
that $\sum_{k=1}^{\infty}\mathbb{P}[\sum_{t=0}^{k-1}l_{\mathcal{P}}(t)>\rho_{\mathcal{P}}k]<\infty,$
which completes the proof. \end{proof}

Now, for a given path $\mathcal{P}_{i}$, repeated application of
Lemma~\ref{Lemma-recursive} to $\mathcal{P}_{i}$, $\mathcal{R}(\mathcal{P}_{i})$,
$\mathcal{R}(\mathcal{R}(\mathcal{P}_{i}))$, $\ldots$, leads us
to the following result. \vskip 5pt

\begin{theorem}\label{Theorem-lp-general} Assume $l_{\mathcal{P}_{i}}^{\mathcal{P}_{i,j}}\in\Lambda_{\rho_{\mathcal{P}_{i}}^{\mathcal{P}_{i,j}}}$
with $\rho_{\mathcal{P}_{i}}^{\mathcal{P}_{i,j}}\in[0,1]$, $j\in\{1,\ldots,|\mathcal{P}_{i}|\}$,
that satisfy $\sum_{j=1}^{|\mathcal{P}_{i}|}\rho_{\mathcal{P}_{i}}^{\mathcal{P}_{i,j}}\leq1$.
Then $l_{\mathcal{P}_{i}}\in\Lambda_{\rho_{\mathcal{P}_{i}}}$ with
$\rho_{\mathcal{P}_{i}}\triangleq\sum_{j=1}^{|\mathcal{P}_{i}|}\rho_{\mathcal{P}_{i}}^{\mathcal{P}_{i,j}}$. 

\end{theorem}

Note that Theorem~\ref{Theorem-lp-general} allows us to consider
both data-corruption and packet-dropouts on links, and hence it generalizes
Proposition~\ref{Proposition-lp-or}. 

\begin{remark} \label{Remark-General-Rho} By utilizing Theorem~\ref{Theorem-lp-general}
together with Proposition~\ref{Proposition-and-operation}, we can
obtain $\rho_{G},\rho_{\tilde{G}}\in[0,1]$ as upper-bounds for the
average number of packet exchange failures on networks $G$ and $\tilde{G}$
such that $l_{G}\in\Lambda_{\rho_{G}}$ and $l_{\tilde{G}}\in\Lambda_{\rho_{\tilde{G}}}$.
Note that when $l_{G}(t)=1$, then the controller either does not
receive the state packet or receives corrupted versions, which are
discarded. Hence, when $l_{G}(t)=1$, no control input packet is attempted
to be transmitted on $\tilde{G}$. This setting is similar to the
situation that we discussed above for packet dropping links. Here
we can consider the whole networked system as a path $\mathcal{P}$
from node $v_{\mathrm{P}}$ to node $\tilde{v}_{\mathrm{P}}$, where
$l_{G}(\cdot)$ corresponds to the indicator for the first packet
dropping link ($f(\mathcal{P})$) and $l_{\tilde{G}}(\cdot)$ corresponds
to packet transmission failure indicator for the rest of the path
($\mathcal{R}(\mathcal{P})$). Hence, by arguments similar to the
ones used above, we can show that if $\rho_{G}+\rho_{\tilde{G}}\leq1$,
then $l\in\Lambda_{\rho}$ with $\rho=\rho_{G}+\rho_{\tilde{G}}$.
This $\rho$ value is utilized for stability analysis with Theorem~\ref{Stability-Theorem}.
\end{remark}

\section{Illustrative Numerical Examples}

\label{sec:Illustrative}

In this section, we present illustrative examples to demonstrate the
utility of our results in the characterization of communication failures
on multi-hop networks. We also investigate the effects of those failures
on the stability of a multi-hop networked control system. 

Consider the networked control system \eqref{eq:system}, \eqref{eq:control-input}
with 
\begin{align}
A=\left[\begin{array}{cc}
1 & 0.1\\
-0.5 & 1.1
\end{array}\right],\quad B=\left[\begin{array}{c}
0.1\\
1.2
\end{array}\right],\quad K=\left[\begin{array}{cc}
-2.9012 & -0.9411\end{array}\right]\label{eq:example-abk-values}
\end{align}
together with the networks $G$ and $\tilde{G}$ in Fig.~\ref{Flo:operation}.
This system was explored previously in \cite{cetinkaya-tac} with
a single channel network model. Here, differently from \cite{cetinkaya-tac},
we consider networks $G$ and $\tilde{G}$ that incorporate multiple
paths and multiple links for packet transmissions. 

In what follows, we investigate various scenarios where we demonstrate
the utility of our results in Sections~\ref{sec:Random-and-Malicious}
and \ref{Sec:FailuresOnPaths} for characterizing overall network
failures. For each different scenario, our goal is to find out the
level of transmission failures that can be tolerated on the communication
links so that the stability of the system \eqref{eq:system}, \eqref{eq:control-input}
is guaranteed. 

\subsection{Data corruption/packet dropout issues on both networks $G$ and $\tilde{G}$\label{subsec:Data-corruption}}

Consider the scenario where all links on both networks $G$ and $\tilde{G}$
are subject to malicious/nonmalicious data corruption or packet dropout
issues. We explore the general situation where the failures may depend
on each other. For this general setup, we can use Proposition~\ref{Proposition-and-operation}
and Theorem~\ref{Theorem-lp-general} for the characterization of
the overall transmission failures between the plant and the controller.
As explained in Remark~\ref{Remark-General-Rho}, the overall packet
exchange failure process $l(\cdot)$ satisfies Assumption~\ref{MainAssumption}
with $\rho=\rho_{G}+\rho_{\tilde{G}}$. Here, $\rho_{G}$ and $\rho_{\tilde{G}}$
are asymptotic failure ratios for the networks $G$ and $\tilde{G}$,
that is, $l_{G}\in\Lambda_{\rho_{G}}$, $l_{\tilde{G}}\in\Lambda_{\rho_{\tilde{G}}}$. 

To find the values of $\rho_{G}$ and $\rho_{\tilde{G}}$, we can
use Proposition~\ref{Proposition-and-operation} and Theorem~\ref{Theorem-lp-general}.
In particular, by Proposition~\ref{Proposition-and-operation} and
Theorem~\ref{Theorem-lp-general}, we obtain $\rho_{G}=\min_{i\in\{1,\ldots,c=3\}}\rho_{\mathcal{P}_{i}}$
where $\rho_{\mathcal{P}_{i}}=\sum_{j=1}^{|\mathcal{P}_{i}|}\rho_{\mathcal{P}_{i}}^{\mathcal{P}_{i,j}}$.
Similarly, we have $\rho_{\tilde{G}}=\min_{i\in\{1,\ldots,\tilde{c}\}}\rho_{\tilde{\mathcal{P}}_{i}}$
and $\rho_{\tilde{\mathcal{P}}_{i}}=\sum_{j=1}^{|\tilde{\mathcal{P}}_{i}|}\rho_{\tilde{\mathcal{P}}_{i}}^{\tilde{\mathcal{P}}_{i,j}}$,
where $\tilde{c}=4$ is the number of paths (denoted by $\tilde{\mathcal{P}}_{i}$)
from the controller node $\tilde{v}_{\mathrm{C}}$ to the plant node
$\tilde{v}_{\mathrm{P}}$ on the network $\tilde{G}$. Consequently,
Assumption~\ref{MainAssumption} holds with 
\begin{align*}
\rho & =\min_{i\in\{1,\ldots,c\}}\sum_{j=1}^{|\mathcal{P}_{i}|}\rho_{\mathcal{P}_{i}}^{\mathcal{P}_{i,j}}+\min_{i\in\{1,\ldots,\tilde{c}\}}\sum_{j=1}^{|\tilde{\mathcal{P}}_{i}|}\rho_{\tilde{\mathcal{P}}_{i}}^{\tilde{\mathcal{P}}_{i,j}}.
\end{align*}
The effect of the asymptotic packet failure ratios $\rho_{\mathcal{P}_{i}}^{\mathcal{P}_{i,j}}$,
$\rho_{\tilde{\mathcal{P}}_{i}}^{\tilde{\mathcal{P}}_{i,j}}$ on the
stability of the networked system can be analyzed by using Theorem~\ref{Stability-Theorem}.
First, we identify the values of asymptotic packet exchange failure
ratio $\rho$ in Assumption~\ref{MainAssumption}, for which the
stability conditions \eqref{eq:betacond}--\eqref{eq:betaandvarphicond}
hold. For this numerical example, there exist a positive-definite
matrix $P$ and scalars $\beta\in(0,1),\varphi\in[1,\infty)$ that
satisfy \eqref{eq:betacond}--\eqref{eq:betaandvarphicond}, when
$\rho$ is less than $0.411$. Hence, Theorem~\ref{Stability-Theorem}
guarantees that the zero solution of the closed-loop system \eqref{eq:system},
\eqref{eq:control-input} is asymptotically stable almost surely if
the asymptotic packet transmission failure ratios satisfy 
\begin{align}
\min_{i\in\{1,\ldots,c\}}\sum_{j=1}^{|\mathcal{P}_{i}|}\rho_{\mathcal{P}_{i}}^{\mathcal{P}_{i,j}}+\min_{i\in\{1,\ldots,\tilde{c}\}}\sum_{j=1}^{|\tilde{\mathcal{P}}_{i}|}\rho_{\tilde{\mathcal{P}}_{i}}^{\tilde{\mathcal{P}}_{i,j}} & \leq0.411.\label{eq:stability-condition}
\end{align}
The system operator can guarantee \eqref{eq:stability-condition}
by ensuring that at least one path in each network is sufficiently
secure and reliable. In particular, if there exist a path $\mathcal{P}_{i^{*}}$
in network $G$ and a path $\tilde{\mathcal{P}}_{i^{*}}$ in network
$\tilde{G}$ such that $\sum_{j=1}^{|\mathcal{P}_{i^{*}}|}\rho_{\mathcal{P}_{i^{*}}}^{\mathcal{P}_{i^{*},j}}+\sum_{j=1}^{|\mathcal{\tilde{P}}_{i^{*}}|}\rho_{\mathcal{\tilde{P}}_{i^{*}}}^{\mathcal{\tilde{P}}_{i^{*},j}}\leq0.411$,
then \eqref{eq:stability-condition} holds and the stability is guaranteed
regardless of the security/reliability of all other paths. Another
approach to guarantee \eqref{eq:stability-condition} is to ensure
a certain level of security/reliability for all links. For this example,
if all links are sufficiently secure and reliable so that $\rho_{\mathcal{P}_{i}}^{\mathcal{P}_{i,j}},\rho_{\mathcal{\tilde{P}}_{i}}^{\mathcal{\tilde{P}}_{i,j}}\leq\overline{\rho}\triangleq\frac{0.411}{6}$,
then \eqref{eq:stability-condition} holds and the stability is guaranteed.
To see this, first note that the network $G$ contains $3$ paths,
and the number of links on these paths are given by $2,3,3$. Furthermore,
the network $\tilde{G}$ contains $4$ paths, each of which contains
$4$ links. It follows that 
\begin{align*}
 & \min_{i\in\{1,\ldots,c\}}\sum_{j=1}^{|\mathcal{P}_{i}|}\rho_{\mathcal{P}_{i}}^{\mathcal{P}_{i,j}}+\min_{i\in\{1,\ldots,\tilde{c}\}}\sum_{j=1}^{|\tilde{\mathcal{P}}_{i}|}\rho_{\tilde{\mathcal{P}}_{i}}^{\tilde{\mathcal{P}}_{i,j}}\leq\min_{i\in\{1,\ldots,c\}}\sum_{j=1}^{|\mathcal{P}_{i}|}\overline{\rho}+\min_{i\in\{1,\ldots,\tilde{c}\}}\sum_{j=1}^{|\tilde{\mathcal{P}}_{i}|}\overline{\rho}\\
 & \quad=\min\{2\overline{\rho},3\overline{\rho},3\overline{\rho}\}+\min\{4\overline{\rho},4\overline{\rho},4\overline{\rho},4\overline{\rho}\}=6\overline{\rho}\leq0.411,
\end{align*}
 which implies that \eqref{eq:stability-condition} holds, and hence
stability is guaranteed.

Notice that in this example, we have not made any particular assumption
on the independence or randomness of the failures on the links. In
fact, all links may be subject to failures caused by actions of coordinated
adversaries. In such a case, the occurrence of the failures may be
nonrandom, and moreover, the binary-valued processes that characterize
failures on different links would depend on each other. For example,
in the case of data-corruption attacks, the worst-case scenario would
be that the failures on the paths are synchronized so that packet
transmissions necessarily fail in all parallel paths (such as paths
$\mathcal{P}_{1}$, $\mathcal{P}_{2}$, and $\mathcal{P}_{3}$ of
the graph $G$) at the same time. Notice that the condition $\rho_{\mathcal{P}_{i}}^{\mathcal{P}_{i,j}},\rho_{\mathcal{\tilde{P}}_{i}}^{\mathcal{\tilde{P}}_{i,j}}\leq\overline{\rho}$
guarantees that such failures happen sufficiently rarely in average
in the long run. Thus, networked stabilization can be achieved through
the successful exchanges of measurement and control data. 

In the following examples, we will illustrate how our results in Sections~\ref{sec:Random-and-Malicious}
and \ref{Sec:FailuresOnPaths} can be used for scenarios where some
information on the properties of the communication links are available. 

\begin{figure}
\centering  \includegraphics[width=0.9\columnwidth]{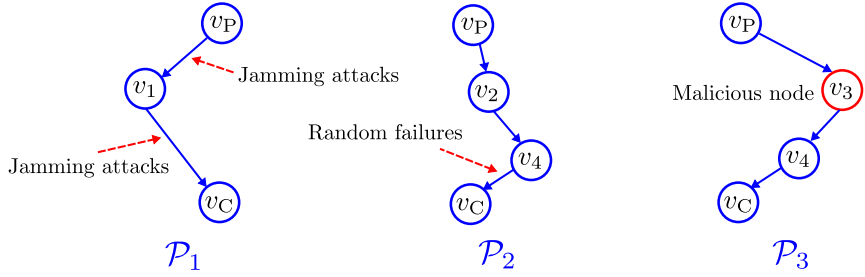} \vskip -3pt

\caption{\fontsize{9}{9}\selectfont{ Paths from the plant node $v_{\mathrm{P}}$
to the controller node $v_{\mathrm{C}}$ in network $G$} }
 \label{Flo:pathsofg}
\end{figure}

\begin{figure}
\centering  \includegraphics[width=0.9\columnwidth]{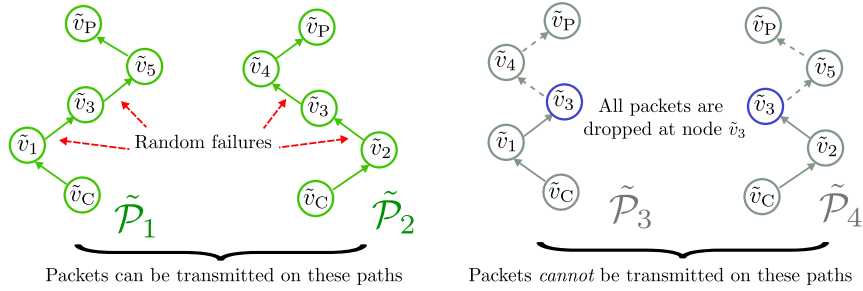}
\vskip -3pt

\caption{\fontsize{9}{9}\selectfont{ Paths from the controller node $\tilde{v}_{\mathrm{C}}$
to the plant node $\tilde{v}_{\mathrm{P}}$ in network $\tilde{G}$
} }
 \label{Flo:pathsofgtilde}
\end{figure}

\subsection{Jamming Attacks and Random Transmission Failures on Multiple Links\label{subsec:Jamming-Attacks-and}}

Consider the network $G$ and its paths shown in Fig.~\ref{Flo:pathsofg}.
We assume that the paths $\mathcal{P}_{1}$ and $\mathcal{P}_{3}$
are subject to malicious attacks. In particular, the node $v_{3}$
on $\mathcal{P}_{3}$ is assumed to be controlled by a malicious agent
and all packets arriving at node $v_{3}$ are dropped. Hence, we have
$l_{\mathcal{P}_{3}}^{\mathcal{P}_{3,2}}(t)=1$, $t\in\mathbb{N}_{0}$,
which implies that $l_{\mathcal{P}_{3}}(t)=1$, $t\in\mathbb{N}_{0}$.
Moreover, the first and the second links on path $\mathcal{P}_{1}$
are assumed to be subject to jamming attacks that cause data corruption.
Here the worst-case scenario happens when the attackers coordinate
and jam only one of the links at once. This would maximize the number
of packet losses within the total energy constraint of the attackers.
Our results take this worst-case into account. 

In particular, suppose that the attacked links $\mathcal{P}_{1,1}$
and $\mathcal{P}_{1,2}$ satisfy $l_{\mathcal{P}_{1}}^{\mathcal{P}_{1,1}}\in\Pi_{\kappa,w}$,
$l_{\mathcal{P}_{1}}^{\mathcal{P}_{1,2}}\in\Pi_{\kappa,w}$ with $\kappa\geq0$
and $w\in(0,1)$. Here, $w$ provides a bound on average failures
on each link and it is related to the energy available to each attacker.
The fact that $l_{\mathcal{P}_{1}}^{\mathcal{P}_{1,1}}\in\Pi_{\kappa,w}$,
$l_{\mathcal{P}_{1}}^{\mathcal{P}_{1,2}}\in\Pi_{\kappa,w}$ implies
$l_{\mathcal{P}_{1}}^{\mathcal{P}_{1,1}}\in\Lambda_{\rho_{\mathcal{P}_{1}}^{\mathcal{P}_{1,1}}}$,
$l_{\mathcal{P}_{1}}^{\mathcal{P}_{1,2}}\in\Lambda_{\rho_{\mathcal{P}_{1}}^{\mathcal{P}_{1,2}}}$
for all $\rho_{\mathcal{P}_{1}}^{\mathcal{P}_{1,1}},\rho_{\mathcal{P}_{1}}^{\mathcal{P}_{1,1}}\in(w,1]$.
It then follows from Proposition~\ref{Proposition-lp-or} that $\rho_{\mathcal{P}_{1}}=\rho_{\mathcal{P}_{1}}^{\mathcal{P}_{1,1}}+\rho_{\mathcal{P}_{1}}^{\mathcal{P}_{1,1}}$,
and hence $l_{\mathcal{P}_{1}}\in\Lambda_{\rho_{\mathcal{P}_{1}}}$
for $\rho_{\mathcal{P}_{1}}\in(2w,1]$. Notice that $l_{\mathcal{P}_{1}}\in\Lambda_{\rho_{\mathcal{P}_{1}}}$
holds even in the worst-case scenario mentioned above. Here $\rho_{\mathcal{P}_{1}}\in(2w,1]$
provides an upper-bound on the average number of overall failures
on path $\mathcal{P}_{1}$ in the worst-case scenario. 

We further assume that the link $(v_{4},v_{\mathrm{C}})$ on path
$\mathcal{P}_{2}$ is subject to random data corruption and the associated
failure indicator process $\{l_{\mathcal{P}_{2}}^{\mathcal{P}_{2,3}}(t)\in\{0,1\}\}_{t\in\mathbb{N}_{0}}$
satisfies $l_{\mathcal{P}_{2}}^{\mathcal{P}_{2,3}}\in\Gamma_{p_{0},p_{1}}$
with $p_{0},p_{1}\in(0,1)$. We consider ideal communication on the
other links on path $\mathcal{P}_{2}$. Thus, by using \eqref{eq:lp-or},
we obtain 
\begin{align*}
 & l_{\mathcal{P}_{2}}(t)=l_{\mathcal{P}_{2}}^{\mathcal{P}_{2,1}}(t)\vee l_{\mathcal{P}_{2}}^{\mathcal{P}_{2,2}}(t)\vee l_{\mathcal{P}_{2}}^{\mathcal{P}_{2,3}}(t)=0\vee0\vee l_{\mathcal{P}_{2}}^{\mathcal{P}_{2,3}}(t)=l_{\mathcal{P}_{2}}^{\mathcal{P}_{2,3}}(t),\quad t\in\mathbb{N}_{0},
\end{align*}
 and hence we also have $l_{\mathcal{P}_{2}}\in\Gamma_{p_{0},p_{1}}$. 

To characterize overall failures on network $G$, we use \eqref{eq:lg-multiple-paths}
and obtain 
\begin{align*}
 & l_{G}(t)=l_{\mathcal{P}_{1}}(t)\wedge l_{\mathcal{P}_{2}}(t)\wedge l_{\mathcal{P}_{3}}(t)=l_{\mathcal{P}_{1}}(t)\wedge l_{\mathcal{P}_{2}}(t)\wedge1=l_{\mathcal{P}_{1}}(t)\wedge l_{\mathcal{P}_{2}}(t),\quad t\in\mathbb{N}_{0}.
\end{align*}
 Now, assuming that the failures on paths $\mathcal{P}_{1}$ and $\mathcal{P}_{2}$
are mutually-independent, it follows from Theorem~\ref{Theorem-Gamma-Lambda}
with $l^{(1)}(t)=l_{\mathcal{P}_{2}}(t)$, $l^{(2)}(t)=l_{\mathcal{P}_{1}}(t)$,
$p_{0}^{(1)}=p_{0}$, $p_{1}^{(1)}=p_{1}$, and $\rho^{(2)}=\rho_{\mathcal{P}_{1}}$
that if $p_{1}\rho_{\mathcal{P}_{1}}<1$, then $l_{G}\in\Lambda_{\rho_{G}}$
with $\rho_{G}\in(p_{1}\rho_{\mathcal{P}_{1}},1]$. Noting that $\rho_{\mathcal{P}_{1}}\in(2w,1]$,
we see that if $2p_{1}w<1$, then $l_{G}\in\Lambda_{\rho_{G}}$ for
all $\rho_{G}\in(2p_{1}w,1]$. 

Next, consider the network $\tilde{G}$. Suppose that $\tilde{G}$
is secure against attacks, but it is unreliable and subject to random
transmission failures. To characterize the overall transmission failures
on the network $\tilde{G}$, we utilize Theorems~\ref{Theorem-HiddenMarkov}
and \ref{Theorem-HiddenMarkov-1}. Before we apply these results,
we first describe the routing scheme on this network. Specifically,
in this network $\tilde{v}_{3}$ is assumed to be a router that forwards
all incoming packets from node $\tilde{v}_{1}$ only to node $\tilde{v}_{5}$,
and all incoming packets from node $\tilde{v}_{2}$ only to node $\tilde{v}_{4}$.
Among the paths shown in Fig.~\ref{Flo:pathsofgtilde}, packets may
be transmitted on $\tilde{\mathcal{P}}_{1}$ and $\tilde{\mathcal{P}}_{2}$,
but are never transmitted on $\tilde{\mathcal{P}}_{3}$ and $\tilde{\mathcal{P}}_{4}$
due to the routing scheme. Hence $l_{\tilde{P}_{3}}(t)=1$, $l_{\tilde{P}_{4}}(t)=1$,
$t\in\mathbb{N}_{0}$. 

We assume that the links $(\tilde{v}_{1},\tilde{v}_{3})$ and $(\tilde{v}_{3},\tilde{v}_{5})$
on path $\tilde{\mathcal{P}}_{1}$ and the links $(\tilde{v}_{2},\tilde{v}_{3})$
and $(\tilde{v}_{3},\tilde{v}_{4})$ on $\tilde{\mathcal{P}}_{2}$
face random data corruption issues. Furthermore, the failure indicator
processes $\{l_{\tilde{\mathcal{P}}_{1}}^{(\tilde{v}_{1},\tilde{v}_{3})}(t)\in\{0,1\}\}_{t\in\mathbb{N}_{0}}$,
$\{l_{\tilde{\mathcal{P}}_{1}}^{(\tilde{v}_{3},\tilde{v}_{5})}(t)\in\{0,1\}\}_{t\in\mathbb{N}_{0}}$,
$\{l_{\tilde{\mathcal{P}}_{2}}^{(\tilde{v}_{2},\tilde{v}_{3})}(t)\in\{0,1\}\}_{t\in\mathbb{N}_{0}}$,
$\{l_{\tilde{\mathcal{P}}_{2}}^{(\tilde{v}_{3},\tilde{v}_{4})}(t)\in\{0,1\}\}_{t\in\mathbb{N}_{0}}$
are assumed to be mutually independent processes that belong to the
hidden Markov model class $\Gamma_{p_{0},p_{1}}$ with $p_{0},p_{1}\in(0,1)$
(see Definition~\ref{Definition-Hidden-Markov}). The links that
are connected directly to the plant and the controller nodes ($\tilde{v}_{\mathrm{P}}$
and $\tilde{v}_{\mathrm{C}}$) are considered to be ideal communication
links. In other words, $l_{\tilde{\mathcal{P}}_{i}}^{\tilde{\mathcal{P}}_{i,1}}(t)=0$,
$l_{\tilde{\mathcal{P}}_{i}}^{\tilde{\mathcal{P}}_{i,4}}(t)=0$, $t\in\mathbb{N}_{0}$,
$i\in\{1,2\}$. 

Now, observe that the failure indicators for paths $\tilde{\mathcal{P}}_{1}$
and $\tilde{\mathcal{P}}_{2}$ satisfy 
\begin{align*}
l_{\tilde{\mathcal{P}}_{i}}(t) & =l_{\tilde{\mathcal{P}}_{i}}^{\tilde{\mathcal{P}}_{i,1}}(t)\vee l_{\tilde{\mathcal{P}}_{i}}^{\tilde{\mathcal{P}}_{i,2}}(t)\vee l_{\tilde{\mathcal{P}}_{i}}^{\tilde{\mathcal{P}}_{i,3}}(t)\vee l_{\tilde{\mathcal{P}}_{i}}^{\tilde{\mathcal{P}}_{i,4}}(t)=0\vee l_{\tilde{\mathcal{P}}_{i}}^{\tilde{\mathcal{P}}_{i,2}}(t)\vee l_{\tilde{\mathcal{P}}_{i}}^{\tilde{\mathcal{P}}_{i,3}}(t)\vee0\\
 & =l_{\tilde{\mathcal{P}}_{i}}^{\tilde{\mathcal{P}}_{i,2}}(t)\vee l_{\tilde{\mathcal{P}}_{i}}^{\tilde{\mathcal{P}}_{i,3}}(t),\quad t\in\mathbb{N}_{0},\,\,i\in\{1,2\}.
\end{align*}
 By applying Theorem~\ref{Theorem-HiddenMarkov-1} with $l^{(1)}(t)=l_{\tilde{\mathcal{P}}_{i}}^{\tilde{\mathcal{P}}_{i,2}}(t)$,
$l^{(2)}(t)=l_{\tilde{\mathcal{P}}_{i}}^{\tilde{\mathcal{P}}_{i,3}}(t)$,
$p_{0}^{(1)}=p_{0}^{(2)}=p_{0}$, and $p_{1}^{(1)}=p_{1}^{(2)}=p_{1}$,
we obtain $l_{\tilde{\mathcal{P}}_{i}}\in\Gamma_{\tilde{p}_{0},\tilde{p}_{1}}$,
$i\in\{1,2\}$,  where $\tilde{p}_{0}=p_{0}^{2}$ and $\tilde{p}_{1}=\min\{p_{1}+p_{0}p_{1},p_{1}+p_{0}p_{1},1\}=\min\{p_{1}+p_{0}p_{1},1\}$. 

Next, since $l_{\tilde{P}_{3}}(t)=1$, $l_{\tilde{P}_{4}}(t)=1$,
$t\in\mathbb{N}_{0}$, we have 
\begin{align*}
l_{\tilde{G}}(t) & =l_{\tilde{\mathcal{P}}_{1}}(t)\wedge l_{\tilde{\mathcal{P}}_{2}}(t)\wedge l_{\tilde{\mathcal{P}}_{3}}(t)\wedge l_{\tilde{\mathcal{P}}_{4}}(t)=l_{\tilde{\mathcal{P}}_{1}}(t)\wedge l_{\tilde{\mathcal{P}}_{2}}(t)\wedge1\wedge1=l_{\tilde{\mathcal{P}}_{1}}(t)\wedge l_{\tilde{\mathcal{P}}_{2}}(t),
\end{align*}
 for $t\in\mathbb{N}_{0}$. It then follows from Theorem~\ref{Theorem-HiddenMarkov}
with $l^{(1)}(t)=l_{\tilde{\mathcal{P}}_{1}}(t)$, $l^{(2)}(t)=l_{\tilde{\mathcal{P}}_{2}}(t)$,
$p_{0}^{(1)}=p_{0}^{(2)}=p_{0}^{2}$, and $p_{1}^{(1)}=p_{1}^{(2)}=\min\{p_{1}+p_{0}p_{1},1\}$,
that we have $l_{\tilde{G}}\in\Gamma_{p_{0}^{\tilde{G}},p_{1}^{\tilde{G}}}$
with 
\begin{align*}
p_{0}^{\tilde{G}} & =\min\{p_{0}^{(1)}+p_{1}^{(1)}p_{0}^{(2)},p_{0}^{(2)}+p_{1}^{(2)}p_{0}^{(1)},1\}=\min\{p_{0}^{(1)}+p_{1}^{(1)}p_{0}^{(2)},1\}\\
 & =\min\{p_{0}^{2}+\min\{p_{1}+p_{0}p_{1},1\}p_{0}^{2},1\}=\min\{p_{0}^{2}+\min\{p_{0}^{2}p_{1}+p_{0}^{3}p_{1},p_{0}^{2}\},1\}\\
 & =\min\{\min\{p_{0}^{2}+p_{0}^{2}p_{1}+p_{0}^{3}p_{1},2p_{0}^{2}\},1\}=\min\{p_{0}^{2}+p_{0}^{2}p_{1}+p_{0}^{3}p_{1},2p_{0}^{2},1\}
\end{align*}
and $p_{1}^{\tilde{G}}=p_{1}^{(1)}p_{1}^{(2)}=(\min\{p_{1}+p_{0}p_{1},1\})^{2}=\min\{(p_{1}+p_{0}p_{1})^{2},1\}$.
Finally, in the case where $(p_{1}+p_{0}p_{1})^{2}<1$, as a consequence
of Proposition~\ref{PropositionGammaLambdaRelation}, we obtain $l_{\tilde{G}}\in\Lambda_{\tilde{\rho}_{G}}$
for $\rho_{\tilde{G}}\in(p_{1}^{\tilde{G}},1]=((p_{1}+p_{0}p_{1})^{2},1]$. 

Now, we note once again that the overall packet exchange failure indicator
$\{l(t)\in\{0,1\}\}_{t\in\mathbb{N}_{0}}$ satisfies $l\in\Lambda_{\rho}$
with $\rho=\rho_{G}+\rho_{\tilde{G}}$. Since $\rho_{G}\in(2p_{1}w,1]$
and $\rho_{\tilde{G}}\in((p_{1}+p_{0}p_{1})^{2},1]$, it follows that
if $2p_{1}w+(p_{1}+p_{0}p_{1})^{2}<1$, then $l\in\Lambda_{\rho}$
with $\rho\in(2p_{1}w+(p_{1}+p_{0}p_{1})^{2},1]$. As we discussed
in Section~\ref{subsec:Data-corruption}, stability of the networked
control system can be ensured when $\rho\leq0.411$. It follows that
if 
\begin{align}
2p_{1}w+(p_{1}+p_{0}p_{1})^{2} & <0.411,\label{eq:secondexamplesufficientcondition}
\end{align}
then $l\in\Lambda_{\rho}$ holds with $\rho=0.411$, implying almost
sure asymptotic stability of the networked control system. Observe
that by utilizing the results in Sections~\ref{sec:Random-and-Malicious}
and \ref{Sec:FailuresOnPaths}, we are able to derive the sufficient
stability condition \eqref{eq:secondexamplesufficientcondition} in
terms of the attack rate $w\in(0,1)$ associated with the links on
path $\mathcal{P}_{1}$ that are attacked by jamming attackers as
well as random failure parameters $p_{0},p_{1}\in(0,1)$. 

\subsection{State-dependent Attacks by a Malicious Node\label{subsec:State-dependent-Attacks-by}}

In the scenarios discussed in Sections~\ref{subsec:Data-corruption}
and \ref{subsec:Jamming-Attacks-and}, the strategies of the attackers
are not specified. However, in certain cases an attacker may have
access to the state or control input information and be able to directly
cause transmission failures between the plant and the controller.
In such scenarios, the goal of the attacker might be to increase state
norm with small amount of attacks. In this section, our goal is to
illustrate such an attack strategy. In particular, we consider the
case where the plant node $v_{\mathrm{P}}$ is compromised by an attacker.
The attacker is assumed to have access to the state information. 

We consider an attack strategy that is based on an optimization problem.
In particular, the attacker at node $v_{\mathrm{P}}$ decides whether
to transmit the state information on links $\mathcal{P}_{1,1}=(v_{\mathrm{P}},v_{1})$,
$\mathcal{P}_{2,1}=(v_{\mathrm{P}},v_{2})$, and $\mathcal{P}_{3,1}=(v_{\mathrm{P}},v_{3})$
after solving an optimization problem for maximizing the norm of the
state at a future time. In particular, we consider the attack strategy
\eqref{eq:attackstr1}, \eqref{eq:attackstr2} discussed in Example~\ref{StateDependentOptimizationStrategy},
where we represent the attackers actions with a binary-valued process
$\{l_{\mathrm{A}}(t)\in\{0,1\}\}_{t\in\mathbb{N}_{0}}$.  In the case
where $l_{\mathrm{A}}(t)=0$, the attacker transmits the state packet
on $\mathcal{P}_{1,1}=(v_{\mathrm{P}},v_{1})$, $\mathcal{P}_{2,1}=(v_{\mathrm{P}},v_{2})$,
and $\mathcal{P}_{3,1}=(v_{\mathrm{P}},v_{3})$; moreover, $l_{\mathrm{A}}(t)=1$
indicates no transmission. 

Observe that in this scenario, we have 
\begin{align*}
l_{\mathcal{P}_{1}}^{\mathcal{P}_{1,1}}(t)= & l_{\mathcal{P}_{2}}^{\mathcal{P}_{2,1}}(t)=l_{\mathcal{P}_{1}}^{\mathcal{P}_{1,1}}(t)=l_{\mathrm{A}}(t),\quad t\in\mathbb{N}_{0}.
\end{align*}
 Since $l_{\mathrm{A}}\in\Pi_{\kappa,w}$, we have $l_{\mathrm{A}}\in\Lambda_{\rho_{\mathrm{A}}}$
with $\rho_{\mathrm{A}}\in(w,1]$. As a consequence, $l_{\mathcal{P}_{i}}^{\mathcal{P}_{i,1}}\in\Lambda_{\rho_{\mathcal{P}_{i}}^{\mathcal{P}_{i,1}}}$
with $\rho_{\mathcal{P}_{i}}^{\mathcal{P}_{i,1}}\in(w_{\mathrm{A}},1]$
for all $i\in\{1,2,3\}$. Notice that $l_{\mathrm{A}}(t)=1$ implies
$l(t)=1$, since the attacker can completely prevent the packet exchange
between the plant and the controller. Observe also that if there are
other sources of transmission failures on the network, then there
may be times when $l(t)=1$ even if $l_{\mathrm{A}}(t)=0$. As a result,
if $l_{\mathrm{A}}(t)=0$ then the attacker may not be able to correctly
predict the state $x(t+1)$ at time $t$, as there may or may not
be a failure in the network that prevents control input to reach the
plant. However, as the optimization problem in \eqref{eq:attackstr2}
is solved at each time step, the updated state information is used
for decision. 

Notice that in the case where all links other than $\mathcal{P}_{1,1}=(v_{\mathrm{P}},v_{1})$,
$\mathcal{P}_{2,1}=(v_{\mathrm{P}},v_{2})$, and $\mathcal{P}_{3,1}=(v_{\mathrm{P}},v_{3})$
are secure and reliable, we have $l_{G}(t)=l_{\mathrm{A}}(t)$, since
\begin{align*}
l_{G}(t) & =l_{\mathcal{P}_{1}}(t)\wedge l_{\mathcal{P}_{2}}(t)\wedge l_{\mathcal{P}_{3}}(t)=l_{\mathcal{P}_{1}}^{\mathcal{P}_{1,1}}(t)\wedge l_{\mathcal{P}_{2}}^{\mathcal{P}_{2,1}}(t)\wedge l_{\mathcal{P}_{3}}^{\mathcal{P}_{3,1}}(t)\\
 & =l_{\mathrm{A}}(t)\wedge l_{\mathrm{A}}(t)\wedge l_{\mathrm{A}}(t)=l_{\mathrm{A}}(t),\quad t\in\mathbb{N}_{0}.
\end{align*}
Hence, $l_{G}\in\Lambda_{\rho_{G}}$ with $\rho_{G}\in(w,1]$. 

Now suppose that the network $\tilde{G}$ also faces failures. In
particular, consider the setup in Section~\ref{subsec:Jamming-Attacks-and},
where $l_{\tilde{G}}\in\Lambda_{\tilde{\rho}_{G}}$ for $\rho_{\tilde{G}}\in(p_{1}^{\tilde{G}},1]=((p_{1}+p_{0}p_{1})^{2},1]$.
Since $l\in\Lambda_{\rho}$ with $\rho=\rho_{G}+\rho_{\tilde{G}}$,
we have $l\in\Lambda_{\rho}$ for all $\rho\in(w+(p_{1}+p_{0}p_{1})^{2},1]$.
Noting that the stability of the networked control system can be ensured
when $\rho\leq0.411$, we can impose a sufficient condition on the
attack rate $w$. Specifically, for the scenario of this section,
if 
\begin{align}
w+(p_{1}+p_{0}p_{1})^{2} & <0.411,\label{eq:thirdexamplesufficientcondition}
\end{align}
then $l\in\Lambda_{\rho}$ holds with $\rho=0.411$, and thus the
networked control system is almost surely asymptotically stable.

To illustrate the effect of different parameters in the attack strategy
\eqref{eq:attackstr2}, we conduct simulations. First, we generate
$50$ sample trajectories of the process $l_{\tilde{G}}$ by setting
the failure processes $l_{\tilde{\mathcal{P}}_{i}}^{\tilde{\mathcal{P}}_{i,2}}$,
$l_{\tilde{\mathcal{P}}_{i}}^{\tilde{\mathcal{P}}_{i,3}}$, $i\in\{1,2\}$,
as outputs of time-inhomogeneous hidden Markov models such that 
\begin{align}
l_{\tilde{\mathcal{P}}_{i}}^{\tilde{\mathcal{P}}_{i,j}}(t) & =\theta_{i,j}(t),\quad t\in\mathbb{N}_{0},\quad i\in\{1,2\},\quad j\in\{2,3\},\label{eq:hiddenmarkov-1}
\end{align}
where $\{\theta_{i,j}(t)\in\Theta^{l}=\{0,1\}\}_{t\in\mathbb{N}_{0}}$
is an $\mathcal{F}_{t}$-adapted finite-state time-inhomogeneous Markov
chain with initial distributions and time-varying transition probability
functions satisfying $\mathbb{P}[\theta_{i,j}(0)=0]=0.6$ and 
\begin{align*}
 & \mathbb{P}[\theta_{i,j}(t+1)=1|\theta_{i,j}(t)=0]=0.2+0.02\cos^{2}(0.1t),\\
 & \mathbb{P}[\theta_{i,j}(t+1)=1|\theta_{i,j}(t)=1]=0.2+0.02\sin^{2}(0.2t),
\end{align*}
 for $t\in\mathbb{N}_{0},i\in\{1,2\},j\in\{2,3\}$. Notice that $l_{\tilde{\mathcal{P}}_{i}}^{\tilde{\mathcal{P}}_{i,2}},l_{\tilde{\mathcal{P}}_{i}}^{\tilde{\mathcal{P}}_{i,3}}\in\Gamma_{p_{0},p_{1}}$,
$i\in\{1,2\}$, with $p_{0}=0.8$, $p_{1}=0.22$. Next, for each sample
trajectory of $l_{\tilde{G}}$, we simulate the networked control
system \eqref{eq:system}, \eqref{eq:control-input} under the attack
strategy \eqref{eq:attackstr2}. In Fig.~\ref{Flo:ComparisonN},
we show state trajectories when attack parameters are selected as
$w=0.25$ and $\kappa=1$. In top part of Fig.~\ref{Flo:ComparisonN},
the horizon parameter $N=2$, and in the bottom part $N=10$. Observe
that with larger horizon variable, the attacker following the strategy
\eqref{eq:attackstr2} can utilize the resources more efficiently
so that in certain sample state trajectories the state norm is set
to larger values for longer durations even though in both cases $N=2$
and $N=10$, the attack belongs to the same class $\Pi_{\kappa,w}$.
Notice that with $p_{0}=0.8$, $p_{1}=0.22$, and $w=0.25$, the inequality
\eqref{eq:thirdexamplesufficientcondition} holds. Therefore, the
zero solution of the closed-loop system is almost surely asymptotically
stable by Theorem~\ref{Stability-Theorem}. 

\begin{figure}
\centering  \includegraphics[width=0.75\columnwidth]{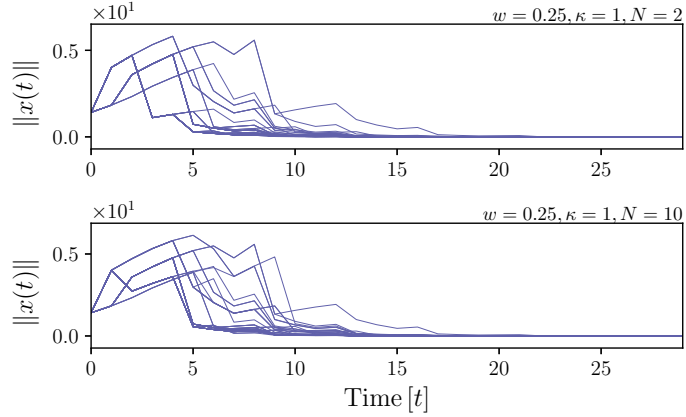}
\vskip -3pt

\caption{\fontsize{9}{9}\selectfont{ Comparison of state trajectories for
different values of $N$ in \eqref{eq:attackstr2} } }
 \label{Flo:ComparisonN}
\end{figure}

\begin{figure}
\centering  \includegraphics[width=0.75\columnwidth]{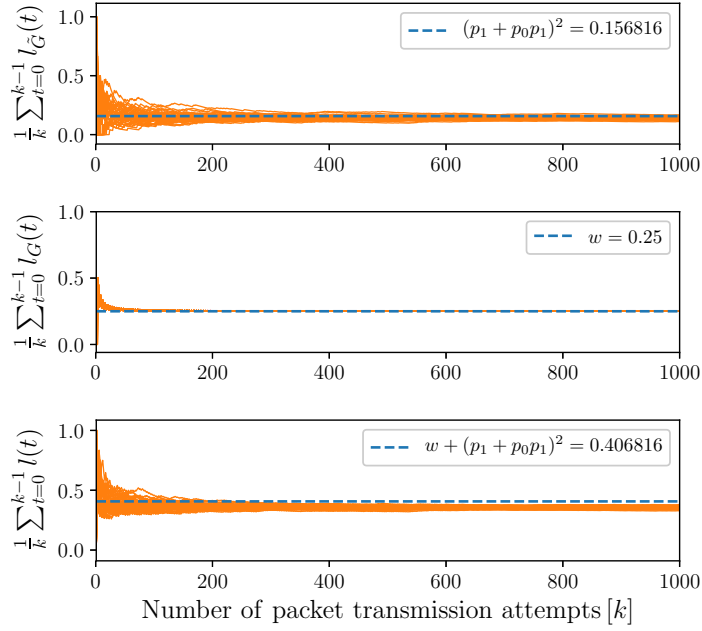}
\vskip -3pt

\caption{\fontsize{9}{9}\selectfont{ Average number of failures on networks
$\tilde{G}$ and $G$ together with average number of overall packet
exchange failures } }
 \label{Flo:AverageFailures}
\end{figure}

For $w=0.25$, $\kappa=1$, and $N=2$, we show in Fig.~\ref{Flo:AverageFailures}
the average number of failures on the networks $\tilde{G}$ and $G$
as well as the average number of total packet exchange failures between
the plant and the controller. These averages are upper bounded in
the long run by certain scalars. In particular, for a process $l$
that satisfies $l\in\Lambda_{\rho}$ with $\rho\in(\underline{\rho},1]$,
we have $\limsup_{k\to\infty}\frac{1}{k}\sum_{t=0}^{k-1}l(t)\leq\underline{\rho}$
(see Lemma~3.3 in \cite{cetinkaya-tac}). As a result, we have $\limsup_{k\to\infty}\frac{1}{k}\sum_{t=0}^{k-1}l_{\tilde{G}}(t)\leq(p_{1}+p_{0}p_{1})^{2}$,
$\limsup_{k\to\infty}\frac{1}{k}\sum_{t=0}^{k-1}l_{G}(t)\leq w$,
and $\limsup_{k\to\infty}\frac{1}{k}\sum_{t=0}^{k-1}l(t)\leq w+(p_{1}+p_{0}p_{1})^{2}$. 

Next, we consider the case where $\kappa=10$, which corresponds to
the case where the attacker has more initial resources. In this case,
we consider two scenarios $w=0.25$ and $w=0.75$, for which we obtain
the sample state trajectories shown in Fig.~\ref{Flo:Comparisonw}.
First, notice that in the case with $\kappa=10,w=0.25$ (top plot
in Fig.~\ref{Flo:Comparisonw}), the state grows to larger values
for longer durations in comparison to the case with $\kappa=1,w=0.25$
(top plot in Fig.~\ref{Flo:ComparisonN}). In both cases, the stability
conditions hold with $w=0.25$, and therefore the state eventually
converges to zero as guaranteed by Theorem~\ref{Stability-Theorem}.
Note that, when $w$ is set to larger values, the attacker can attack
at a higher rate. For this example, we set $w=0.75$ and observe that
for such a high attack rate all sample state trajectories diverge
(see bottom plot in Fig.~\ref{Flo:Comparisonw}). 

\begin{figure}
\centering  \includegraphics[width=0.75\columnwidth]{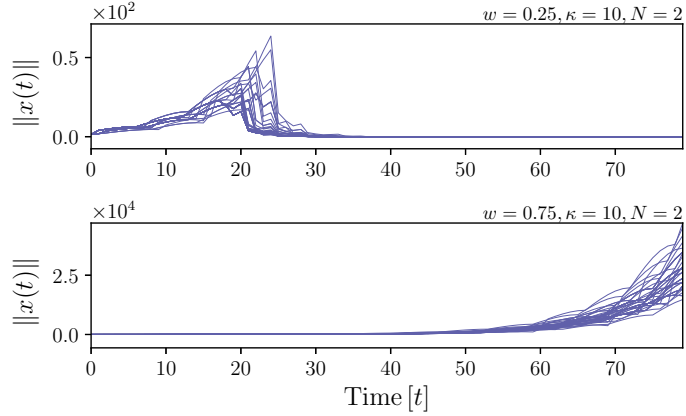}
\vskip -3pt

\caption{\fontsize{9}{9}\selectfont{ Comparison of state trajectories for
different values of $w$ in \eqref{eq:attackstr1}} }
 \label{Flo:Comparisonw}
\end{figure}

\section{Conclusion}

\label{sec:Conclusion}

In this paper, we explored state feedback control of a linear plant
over an unreliable network that may also face malicious attacks. We
developed a probabilistic approach to characterize the failures on
the network in terms of the failures on different paths between the
plant and the controller. We showed that the failures on each path
can be described as a combination of data-corruption and packet-dropout
failures on the communication links of that particular path. We obtained
sufficient conditions for almost sure asymptotic stability of the
overall networked control system, which allow us to check stability
by using a probabilistic upper-bound obtained for the average number
of packet exchange failures between the plant and the controller.

Our framework can take into account mutual independence/dependence
relationships between the failures on different links and paths of
a network, and as a result, it can provide relatively tight upper
bounds for the long run average number of overall transmission failures
on a network. This is achieved by utilizing upper bounds on the tail
probabilities of sums involving a binary-valued process from a general
class together with a binary-valued output process associated with
a time-inhomogeneous hidden Markov model.

\appendix

\section{Proof of Lemma~\ref{KeyMarkovLemma}}

In this section, we prove Lemma~\ref{KeyMarkovLemma}. As in the
proof of Lemma~A.1 of \cite{cetinkaya-tac}, for proving Lemma~\ref{KeyMarkovLemma},
we use Markov's inequality and follow the approaches used for obtaining
Chernoff-type tail distribution inequalities for sums of random variables
(see Appendix B of \cite{madhow2008fundamentals} and Section 1.9
of \cite{billingsley1986}). In the proof Lemma~\ref{KeyMarkovLemma},
the following result also plays a key role. 

\vskip 10pt

\begin{aplemma} \label{ExpectationLemma} Let $\{\xi(t)\in\Xi\}_{t\in\mathbb{N}_{0}}$
be a finite-state time-inhomogeneous Markov chain with transition
probabilities $p_{q,r}\colon\mathbb{N}_{0}\to[0,1]$, $q,r\in\Xi$.
Furthermore, let $\Xi_{1}\subset\Xi$ be given by $\Xi_{1}\triangleq\{r\in\Xi\colon h(r)=1\}$,
where $h\colon\Xi\to\{0,1\}$ is a binary-valued function. Then for
all $\phi>1$, $s\in\mathbb{N}$, and $\tilde{p}\in[0,1]$ such that
\begin{align}
\sum_{r\in\Xi_{1}}p_{q,r}(t)\leq\tilde{p},\quad q\in\Xi,\quad t\in\mathbb{N}_{0},\label{eq:transitioncondition}
\end{align}
 we have 
\begin{align}
\mathbb{E}[\phi^{\sum_{j=1}^{s}h(\xi(t_{j}))}] & \leq\phi\left((\phi-1)\tilde{p}+1\right)^{s-1},\label{eq:lemmaresultfors}
\end{align}
 where $t_{1},t_{2},\ldots,t_{s}\in\mathbb{N}_{0}$ denote indices
such that $0\leq t_{1}<t_{2}<\ldots<t_{s}$. \end{aplemma}

\begin{proof} We use induction to show \eqref{eq:lemmaresultfors}.
First, we consider the case where $s=1$. In this case, we have 
\begin{align}
\mathbb{E}[\phi^{\sum_{j=1}^{s}h(\xi(t_{j}))}] & =\mathbb{E}[\phi^{h(\xi(t_{1}))}]\leq\phi.\label{eq:resultfors1}
\end{align}
Next, consider the case where $s=2$. Observe that $t_{1}\leq t_{2}-1$,
thus the random variable $\xi(t_{1})$ is $\mathcal{F}_{t_{2}-1}$-measurable.
Consequently, 
\begin{align}
\mathbb{E}[\phi^{\sum_{j=1}^{s}h(\xi(t_{j}))}] & =\mathbb{E}[\phi^{h(\xi(t_{1}))}\phi^{h(\xi(t_{2}))}]=\mathbb{E}[\mathbb{E}[\phi^{h(\xi(t_{1}))}\phi^{h(\xi(t_{2}))}\mid\mathcal{F}_{t_{2}-1}]]\nonumber \\
 & =\mathbb{E}[\phi^{h(\xi(t_{1}))}\mathbb{E}[\phi^{h(\xi(t_{2}))}\mid\mathcal{F}_{t_{2}-1}]],\label{eq:cases2part1}
\end{align}
where the last equality follows from the fact that $\phi^{h(\xi(t_{1}))}$
is a measurable function of $\xi(t_{1})$, and hence, it is $\mathcal{F}_{t_{2}-1}$-measurable.
Now, since $\{\xi(t)\in\Xi\}_{t\in\mathbb{N}_{0}}$ is a Markov chain,
we have $\mathbb{E}[\phi^{h(\xi(t_{2}))}\mid\mathcal{F}_{t_{2}-1}]=\mathbb{E}[\phi^{h(\xi(t_{2}))}\mid\xi(t_{2}-1)]$.
Therefore, 
\begin{align}
 & \mathbb{E}[\phi^{\sum_{j=1}^{s}h(\xi(t_{j}))}]=\mathbb{E}[\phi^{h(\xi(t_{1}))}\mathbb{E}[\phi^{h(\xi(t_{2}))}\mid\xi(t_{2}-1)]]\nonumber \\
 & \,\,=\mathbb{E}\Big[\phi^{h(\xi(t_{1}))}\Big(\phi\mathbb{P}[h(\xi(t_{2}))=1\mid\xi(t_{2}-1)]+\mathbb{P}[h(\xi(t_{2}))=0\mid\xi(t_{2}-1)]\Big)\Big]\nonumber \\
 & \,\,=\mathbb{E}\Big[\phi^{h(\xi(t_{1}))}\Big(\phi\mathbb{P}[h(\xi(t_{2}))=1\mid\xi(t_{2}-1)]+1-\mathbb{P}[h(\xi(t_{2}))=1\mid\xi(t_{2}-1)]\Big)\Big]\nonumber \\
 & \,\,=\mathbb{E}\Big[\phi^{h(\xi(t_{1}))}\Big((\phi-1)\mathbb{P}[h(\xi(t_{2}))=1\mid\xi(t_{2}-1)]+1\Big)\Big].\label{eq:cases2part2}
\end{align}
Here we obtain $\mathbb{P}[h(\xi(t_{2}))=1\mid\xi(t_{2}-1)]=\mathbb{P}[\xi(t_{2})\in\Xi_{1}|\xi(t_{2}-1)]=\sum_{r\in\Xi_{1}}\mathbb{P}[\xi(t_{2})=r|\xi(t_{2}-1)]$.
Thus, by using \eqref{eq:transitioncondition} and \eqref{eq:resultfors1},
we arrive at 
\begin{align}
\mathbb{E}[\phi^{\sum_{j=1}^{s}h(\xi(t_{j}))}]\Big)\Big] & \leq\mathbb{E}\Big[\phi^{h(\xi(t_{1}))}\Big((\phi-1)\tilde{p}+1\Big)\Big]=\mathbb{E}[\phi^{h(\xi(t_{1}))}]((\phi-1)\tilde{p}+1)\nonumber \\
 & \leq\phi((\phi-1)\tilde{p}+1).\label{eq:resultfors2}
\end{align}
Hence, we have that \eqref{eq:lemmaresultfors} is satisfied for $s\in\{1,2\}$. 

Now, assume that \eqref{eq:lemmaresultfors} holds for $s=\tilde{s}>2$,
i.e., 
\begin{align}
\mathbb{E}[\phi^{\sum_{j=1}^{\tilde{s}}h(\xi(t_{j}))}] & \leq\phi\left((\phi-1)\tilde{p}+1\right)^{\tilde{s}-1}.\label{eq:resulttildes}
\end{align}
We will to prove that \eqref{eq:lemmaresultfors} holds for $s=\tilde{s}+1$.
By employing arguments similar to those used for obtaining \eqref{eq:cases2part1}--\eqref{eq:resultfors2},
we get 
\begin{align}
 & \mathbb{E}[\phi^{\sum_{j=1}^{\tilde{s}+1}h(\xi(t_{j}))}]=\mathbb{E}[\phi^{\sum_{j=1}^{\tilde{s}}h(\xi(t_{j}))}\phi^{h(\xi(t_{\tilde{s}+1}))}]\nonumber \\
 & \,\,=\mathbb{E}[\mathbb{E}[\phi^{\sum_{j=1}^{\tilde{s}}h(\xi(t_{j}))}\phi^{h(\xi(t_{\tilde{s}+1}))}\mid\mathcal{F}_{t_{\tilde{s}+1}-1}]]\nonumber \\
 & \,\,=\mathbb{E}[\phi^{\sum_{j=1}^{\tilde{s}}h(\xi(t_{j}))}\mathbb{E}[\phi^{h(\xi(t_{\tilde{s}+1}))}\mid\mathcal{F}_{t_{\tilde{s}+1}-1}]]\nonumber \\
 & \,\,=\mathbb{E}[\phi^{\sum_{j=1}^{\tilde{s}}h(\xi(t_{j}))}\mathbb{E}[\phi^{h(\xi(t_{\tilde{s}+1}))}\mid\xi(t_{\tilde{s}+1}-1)]]\leq\mathbb{E}[\phi^{\sum_{j=1}^{\tilde{s}}h(\xi(t_{j}))}]((\phi-1)\tilde{p}+1).\label{eq:casesfortildesplus1part1}
\end{align}
Finally, by \eqref{eq:resulttildes} and \eqref{eq:casesfortildesplus1part1},
we obtain \eqref{eq:lemmaresultfors} with $s=\tilde{s}+1$. \end{proof}

\vskip 10pt 

Next, by utilizing Lemma~\ref{ExpectationLemma}, we prove Lemma~\ref{KeyMarkovLemma}. 

\begin{proof}[Proof of Lemma~\ref{KeyMarkovLemma}] First, we define
\begin{align*}
\overline{h}(k) & \triangleq[h(\xi(0)),h(\xi(1)),\ldots,h(\xi(k-1))]^{\mathrm{T}},\\
\overline{\chi}(k) & \triangleq[\chi(0),\chi(1),\ldots,\chi(k-1)]^{\mathrm{T}},\quad k\in\mathbb{N}.
\end{align*}
Next, let $F_{s,k}\triangleq\{\overline{\chi}\in\{0,1\}^{k}\colon\overline{\chi}^{\mathrm{T}}\overline{\chi}=s\}$,
$s\in\{0,1,\ldots,k\}$, $k\in\mathbb{N}$. Here, we have $F_{s_{1},k}\cap F_{s_{2},k}=\emptyset$,
$s_{1}\neq s_{2}$, and moreover, $\mathbb{P}[\overline{\chi}(k)\in\cup_{s=0}^{k}F_{s,k}]=1$,
$k\in\mathbb{N}$. By utilizing these definitions we obtain for all
$\rho\in(\tilde{p}\tilde{w},1)$ and $k\in\mathbb{N}$, 
\begin{align}
\mathbb{P}[\sum_{t=0}^{k-1}h(\xi(t))\chi(t)>\rho k] & =\mathbb{P}[\overline{h}^{\mathrm{T}}(k)\overline{\chi}(k)>\rho k]\nonumber \\
 & =\sum_{s=0}^{k}\sum_{\overline{\chi}\in F_{s,k}}\mathbb{P}[\overline{h}^{\mathrm{T}}(k)\overline{\chi}(k)>\rho k\mid\overline{\chi}(k)=\overline{\chi}]\mathbb{P}[\overline{\chi}(k)=\overline{\chi}].\label{eq:xibarchibarequation}
\end{align}
 Since $\xi(\cdot)$ and $\chi(\cdot)$ are mutually-independent,
\begin{align}
\mathbb{P}[\overline{h}^{\mathrm{T}}(k)\overline{\chi}(k)>\rho k\mid\overline{\chi}(k)=\overline{\chi}] & =\mathbb{P}[\overline{h}^{\mathrm{T}}(k)\overline{\chi}>\rho k].\label{eq:conditionalprobabilityresolution}
\end{align}
Therefore, it follows from \eqref{eq:xibarchibarequation} and \eqref{eq:conditionalprobabilityresolution}
that for $k\in\mathbb{N}$, 
\begin{align}
\mathbb{P}[\sum_{t=0}^{k-1}h(\xi(t))\chi(t)>\rho k] & =\sum_{s=0}^{k}\sum_{\overline{\chi}\in F_{s,k}}\mathbb{P}[\overline{h}^{\mathrm{T}}(k)\overline{\chi}>\rho k]\mathbb{P}[\overline{\chi}(k)=\overline{\chi}]\nonumber \\
 & =\sum_{s=0}^{\lfloor\tilde{w}k\rfloor}\sum_{\overline{\chi}\in F_{s,k}}\mathbb{P}[\overline{h}^{\mathrm{T}}(k)\overline{\chi}>\rho k]\mathbb{P}[\overline{\chi}(k)=\overline{\chi}]\nonumber \\
 & \quad+\sum_{s=\lfloor\tilde{w}k\rfloor+1}^{k}\sum_{\overline{\chi}\in F_{s,k}}\mathbb{P}[\overline{h}^{\mathrm{T}}(k)\overline{\chi}>\rho k]\mathbb{P}[\overline{\chi}(k)=\overline{\chi}].\label{eq:new-two-terms}
\end{align}

In what follows, our goal is to find upper bounds for the two summation
terms in \eqref{eq:new-two-terms}. We start with the second term.
Since $\mathbb{P}[\overline{h}^{\mathrm{T}}(k)\overline{\chi}>\rho k]\leq1$,
$k\in\mathbb{N}$, we obtain
\begin{align}
 & \sum_{s=\lfloor\tilde{w}k\rfloor+1}^{k}\sum_{\overline{\chi}\in F_{s,k}}\mathbb{P}[\overline{h}^{\mathrm{T}}(k)\overline{\chi}>\rho k]\mathbb{P}[\overline{\chi}(k)=\overline{\chi}]\nonumber \\
 & \quad\leq\sum_{s=\lfloor\tilde{w}k\rfloor+1}^{k}\sum_{\overline{\chi}\in F_{s,k}}\mathbb{P}[\overline{\chi}(k)=\overline{\chi}]=\mathbb{P}[\sum_{t=0}^{k-1}\chi(t)>\tilde{w}k]=\tilde{\sigma}_{k},\label{eq:finalsigmainequality}
\end{align}
for $k\in\mathbb{N}$. Next, we obtain an upper bound for the first
term in \eqref{eq:new-two-terms}. Observe that $\mathbb{P}[\overline{h}^{\mathrm{T}}(k)\overline{\chi}>\rho k]=0$
for $\overline{\chi}\in F_{0,k}$. It then follows that, for all $k\in\mathbb{N}$
such that $\lfloor\tilde{w}k\rfloor=0$, we have 
\begin{align}
\sum_{s=0}^{\lfloor\tilde{w}k\rfloor}\sum_{\overline{\chi}\in F_{s,k}}\mathbb{P}[\overline{h}^{\mathrm{T}}(k)\overline{\chi}>\rho k]\mathbb{P}[\overline{\chi}(k)=\overline{\chi}] & =0.
\end{align}
Moreover, for all $k\in\mathbb{N}$ such that $\lfloor\tilde{w}k\rfloor\geq1$,
we obtain
\begin{align}
 & \sum_{s=0}^{\lfloor\tilde{w}k\rfloor}\sum_{\overline{\chi}\in F_{s,k}}\mathbb{P}[\overline{h}^{\mathrm{T}}(k)\overline{\chi}>\rho k]\mathbb{P}[\overline{\chi}(k)=\overline{\chi}]=\sum_{s=1}^{\lfloor\tilde{w}k\rfloor}\sum_{\overline{\chi}\in F_{s,k}}\mathbb{P}[\overline{h}^{\mathrm{T}}(k)\overline{\chi}>\rho k]\mathbb{P}[\overline{\chi}(k)=\overline{\chi}].\label{eq:sumoverchibar}
\end{align}
To obtain an upper bound for the term $\mathbb{P}[\overline{h}^{\mathrm{T}}(k)\overline{\chi}>\rho k]$,
we will utilize Markov's inequality and Lemma~\ref{ExpectationLemma}.
First, for $s\in\{1,2,\ldots,\lfloor\tilde{w}k\rfloor\}$, let $t_{1}(\overline{\chi}),t_{2}(\overline{\chi}),\ldots,t_{s}(\overline{\chi})$
denote the indices of the nonzero entries of $\overline{\chi}\in F_{s,k}$
such that $t_{1}(\overline{\chi})<t_{2}(\overline{\chi})<\cdots<t_{s}(\overline{\chi})$.
Consequently, 
\begin{align}
 & \mathbb{P}[\overline{h}^{\mathrm{T}}(k)\overline{\chi}>\rho k]=\mathbb{P}[\sum_{j=1}^{s}\overline{h}_{t_{j}(\bar{\chi})}(k)>\rho k]=\mathbb{P}[\sum_{j=1}^{s}h(\xi(t_{j}(\overline{\chi})-1))>\rho k],\label{eq:xibarequation}
\end{align}
for $\overline{\chi}\in F_{s,k}$, $s\in\{1,2,\ldots,\lfloor\tilde{w}k\rfloor\}$,
and $k\in\mathbb{N}$ such that $\lfloor\tilde{w}k\rfloor\geq1$. 

Now observe that $\phi>1$, since $\rho\in(\tilde{p}\tilde{w},\tilde{w})$.
By using Markov's inequality, we obtain 
\begin{align}
\mathbb{P}[\overline{h}^{\mathrm{T}}(k)\overline{\chi}>\rho k] & \leq\mathbb{P}[\sum_{j=1}^{s}h(\xi(t_{j}(\overline{\chi})-1))\geq\rho k]=\mathbb{P}[\phi^{\sum_{j=1}^{s}h(\xi(t_{j}(\overline{\chi})-1))}\geq\phi^{\rho k}]\nonumber \\
 & \leq\phi^{-\rho k}\mathbb{E}[\phi^{\sum_{j=1}^{s}h(\xi(t_{j}(\overline{\chi})-1))}].\label{eq:xiphiineq}
\end{align}
It follows from Lemma~\ref{ExpectationLemma} that $\mathbb{E}[\phi^{\sum_{j=1}^{s}h(\xi(t_{j}(\overline{\chi})-1))}]\leq\phi\left((\phi-1)\tilde{p}+1\right)^{s-1}$.
This inequality together with \eqref{eq:sumoverchibar} and \eqref{eq:xiphiineq}
imply that for all $k\in\mathbb{N}$ such that $\lfloor\tilde{w}k\rfloor\geq1$,
\begin{align}
 & \sum_{s=0}^{\lfloor\tilde{w}k\rfloor}\sum_{\overline{\chi}\in F_{s,k}}\mathbb{P}[\overline{h}^{\mathrm{T}}(k)\overline{\chi}>\rho k]\mathbb{P}[\overline{\chi}(k)=\overline{\chi}]\nonumber \\
 & \quad\leq\sum_{s=1}^{\lfloor\tilde{w}k\rfloor}\sum_{\overline{\chi}\in F_{s,k}}\phi^{-\rho k}\phi\left((\phi-1)\tilde{p}+1\right)^{s-1}\mathbb{P}[\overline{\chi}_{k}=\overline{\chi}]\nonumber \\
 & \quad=\phi^{-\rho k+1}\sum_{s=1}^{\lfloor\tilde{w}k\rfloor}\left((\phi-1)\tilde{p}+1\right)^{s-1}\sum_{\overline{\chi}\in F_{s,k}}\mathbb{P}[\overline{\chi}_{k}=\overline{\chi}]\nonumber \\
 & \quad=\phi^{-\rho k+1}\sum_{s=1}^{\lfloor\tilde{w}k\rfloor}\left((\phi-1)\tilde{p}+1\right)^{s-1}\mathbb{P}[\overline{\chi}_{k}\in F_{s,k}]\nonumber \\
 & \quad\leq\phi^{-\rho k+1}\sum_{s=1}^{\lfloor\tilde{w}k\rfloor}\left((\phi-1)\tilde{p}+1\right)^{s-1}.\label{eq:finaxichiinequality}
\end{align}
Notice that to obtain the last inequality in \eqref{eq:finaxichiinequality},
we employed the fact that $\mathbb{P}[\overline{\chi}_{k}\in F_{s,k}]\leq1$.
The summation term on the far right-hand side of \eqref{eq:finaxichiinequality}
satisfies 
\begin{align}
 & \sum_{s=1}^{\lfloor\tilde{w}k\rfloor}\left((\phi-1)\tilde{p}+1\right)^{s-1}=\frac{\left((\phi-1)\tilde{p}+1\right)^{\lfloor\tilde{w}k\rfloor}-1}{\left((\phi-1)\tilde{p}+1\right)-1}\leq\frac{\left((\phi-1)\tilde{p}+1\right)^{\tilde{w}k}-1}{(\phi-1)\tilde{p}}.\label{eq:geometricseriesfinitesum}
\end{align}
Therefore, from \eqref{eq:finaxichiinequality} and \eqref{eq:geometricseriesfinitesum},
we obtain 
\begin{align}
\sum_{s=0}^{\lfloor\tilde{w}k\rfloor}\sum_{\overline{\chi}\in F_{s,k}}\mathbb{P}[\overline{h}^{\mathrm{T}}(k)\overline{\chi}>\rho k]\mathbb{P}[\overline{\chi}(k)=\overline{\chi}] & \leq\phi^{-\rho k+1}\frac{\left((\phi-1)\tilde{p}+1\right)^{\tilde{w}k}-1}{(\phi-1)\tilde{p}},\label{eq:finalphiwinequality}
\end{align}
for all $k\in\mathbb{N}$ such that $\lfloor\tilde{w}k\rfloor\geq1$.
The right-hand side of this inequality is zero if $\lfloor\tilde{w}k\rfloor=0$.
Therefore, \eqref{eq:finalphiwinequality} holds also for all $k\in\mathbb{N}$.
By using this fact together with \eqref{eq:new-two-terms}, \eqref{eq:finalsigmainequality},
we arrive at \eqref{eq:keylemmaresult1}.

Next, we show $\sum_{k=1}^{\infty}\psi_{k}<\infty$. First of all,
we have 
\begin{align}
 & \sum_{k=1}^{\infty}\phi^{-\rho k+1}\frac{\left((\phi-1)\tilde{p}+1\right)^{\tilde{w}k}-1}{(\phi-1)\tilde{p}}\nonumber \\
 & \quad=\frac{\phi}{(\phi-1)\tilde{p}}\sum_{k=1}^{\infty}\phi^{-\rho k}\left((\phi-1)\tilde{p}+1\right)^{\tilde{w}k}-\frac{\phi}{(\phi-1)\tilde{p}}\sum_{k=1}^{\infty}\phi^{-\rho k}.\label{eq:psisum}
\end{align}
We will prove that the summation terms on the right-hand side of \eqref{eq:psisum}
are both convergent. First, we have $\phi^{-\rho}<1$, because $\phi>1$.
Therefore, the geometric series $\sum_{k=1}^{\infty}\phi^{-\rho k}$
converges, that is, $\sum_{k=1}^{\infty}\phi^{-\rho k}<\infty$. Next,
we show $\phi^{-\rho}\left((\phi-1)\tilde{p}+1\right)^{\tilde{w}}<1$.
Observe that 
\begin{align}
\phi^{-\rho}\left((\phi-1)\tilde{p}+1\right)^{\tilde{w}} & =\left(\phi^{-\frac{\rho}{\tilde{w}}}\left((\phi-1)\tilde{p}+1\right)\right)^{\tilde{w}}.\label{eq:wpower}
\end{align}
Moreover, 
\begin{align*}
\phi^{-\frac{\rho}{\tilde{w}}}\left((\phi-1)\tilde{p}+1\right) & =\left(\frac{\frac{\rho}{\tilde{w}}(1-\tilde{p})}{\tilde{p}(1-\frac{\rho}{\tilde{w}})}\right)^{-\frac{\rho}{\tilde{w}}}\left(\left(\frac{\frac{\rho}{\tilde{w}}(1-\tilde{p})}{\tilde{p}(1-\frac{\rho}{\tilde{w}})}-1\right)\tilde{p}+1\right)\\
 & =\left(\frac{\tilde{p}\tilde{w}}{\rho}\right)^{\frac{\rho}{\tilde{w}}}\left(\frac{1-\tilde{p}}{1-\frac{\rho}{\tilde{w}}}\right)^{1-\frac{\rho}{\tilde{w}}}.
\end{align*}
 Here, we have $\frac{\tilde{p}\tilde{w}}{\rho},\frac{1-\tilde{p}}{1-\frac{\rho}{\tilde{w}}}\in(0,1)\cup(1,\infty)$.
As a result, since $\ln v<v-1$ for any $v\in(0,1)\cup(1,\infty)$,
we obtain
\begin{align*}
 & \ln\left(\phi^{-\frac{\rho}{\tilde{w}}}\left((\phi-1)\tilde{p}+1\right)\right)=\frac{\rho}{\tilde{w}}\ln\left(\frac{\tilde{p}\tilde{w}}{\rho}\right)+(1-\frac{\rho}{\tilde{w}})\ln\left(\frac{1-\tilde{p}}{1-\frac{\rho}{\tilde{w}}}\right)\\
 & \quad<\frac{\rho}{\tilde{w}}\left(\frac{\tilde{p}\tilde{w}}{\rho}-1\right)+(1-\frac{\rho}{\tilde{w}})\left(\frac{1-\tilde{p}}{1-\frac{\rho}{\tilde{w}}}-1\right)=\tilde{p}-\frac{\rho}{\tilde{w}}+\frac{p}{\tilde{w}}-\tilde{p}\,=\,0,
\end{align*}
 which implies $\phi^{-\frac{\rho}{\tilde{w}}}\left((\phi-1)\tilde{p}+1\right)<1$.
Thus, by \eqref{eq:wpower}, $\phi^{-\rho}\left((\phi-1)\tilde{p}+1\right)^{\tilde{w}}<1$.
It then follows that 
\begin{align}
\sum_{k=1}^{\infty}\phi^{-\rho k}\left((\phi-1)\tilde{p}+1\right)^{\tilde{w}k}<\infty.\label{eq:difficultcasesum}
\end{align}
Finally, since $\sum_{k=1}^{\infty}\phi^{-\rho k}<\infty$, we obtain
$\sum_{k=1}^{\infty}\psi_{k}<\infty$ from \eqref{eq:psisum} and
\eqref{eq:difficultcasesum}. \end{proof}

\bibliographystyle{ieeetr}
\bibliography{references}

\end{document}